\documentclass{SciPost}

\usepackage[bitstream-charter]{mathdesign}
\DeclareSymbolFont{usualmathcal}{OMS}{cmsy}{m}{n}
\DeclareSymbolFontAlphabet{\mathcal}{usualmathcal}

\fancypagestyle{SPstyle}{
  \fancyhf{}
  \lhead{\colorbox{scipostblue}{\bfseries\color{white}~SciPost Physics~}}
  \rhead{{\bfseries\color{scipostdeepblue}~Submission}}
  
  \fancyfoot[C]{\textbf{\thepage}}
}

\usepackage[T1]{fontenc}
\usepackage{mathtools}
\usepackage{physics}
\usepackage{tikz}
\usetikzlibrary{arrows.meta}
\usepackage{pgfplots}
\pgfplotsset{compat=1.18}

\tikzset{
  manuscript flow diagram/.style={
    x=1cm,
    y=1cm,
    font=\footnotesize,
    flow/.style={
      -{Latex[length=2.1mm,width=1.4mm]},
      line width=0.7pt,
      draw=black!75
    },
    box/.style={
      draw=black!65,
      rounded corners=2pt,
      line width=0.6pt,
      align=center,
      inner xsep=5pt,
      inner ysep=5pt,
      font=\scriptsize
    }
  }
}

\hypersetup{
  breaklinks = true,
  colorlinks = true,
  allcolors = blue,
  pdftitle = {Two-dimensional Toda--Arnoldi correspondence: holomorphic Krylov geometry and counterdiabatic transport},
  pdfauthor = {Urei Miura},
  pdfkeywords = {Toda lattice, Arnoldi method, Krylov subspaces, non-Hermitian Hamiltonians, counterdiabatic driving}
}

\newenvironment{thm}
  {\par\medskip\noindent\textbf{Main theorem.}\ \itshape}
  {\par\medskip}
\newenvironment{proof}
  {\par\medskip\noindent\textbf{Proof.}\ }
  {\hfill\ensuremath{\square}\par\medskip}

\newcommand{\C}{\mathbb{C}}
\newcommand{\R}{\mathbb{R}}
\newcommand{\Z}{\mathbb{Z}}
\newcommand{\CP}{\mathbb{CP}}
\newcommand{\K}{\mathcal{K}}
\DeclareMathOperator{\Span}{span}
\DeclareMathOperator{\Gr}{Gr}

\newcommand{\zb}{\bar{z}}
\newcommand{\ii}{\mathrm{i}}
\newcommand{\e}{\mathrm{e}}
\newcommand{\eno}[1]{\e^{#1}}

\newcommand{\LLop}{\mathsf{L}}
\newcommand{\MM}{\mathsf{M}}

\newcommand{\HH}{\mathsf{H}}
\newcommand{\Had}{\HH_{\mathrm{ad}}}
\newcommand{\Hcd}{\HH_{\mathrm{cd}}}
\newcommand{\Uarn}{\mathsf{U}}
\newcommand{\Aarn}{\mathsf{A}}

\newcommand{\Imat}{\mathsf{I}}
\newcommand{\Gmat}{\mathsf{G}}
\newcommand{\Vmat}{\mathsf{V}}
\newcommand{\Pmat}{\mathsf{P}}

\newcommand{\wket}[1]{%
  \left.\left|#1\right\rangle\!\right\rangle
}
\newcommand{\kket}[1]{\wket{#1}}
\newcommand{\wbraket}[2]{%
  \left\langle\!\left\langle
  #1
  \,\middle|\,
  #2
  \right\rangle\!\right\rangle
}

\begin{document}

\pagestyle{SPstyle}

\begin{center}
{\Large\textbf{\color{scipostdeepblue}{Two-dimensional Toda--Arnoldi correspondence: Holomorphic Krylov geometry and counterdiabatic transport}}}
\end{center}

\begin{center}
Urei Miura\textsuperscript{1,2,$\star$}
\end{center}

\begin{center}
{\bfseries 1} Division of Physics and Astronomy, Graduate School of Science, Kyoto University, Kyoto 606-8502, Japan\\
{\bfseries 2} Center for Gravitational Physics and Quantum Information, Yukawa Institute for Theoretical Physics, Kyoto University, Kitashirakawa Oiwake-Cho, Kyoto 606-8502, Japan\\
${}^\star$ \href{mailto:urei.miura@yukawa.kyoto-u.ac.jp}{\small urei.miura@yukawa.kyoto-u.ac.jp}
\end{center}

\section*{\color{scipostdeepblue}{Abstract}}
{\bfseries\boldmath
Although Arnoldi reduction of a generally non-Hermitian Hamiltonian yields an upper Hessenberg matrix rather than the tridiagonal form of Hermitian Lanczos theory, we show that a closed Toda sector survives in its diagonal and subdiagonal coefficients.
For a fixed finite-dimensional Hamiltonian and a cyclic state vector deformed holomorphically, the Krylov Gram determinants are $\tau$ functions of the finite two-dimensional Toda lattice, whose Flaschka variables coincide exactly with these Arnoldi coefficients.
The Toda dynamics therefore closes on this sector without determining the remaining upper Hessenberg entries.
The subdiagonal part of the same sector also has a direct geometric meaning: the squared subdiagonal coefficients determine both the Fubini--Study metric and the Berry curvature of holomorphic Krylov subspaces, whereas the geometric quantities associated with subspaces lost at Arnoldi breakdown cease to be defined.
Along a smooth real path in the cyclic region, the Arnoldi-frame connection further provides a Hermitian tridiagonal generator of exact isospectral transport.
When added to the Arnoldi matrix, this generator cancels transitions between instantaneous eigenspaces and realizes counterdiabatic driving whenever the matrix is diagonalizable with a nondegenerate spectrum.
}

\vspace{\baselineskip}

\noindent\textcolor{white!90!black}{%
\fbox{\parbox{0.975\linewidth}{%
\textcolor{white!40!black}{\begin{tabular}{lr}%
  \begin{minipage}{0.6\textwidth}%
    {\small Copyright attribution to authors. \newline
    This work is a submission to SciPost Physics. \newline
    License information to appear upon publication. \newline
    Publication information to appear upon publication.}
  \end{minipage} & \begin{minipage}{0.4\textwidth}
    {\small Received Date \newline Accepted Date \newline Published Date}%
  \end{minipage}
\end{tabular}}
}}
}

\vspace{10pt}
\noindent\rule{\textwidth}{1pt}
\tableofcontents\thispagestyle{SPstyle}
\noindent\rule{\textwidth}{1pt}
\vspace{10pt}

\nolinenumbers


\section{Introduction}
\label{sec:introduction}

Understanding how quantum states and observables evolve is a central problem
in quantum physics.
Krylov methods provide both a computational and a conceptual framework for
this problem.
They organize the evolution in the subspace generated by repeated action of
the dynamical generator, allowing matrix functions and quantum time evolution
to be approximated efficiently without diagonalizing the full problem
~\cite{HochbruckLubich1997,NandyEtAl2025Review}.
The connection is direct: the generator of the physical evolution also
generates the Krylov basis used to represent it.

Beyond this numerical role, Krylov bases have been used to study operator
growth, thermalization, and quantum complexity
~\cite{ParkerEtAl2019,RabinoviciEtAl2021,NandyEtAl2025Review,BhattacharjeeEtAl2022Saddle,BhattacharjeeEtAl2023LargeQSYK}.
More recently, they have been used to construct counterdiabatic generators
and to describe dynamics generated by time-dependent operators
~\cite{TakahashiDelCampo2024,TakahashiDelCampo2025}.

For a Hermitian generator, the Lanczos construction yields a three-term
recurrence and a real symmetric tridiagonal representation
~\cite{Lanczos1950}.
Under exponential deformations of correlation functions, the associated
moment and Gram determinants satisfy the Toda $\tau$-function relations, and the
recurrence coefficients serve as Toda variables~\cite{DymarskyGorsky2020}.

For a generally non-Hermitian generator, the Arnoldi method uses a single
orthonormal basis and produces an upper Hessenberg matrix
~\cite{Arnoldi1951,Saad1980}.
The bi-Lanczos method instead retains a tridiagonal representation by using
dual bases, and both methods have been applied to open-system dynamics and
operator growth
~\cite{MingantiHuybrechts2022,BhattacharyaEtAl2022Arnoldi,BhattacharyaEtAl2023BiLanczos,BhattacharjeeEtAl2023DissipativeSYK,BhattacharjeeEtAl2024LindbladianSYK}.
A general Arnoldi matrix, however, contains entries beyond its diagonal and
subdiagonal.
We ask whether the diagonal and subdiagonal Arnoldi coefficients form a
closed Toda system and whether they determine the geometry of the
corresponding Krylov subspaces.

We address this question for a fixed, finite-dimensional, generally
non-Hermitian Hamiltonian and a cyclic seed state.
We deform the seed holomorphically and construct the Arnoldi basis at each
parameter value.
We prove that the resulting Krylov Gram determinants satisfy the $\tau$-function
relations of the finite two-dimensional Toda lattice.
This establishes a Toda--Arnoldi correspondence in which the Flaschka
variables are the diagonal and subdiagonal Arnoldi coefficients.
The Toda equations close on these coefficients without fixing the remaining
upper Hessenberg entries.

We then use the squared subdiagonal coefficients to determine the
Fubini--Study metric and Berry curvature induced by the holomorphic
Krylov-subspace map.
Along a smooth real parameter path, the moving Arnoldi frame generates an
exact isospectral evolution of the Arnoldi matrix.
When this matrix is diagonalizable with a nondegenerate spectrum, the same
generator cancels transitions between instantaneous eigenspaces and acts as
a counterdiabatic term.

The paper is organized as follows.
Section~\ref{sec:main-theorem} establishes the Toda--Arnoldi correspondence,
Sec.~\ref{sec:krylov-geometry} develops its geometric consequences, and
Sec.~\ref{STA-in-Krylov-chain} treats transport along real parameter paths.
Appendices~\ref{app:toda-review} and~\ref{app:arnoldi} review the
two-dimensional Toda lattice and the Arnoldi method, respectively, and
Appendix~\ref{app:sta-basics} reviews counterdiabatic driving.
Appendix~\ref{app:main-theorem-qr-proof} gives the QR proof and Lax equation,
Appendix~\ref{app:grassmannian-geometry} collects the geometric details, and
Appendix~\ref{app:krylov-electrostatics} presents an electrostatic form of the
Toda--Arnoldi relations.

\section{Two-dimensional Toda--Arnoldi correspondence}
\label{sec:main-theorem}
For a generally non-Hermitian finite-dimensional Hamiltonian, we deform a state vector holomorphically with respect to a complex parameter and apply the Arnoldi method at each parameter value.
We show that the resulting $\tau$ functions are the squared volumes spanned by the corresponding unnormalized holomorphic Krylov vectors.
We then prove that the diagonal and subdiagonal Arnoldi coefficients $\alpha_{n}\qty(z,\zb)\in\C$ and $\beta_{n}\qty(z,\zb)>0$ coincide with the Flaschka variables $a_{n}\qty(z,\zb),b_{n}\qty(z,\zb)$ of the two-dimensional Toda lattice.

For the two-dimensional Toda lattice, see Appendix~\ref{app:toda-review} and Refs.~\cite{UenoTakasaki1984,Takasaki2018Toda}; for the Arnoldi method, see Appendix~\ref{app:arnoldi} and Refs.~\cite{Arnoldi1951,Saad1980}.
Figure~\ref{fig:toda-arnoldi-overview} summarizes the construction that connects holomorphic evolution to Arnoldi coefficients and two-dimensional Toda variables through Krylov data and $\tau$ functions.

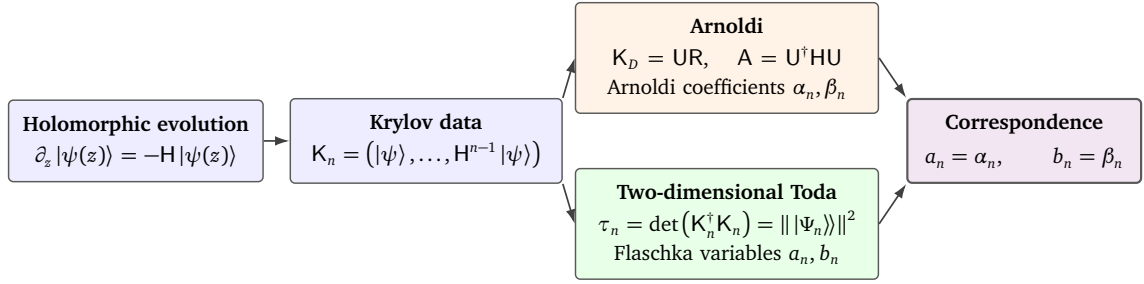
\begin{figure}[t]
  \centering
  \resizebox{\linewidth}{!}{%
  \begin{tikzpicture}[
    manuscript flow diagram,
    box/.append style={
      font=\footnotesize,
      inner xsep=4pt,
      inner ysep=6pt
    }
  ]
    \node[
      box,
      fill=blue!7,
      text width=3.40cm
    ] (state) at (-7.00,0) {
      {\footnotesize\textbf{Holomorphic evolution}}\\[2pt]
      $\partial_{z}\ket{\psi\qty(z)}=-\HH\ket{\psi\qty(z)}$
    };

    \node[
      box,
      fill=blue!7,
      text width=3.65cm
    ] (krylov) at (-2.80,0) {
      {\footnotesize\textbf{Krylov data}}\\[2pt]
      $\mathsf{K}_{n}
      =\qty(\ket{\psi},\ldots,\HH^{n-1}\ket{\psi})$
    };

    \node[
      box,
      fill=orange!10,
      text width=4.10cm
    ] (arnoldi) at (1.55,1.20) {
      {\footnotesize\textbf{Arnoldi}}\\[2pt]
      $\mathsf{K}_{D}=\Uarn\mathsf{R}$,
      \quad
      $\Aarn=\Uarn^{\dagger}\HH\Uarn$\\[1pt]
      Arnoldi coefficients $\alpha_{n},\beta_{n}$
    };

    \node[
      box,
      fill=green!9,
      text width=4.10cm
    ] (toda) at (1.55,-1.20) {
      {\footnotesize\textbf{Two-dimensional Toda}}\\[2pt]
      $\tau_{n}
      =\det\qty(\mathsf{K}_{n}^{\dagger}\mathsf{K}_{n})
      =\norm{\wket{\Psi_{n}}}^{2}$\\[1pt]
      Flaschka variables $a_{n},b_{n}$
    };

    \node[
      box,
      fill=violet!10,
      text width=3.10cm,
      line width=0.9pt
    ] (result) at (5.85,0) {
      {\footnotesize\textbf{Correspondence}}\\[3pt]
      $a_{n}=\alpha_{n}$,
      \qquad
      $b_{n}=\beta_{n}$
    };

    \draw[flow] (state.east) -- (krylov.west);
    \draw[flow] (krylov.north east) -- (arnoldi.west);
    \draw[flow] (krylov.south east) -- (toda.west);
    \draw[flow] (arnoldi.east) -- (result.north west);
    \draw[flow] (toda.east) -- (result.south west);
  \end{tikzpicture}%
  }
  \caption{
    Toda--Arnoldi correspondence obtained from holomorphic evolution through Krylov data.
    The Arnoldi coefficients $\alpha_{n},\beta_{n}$ coincide, respectively, with the two-dimensional Toda variables $a_{n},b_{n}$ wherever the Arnoldi method does not break down.
  }
  \label{fig:toda-arnoldi-overview}
\end{figure}

\subsection{\texorpdfstring{$\tau$}{tau} functions from holomorphic Krylov data}
\label{subsec:holomorphic-krylov-deformation}
Let $z\in\C$ denote a complex deformation parameter,
and let $\zb\coloneqq z^{*}$ be its complex conjugate.
Here, $z$ is a deformation parameter rather than physical time; for Hermitian
$\HH$, its real direction is imaginary-time-like.
For $D\in\Z_{>0}$, let $\mathcal{H}_{D}\simeq\C^{D}$ be a $D$-dimensional Hilbert space and $\HH\in\operatorname{End}\qty(\mathcal{H}_{D})$ a generally non-Hermitian Hamiltonian.
Assume that a nonzero holomorphic state vector $\ket{\psi\qty(z)}\in\mathcal{H}_{D}$ satisfies
\begin{align}
  \partial_{z}\ket{\psi\qty(z)}
  &=
  -\HH\ket{\psi\qty(z)},
  \label{eq:EqMstate}\\
  \partial_{\zb}\ket{\psi\qty(z)}
  &=
  0.
  \label{eq:holomorphic-state-condition}
\end{align}
We later use this vector as the seed vector for the Arnoldi method.
These differential equations have the solution
\begin{align}
  \ket{\psi\qty(z)} = \eno{-z\HH}\ket{\psi\qty(0)}.
  \label{eq:holomorphic-state-solution}
\end{align}
In what follows, we assume that $\ket{\psi\qty(0)}$ is a cyclic vector of $\HH$
[see Eq.~\eqref{eq:arnoldi-cyclic-vector-condition}].
Because $\eno{-z\HH}$ is invertible and commutes with $\HH$,
$\ket{\psi\qty(z)}$ is also a cyclic vector of $\HH$ for every $z\in\C$.
We call the set of parameter values for which the seed vector remains cyclic
the cyclic region; equivalently, it is the region in which all $D$ Arnoldi
vectors are well defined.

Define the $\tau_{1}\qty(z,\zb)$ function of the two-dimensional Toda lattice as the squared norm of the holomorphic state vector $\ket{\psi\qty(z)}$:
\begin{align}
  \tau_{1}\qty(z,\zb)\coloneqq \braket{\psi\qty(\zb)}{\psi\qty(z)}\quad\qty(>0).
  \label{eq:tau-one-norm}
\end{align}
In this inner product, $\bra{\psi\qty(\zb)}$ is antiholomorphic and $\ket{\psi\qty(z)}$ is holomorphic.
Equations~\eqref{eq:EqMstate} and \eqref{eq:holomorphic-state-condition} give the derivatives of $\tau_{1}\qty(z,\zb)$ as
\begin{align}
  \partial_{z}^{i}\partial_{\zb}^{j}\tau_{1}\qty(z,\zb) = \qty(-1)^{i+j}\bra{\psi\qty(\zb)}\qty(\HH^{\dagger})^{j}\HH^{i}\ket{\psi\qty(z)}.
  \label{eq:tau-one-derivatives}
\end{align}
The sign factors arising in the rows and columns cancel in the determinant.
Set $\tau_{0}\coloneqq1$ and define $\tau_{n}\qty(z,\zb)$ for $1\leq n\leq D$ by
\begin{align}
\label{eq:taun}
  \tau_{n}\qty(z,\zb)
  &=
  \det\qty[
    \bra{\psi\qty(\zb)}
    \qty(\HH^{\dagger})^{j}
    \HH^{i}
    \ket{\psi\qty(z)}
  ]_{0 \leq i,j \leq n-1}.
\end{align}
Because $\mathcal{H}_{D}$ is $D$-dimensional and $\ket{\psi\qty(z)}$ is a cyclic vector,
$\tau_{n}>0\quad\qty(1\leq n\leq D)$.
For the finite-chain boundary condition, set $\tau_{D+1}\coloneqq0$.
By the Desnanot--Jacobi identity~\eqref{eq:desnanot-jacobi}, Eq.~\eqref{eq:taun} satisfies the Hirota bilinear identity~\eqref{eq:2d-toda-bilinear} of the two-dimensional Toda lattice for $1\leq n\leq D$.
The Gram determinant formula~\eqref{eq:grassmann-wedge-norm} allows Eq.~\eqref{eq:taun} to be written as the squared norm of a holomorphic Krylov exterior state.
This state is constructed from the holomorphic Krylov basis
\begin{align}
    \ket{\psi\qty(z)},\HH
    \ket{\psi\qty(z)},\HH^{2}
    \ket{\psi\qty(z)},\cdots,\HH^{n-1}
    \ket{\psi\qty(z)},
  \label{eq:holomorphic-krylov-basis}
\end{align}
as
\begin{align}
  \kket{\Psi_{n}\qty(z)}\coloneqq
    \ket{\psi\qty(z)}\wedge\HH
    \ket{\psi\qty(z)}\wedge\cdots\wedge\HH^{n-1}\ket{\psi\qty(z)},
  \label{eq:holomorphic-exterior-state}
\end{align}
and its squared norm is
\begin{align}
\label{eq:taun-innerproduct}
  \tau_{n}\qty(z,\zb)
  &=
  \wbraket{\Psi_{n}\qty(\zb)}{\Psi_{n}\qty(z)}
  \quad\qty(>0).
\end{align}
Equation~\eqref{eq:taun-innerproduct} shows that $\tau_{n}\qty(z,\zb)$ can be interpreted as the squared volume spanned by the first $n$ unnormalized holomorphic Krylov vectors in the Hilbert space.
The same exterior-product representation gives the identity
\begin{align}
  \partial_{z}^{i}\partial_{\zb}^{j}\tau_{n}\qty(z,\zb)
  &=
  \wbraket{\partial_{\zb}^{j}\Psi_{n}\qty(\zb)}{\partial_{z}^{i}\Psi_{n}\qty(z)}.
  \label{eq:taun-exterior-derivatives}
\end{align}

\subsection{Arnoldi coefficients as Flaschka variables}
\label{subsec:Main-theorem}
The main theorem is as follows.

\begin{thm}
At each value of $z$, the normalized state vector
\begin{align}
  \ket{u_{0}\qty(z,\zb)}=\frac{\ket{\psi\qty(z)}}{\norm{\ket{\psi\qty(z)}}}
  =\frac{\ket{\psi\qty(z)}}{\tau_{1}\qty(z,\zb)^{\frac{1}{2}}}
  \label{eq:toda-arnoldi-seed}
\end{align}
serves as the seed vector in the Arnoldi reduction of $\HH$ to upper Hessenberg form.
Denote the resulting upper Hessenberg matrix by $\Aarn\qty(z,\zb)$, its diagonal entries by $\alpha_{n}\qty(z,\zb)\quad\qty(0\leq n\leq D-1)$, and its subdiagonal entries by $\beta_{n}\qty(z,\zb)\quad\qty(1\leq n\leq D-1)$:
\begin{align}
  \Aarn\qty(z,\zb) =
  \begin{pmatrix}
    \alpha_{0}\qty(z,\zb)
    & *
    & \cdots
    & *
    \\
    \beta_{1}\qty(z,\zb)
    & \alpha_{1}\qty(z,\zb)
    & \ddots
    & \vdots
    \\
    \vdots
    & \ddots
    & \ddots
    & *
    \\
    0
    & \cdots
    & \beta_{D-1}\qty(z,\zb)
    & \alpha_{D-1}\qty(z,\zb)
  \end{pmatrix}.
  \label{eq:toda-arnoldi-hessenberg-form}
\end{align}
The Flaschka variables of the two-dimensional Toda lattice defined from the $\tau$ functions~\eqref{eq:taun} by Eqs.~\eqref{eq:2d-toda-flaschka-a} and \eqref{eq:2d-toda-flaschka-b} coincide with the Arnoldi coefficients:
\begin{equation}
  \begin{aligned}
    a_{n}\qty(z,\zb)
    &=
    \alpha_{n}\qty(z,\zb),
    &&0\leq n\leq D-1,
    \\
    b_{n}\qty(z,\zb)
    &=
    \beta_{n}\qty(z,\zb),
    &&1\leq n\leq D-1.
  \end{aligned}
  \label{eq:toda-arnoldi-coefficients}
\end{equation}
\end{thm}

\begin{proof}
We express the $\tau$ functions and the holomorphic Krylov exterior vector $\kket{\Psi_{n}\qty(z)}$ in terms of the Arnoldi coefficients $\alpha_{n}\qty(z,\zb),\beta_{n}\qty(z,\zb)$.
We then substitute these expressions into Eqs.~\eqref{eq:2d-toda-flaschka-a} and \eqref{eq:2d-toda-flaschka-b}.
For the remainder of the proof, we suppress the dependence on $\qty(z,\zb)$.

We apply the Arnoldi method to $\HH$ using the seed vector $\ket{u_{0}}$ and denote the resulting orthonormal basis by
\begin{align}
  \qty{\ket{u_{l}}}_{l=0}^{D-1}.
  \label{eq:toda-arnoldi-basis}
\end{align}
Applying Eq.~\eqref{eq:arnoldi-recurrence} successively from $n=0$ and using Eq.~\eqref{eq:toda-arnoldi-seed} gives
\begin{align}
  &\HH^{0}\ket{\psi}
  =
  \tau_{1}^{\frac{1}{2}}\ket{u_{0}},
  \notag\\
  &\HH^{1}\ket{\psi}
  =
  \tau_{1}^{\frac{1}{2}}
  \qty(\beta_{1}\ket{u_{1}}+\alpha_{0}\ket{u_{0}}),
  \notag\\
  &\HH^{l}\ket{\psi}
  =
  \tau_{1}^{\frac{1}{2}}
  \begin{aligned}[t]
    &\qty[
      \qty(\prod_{m=1}^{l}\beta_{m})\ket{u_{l}}
      +
      \qty(\prod_{m=1}^{l-1}\beta_{m})
      \qty(\sum_{m=0}^{l-1}\alpha_{m})\ket{u_{l-1}}
    ]
    \\
    &\quad+
    \qty(\text{terms proportional to }\ket{u_{l-2}},\ldots,\ket{u_{0}}),
  \end{aligned}
  \notag\\
  &\qty(2\leq l\leq D-1).
  \label{eq:toda-arnoldi-krylov-expansion}
\end{align}

Using the orthonormality of $\qty{\ket{u_{l}}}_{l=0}^{D-1}$ and substituting Eq.~\eqref{eq:toda-arnoldi-krylov-expansion} into Eqs.~\eqref{eq:taun-innerproduct} and \eqref{eq:taun-exterior-derivatives}, we obtain the following identities:
\begin{align}
  &\tau_{1}
  =
  \braket{\psi}{\psi},
  \qquad
  \qty(n=1),
  \notag\\
  &\tau_{n}
  =
  \tau_{1}^{n}
  \qty(
    \prod_{m=1}^{n-1}
    \beta_{m}^{2\qty(n-m)}
  ),
  \qquad
  \qty(2\leq n\leq D),
  \label{eq:toda-arnoldi-tau-product}\\
  &\partial_{z}\ln\qty(\tau_{1})
  =
  -\alpha_{0},
  \qquad
  \qty(n=1),
  \notag\\
  &\partial_{z}\ln\qty(\tau_{n})
  =
  -\qty(
    \sum_{m=0}^{n-1}\alpha_{m}
  ),
  \qquad
  \qty(2\leq n\leq D).
  \label{eq:toda-arnoldi-tau-derivative}
\end{align}
In deriving Eq.~\eqref{eq:toda-arnoldi-tau-derivative}, we used the following relation, which follows from Eq.~\eqref{eq:EqMstate} and the antisymmetry of the exterior product:
\begin{align}
  &\partial_{z}\kket{\Psi_{1}}
  =
  -\HH\ket{\psi},
  \qquad
  \qty(n=1),
  \notag\\
  &\partial_{z}\kket{\Psi_{n}}
  =
  -\kket{\Psi_{n-1}}
  \wedge\HH^{n}\ket{\psi},
  \qquad
  \qty(2\leq n\leq D).
  \label{eq:toda-arnoldi-exterior-derivative}
\end{align}
For $n=D$ in Eq.~\eqref{eq:toda-arnoldi-tau-derivative}, we used Eq.~\eqref{eq:arnoldi-last-column}.

Substituting Eq.~\eqref{eq:toda-arnoldi-tau-derivative} into Eq.~\eqref{eq:2d-toda-flaschka-a} and Eq.~\eqref{eq:toda-arnoldi-tau-product} into Eq.~\eqref{eq:2d-toda-flaschka-b} gives $a_{n}=\alpha_{n}$ and $b_{n}^{2}=\beta_{n}^{2}$.
Because $b_{n}>0$ and $\beta_{n}>0$, we have $b_{n}=\beta_{n}$, which proves Eq.~\eqref{eq:toda-arnoldi-coefficients}.
\end{proof}

Appendix~\ref{app:main-theorem-qr-proof} gives an alternative proof based on the QR decomposition of the matrix formed by the $D$ Krylov vectors, together with the Lax equation for the Arnoldi matrix derived from the same decomposition.

\subsection{Toda--Lanczos correspondence for Hermitian Hamiltonians}
\label{subsec:comments-main-theorem}

When the Hamiltonian $\HH$ is Hermitian, the Arnoldi method reduces to the Lanczos method, and $\Aarn$ is a Hermitian tridiagonal matrix~\cite{Lanczos1950,Saad1980}.
Equations~\eqref{eq:holomorphic-state-solution} and \eqref{eq:taun} then imply that $\tau_{n}$ depends only on the real deformation variable $t\coloneqq z+\zb$.
Thus, $a_{n}$ and $b_{n}$ are also functions of $t$ alone and are independent of the imaginary part of $z$.
The corresponding two-dimensional Toda lattice therefore reduces to the Toda lattice, and our correspondence reduces to the known Toda--Lanczos correspondence~\cite{DymarskyGorsky2020,TakahashiNandyDelCampo2026}.

\section{Geometry of holomorphic Krylov subspaces}
\label{sec:krylov-geometry}

In this section, let the domain of the complex parameter be a Riemann surface $\Sigma$ with local complex coordinate $z$.
The holomorphic Krylov exterior state defines a map from the region $\Sigma_{n}^{\circ}\subseteq\Sigma$, where the Krylov matrix has rank $n$, to the Grassmannian $\Gr\qty(n,\mathcal{H}_{D})$.
Under the Toda--Arnoldi correspondence, the $z$ component of the Berry connection associated with this map is expressed in terms of the sum of the Flaschka variables $a_{m}$.
The variable $b_{n}^{2}$ determines both the Berry curvature and the Fubini--Study metric.
Stokes' theorem equates the area integral of $b_{n}^{2}$ with the boundary integral of the sum of $a_{m}$.

In particular, if this map extends holomorphically to the entire compact Riemann surface $\Sigma$ without boundary, the integral of the Berry curvature defines a quantized Chern number~\cite{Simon1983}.
For a three-site model, we show that when Arnoldi breakdown lowers the Krylov dimension, this Chern number either changes or, together with the corresponding exterior state, ceases to be defined.
Appendix~\ref{app:grassmannian-geometry} defines projective spaces, Grassmannians, and the Fubini--Study metric.

\subsection{Berry connection, curvature, and Chern numbers}
\label{subsec:krylov-berry-connection-curvature}

The holomorphic Krylov exterior state defined in Eq.~\eqref{eq:holomorphic-exterior-state} is a local representative vector of the corresponding Pl\"ucker image.
On a local coordinate neighborhood where $\wket{\Psi_{n}\qty(z)}\neq0$, normalize it as
\begin{align}
  \wket{\widehat{\Psi}_{n}\qty(z,\zb)}
  &\coloneqq
  \frac{
    \wket{\Psi_{n}\qty(z)}
  }{
    \tau_{n}\qty(z,\zb)^{\frac{1}{2}}
  },
  \qquad
  1\leq n\leq D.
  \label{eq:normalized-krylov-exterior-state}
\end{align}
Following Ref.~\cite{Berry1984}, define its Berry connection $\mathcal{A}_{n}\qty(z,\zb)$ by
\begin{align}
  \mathcal{A}_{n}\qty(z,\zb)
  &\coloneqq
  \ii
  \wbraket{
    \widehat{\Psi}_{n}\qty(z,\zb)
  }{
    \dd\widehat{\Psi}_{n}\qty(z,\zb)
  }
  =
  \mathcal{A}_{n,z}\qty(z,\zb)\dd z
  +
  \mathcal{A}_{n,\zb}\qty(z,\zb)\dd\zb.
  \label{eq:krylov-exterior-berry-connection}
\end{align}
Holomorphicity and Eq.~\eqref{eq:toda-arnoldi-tau-derivative} give the components of the Berry connection in the phase convention of Eq.~\eqref{eq:normalized-krylov-exterior-state} as
\begin{align}
  \mathcal{A}_{n,z}\qty(z,\zb)
  &=
  \frac{\ii}{2}
  \partial_{z}\ln\qty(\tau_{n}\qty(z,\zb))
  =
  -\frac{\ii}{2}
  \sum_{m=0}^{n-1}a_{m}\qty(z,\zb),
  \label{eq:krylov-exterior-berry-connection-z}\\
  \mathcal{A}_{n,\zb}\qty(z,\zb)
  &=
  -\frac{\ii}{2}
  \partial_{\zb}\ln\qty(\tau_{n}\qty(z,\zb)).
  \label{eq:krylov-exterior-berry-connection-zbar}
\end{align}
Replacing the local representative vector according to
$\wket{\widehat{\Psi}_{n}}\mapsto
e^{\ii\chi}\wket{\widehat{\Psi}_{n}}$
transforms the Berry connection as
$\mathcal{A}_{n}\mapsto\mathcal{A}_{n}-\dd\chi$.
The Berry curvature $\mathcal{F}_{n}\qty(z,\zb)$, which is independent of the local representative, is
\begin{align}
  \mathcal{F}_{n}\qty(z,\zb)
  &\coloneqq
  \dd\mathcal{A}_{n}\qty(z,\zb)
  =
  -\ii
  \partial_{z}\partial_{\zb}
  \ln\qty(\tau_{n}\qty(z,\zb))
  \dd z\wedge\dd\zb
  \notag\\
  &=
  -\ii b_{n}\qty(z,\zb)^{2}
  \dd z\wedge\dd\zb.
  \label{eq:krylov-exterior-berry-curvature}
\end{align}
If the map to the Grassmannian extends to the entire compact Riemann surface $\Sigma$ without boundary, the locally defined curvature $\mathcal{F}_{n}$ extends to a global $2$-form on $\Sigma$.
Its first Chern number is
\begin{align}
  C_{n}
  &\coloneqq
  \frac{1}{2\pi}
  \int_{\Sigma}\mathcal{F}_{n}\qty(z,\zb)
  =
  -\frac{\ii}{2\pi}
  \int_{\Sigma}\dd z\wedge\dd\zb\,
  b_{n}\qty(z,\zb)^{2}
  \in\Z_{\le0}.
  \label{eq:krylov-exterior-first-chern-number}
\end{align}
Thus, $C_{n}$ is expressed in terms of the Flaschka variable $b_{n}\qty(z,\zb)^{2}$.
Because $b_{n}\qty(z,\zb)^{2}\ge0$, $C_{n}$ is nonpositive.

For a compact region $\Omega\subset\Sigma$ covered by a single local representative and with a smooth boundary, Stokes' theorem gives
\begin{align}
  \int_{\Omega}\mathcal{F}_{n}
  &=
  2\int_{\Omega}\dd\qty(\dd z\mathcal{A}_{n,z})
  =
  2\oint_{\partial\Omega}\dd z\,\mathcal{A}_{n,z}.
  \label{eq:krylov-exterior-stokes-theorem}
\end{align}
Substituting Eqs.~\eqref{eq:krylov-exterior-berry-connection-z} and \eqref{eq:krylov-exterior-berry-curvature} gives
\begin{align}
  \int_{\Omega}
  \dd z\wedge\dd\zb\,
  b_{n}\qty(z,\zb)^{2}
  &=
  \oint_{\partial\Omega}
  \dd z\,
  \qty(
    \sum_{m=0}^{n-1}a_{m}\qty(z,\zb)
  ).
  \label{eq:krylov-exterior-gauss-law}
\end{align}
Appendix~\ref{app:krylov-electrostatics} interprets this relation as Gauss's law in Krylov electrostatics.

Figure~\ref{fig:krylov-stokes-gauss} shows the correspondence between the area and boundary terms in Eq.~\eqref{eq:krylov-exterior-gauss-law}.

\begin{figure}[t]
  \centering
  \begin{tikzpicture}[
    x=1cm,
    y=1cm,
    plane/.style={
      draw=black!70,
      line width=0.7pt,
      fill=gray!3
    },
    region/.style={
      draw=black!78,
      line width=0.8pt,
      fill=blue!7
    },
    orientation/.style={
      -{Latex[length=3.3mm,width=2.4mm]},
      line width=1.35pt,
      line cap=round,
      draw=black!88
    },
    every node/.style={
      font=\small
    }
  ]
    \filldraw[plane]
      (-3.00,-1.50) rectangle (3.00,1.55);
    \node[anchor=north west,font=\footnotesize]
      at (-2.90,1.47)
      {$\qty(z,\zb)\in\Sigma$};

    \filldraw[region]
      (0,0.25) ellipse[x radius=2.15,y radius=0.78];

    \draw[orientation]
      plot[
        domain=200:340,
        samples=60,
        smooth
      ]
      ({2.15*cos(\x)},{0.25+0.78*sin(\x)});
    \node[anchor=east,font=\footnotesize]
      at (-2.10,0.28)
      {$\partial\Omega$};
    \node[font=\footnotesize]
      at (0,0.58)
      {$\Omega$};

    \node[align=center,font=\scriptsize]
      at (0,0.08)
      {$\displaystyle
        \int_{\Omega}
        \dd z\wedge\dd\zb\,
        b_{n}\qty(z,\zb)^{2}$};

    \node[align=center,font=\scriptsize]
      at (0,-1.05)
      {$\displaystyle
        =
        \oint_{\partial\Omega}
        \dd z\,
        \qty(
          \sum_{m=0}^{n-1}
          a_{m}\qty(z,\zb)
        )$};
  \end{tikzpicture}
  \caption{
    Stokes--Gauss relation for a compact region $\Omega$ of the parameter surface $\Sigma$.
    If $\Omega$ is covered by a single local representative and $\partial\Omega$ is smooth, the area integral of $b_{n}^{2}$ equals the boundary integral of $\sum_{m=0}^{n-1}a_{m}$.
    The arrow indicates the positive orientation of $\partial\Omega$.
  }
  \label{fig:krylov-stokes-gauss}
\end{figure}
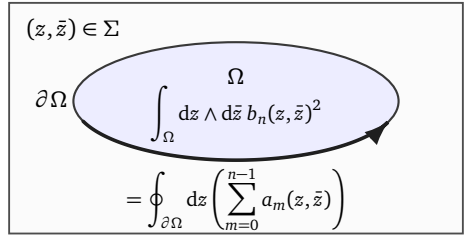

\subsection{Fubini--Study metric}
\label{subsec:krylov-fubini-study-metric}

Each component of the holomorphic Krylov exterior state can be identified with a Pl\"ucker coordinate of the $n$-dimensional Krylov subspace $\K_{n}\qty(\HH,\ket{\psi\qty(z)})$ [see Eq.~\eqref{eq:plucker-embedding}].
We pull back the Fubini--Study metric on the Pl\"ucker image to $\Sigma$ along the Krylov-subspace map~\cite{ProvostVallee1980}.
Equation~\eqref{eq:taun-innerproduct} gives its local K\"ahler potential as
\begin{align}
  K_{n}^{\mathrm{FS}}\qty(z,\zb)
  &\coloneqq
  \ln\qty(\tau_{n}\qty(z,\zb)).
  \label{eq:krylov-fubini-study-kahler-potential}
\end{align}
Here, $\tau_{n}$ is the squared norm of the local representative vector $\wket{\Psi_{n}}$ and changes under a change of representative.
By contrast, $\partial_{z}\partial_{\zb}\ln\qty(\tau_{n})$ is invariant under such a change.
Equation~\eqref{eq:2d-toda-log-tau} gives the metric component and line element as
\begin{align}
  g_{z\zb}^{\qty(n)}\qty(z,\zb)
  &\coloneqq
  \partial_{z}\partial_{\zb}
  K_{n}^{\mathrm{FS}}\qty(z,\zb)
  =
  b_{n}\qty(z,\zb)^{2},
  \label{eq:krylov-fubini-study-metric}\\
  \dd s_{n,\mathrm{FS}}^{2}
  &=
  b_{n}\qty(z,\zb)^{2}\dd z\,\dd\zb.
  \label{eq:krylov-fubini-study-line-element}
\end{align}
The corresponding K\"ahler form is
\begin{align}
  \omega_{n}^{\mathrm{FS}}\qty(z,\zb)
  &\coloneqq
  \ii b_{n}\qty(z,\zb)^{2}
  \dd z\wedge\dd\zb
  =
  -\mathcal{F}_{n}\qty(z,\zb).
  \label{eq:krylov-fubini-study-kahler-form}
\end{align}
Equations~\eqref{eq:krylov-fubini-study-metric} and \eqref{eq:krylov-fubini-study-kahler-form} show that the same Flaschka variable $b_{n}\qty(z,\zb)^{2}$ determines the Fubini--Study metric and the Berry curvature.
Equation~\eqref{eq:krylov-grassmann-metric-arnoldi} in the Appendix gives a direct derivation of this identity using the Arnoldi basis.

Figure~\ref{fig:krylov-geometry-map} shows the relation between the holomorphic Krylov-subspace map and the geometric quantities induced by $\tau_{n}$.

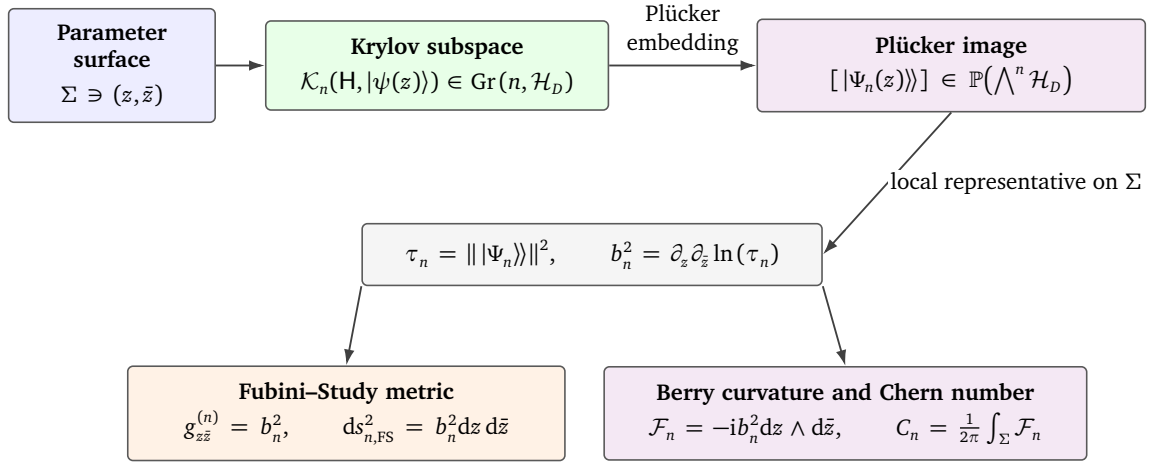
\begin{figure}[t]
  \centering
  \begin{tikzpicture}[
    manuscript flow diagram,
    box/.append style={
      font=\footnotesize,
      inner xsep=4pt,
      inner ysep=6pt
    }
  ]
    \node[
      box,
      fill=blue!7,
      text width=2.45cm
    ] (sigma) at (-6.35,1.25) {
      {\footnotesize\textbf{Parameter}}\\[-0.2ex]
      {\footnotesize\textbf{surface}}\\[2pt]
      $\Sigma\ni\qty(z,\zb)$
    };

    \node[
      box,
      fill=green!9,
      text width=4.25cm
    ] (grassmann) at (-2.05,1.25) {
      {\footnotesize\textbf{Krylov subspace}}\\[2pt]
      $\K_{n}\qty(\HH,\ket{\psi\qty(z)})
      \in\Gr\qty(n,\mathcal{H}_{D})$
    };

    \node[
      box,
      fill=violet!9,
      text width=4.85cm
    ] (projective) at (4.75,1.25) {
      {\footnotesize\textbf{Pl\"ucker image}}\\[2pt]
      $\qty[\wket{\Psi_{n}\qty(z)}]
      \in\mathbb{P}\qty(\bigwedge^{n}\mathcal{H}_{D})$
    };

    \node[
      box,
      fill=gray!8,
      text width=5.80cm
    ] (tau) at (0,-1.25) {
      $\tau_{n}=\norm{\wket{\Psi_{n}}}^{2}$,
      \qquad
      $b_{n}^{2}
      =\partial_{z}\partial_{\zb}\ln\qty(\tau_{n})$
    };

    \node[
      box,
      fill=orange!10,
      text width=5.50cm
    ] (metric) at (-3.25,-3.35) {
      {\footnotesize\textbf{Fubini--Study metric}}\\[2pt]
      $g_{z\zb}^{\qty(n)}=b_{n}^{2}$,
      \qquad
      $\dd s_{n,\mathrm{FS}}^{2}=b_{n}^{2}\dd z\,\dd\zb$
    };

    \node[
      box,
      fill=violet!9,
      text width=6.10cm
    ] (berry) at (3.35,-3.35) {
      {\footnotesize\textbf{Berry curvature and Chern number}}\\[2pt]
      $\mathcal{F}_{n}
      =-\ii b_{n}^{2}\dd z\wedge\dd\zb$,
      \qquad
      $C_{n}=\frac{1}{2\pi}\int_{\Sigma}\mathcal{F}_{n}$
    };

    \draw[flow] (sigma.east) -- (grassmann.west);
    \draw[flow]
      (grassmann.east)
      --
      node[above,align=center,font=\footnotesize] {Pl\"ucker\\embedding}
      (projective.west);
    \draw[flow]
      (projective.south)
      --
      node[
        midway,
        right,
        fill=white,
        inner sep=1pt,
        font=\footnotesize
      ] {local representative on $\Sigma$}
      (tau.east);
    \draw[flow] (tau.south west) -- (metric.north);
    \draw[flow] (tau.south east) -- (berry.north);
  \end{tikzpicture}
  \caption{
    Geometric quantities obtained from the holomorphic Krylov map.
    In the region where $\operatorname{rank}\qty(\mathsf{K}_{n})=n$, $\tau_{n}=\norm{\wket{\Psi_{n}}}^{2}$ gives the squared norm of a local representative of the Pl\"ucker image, and $b_{n}^{2}=\partial_{z}\partial_{\zb}\ln\qty(\tau_{n})$ determines the Fubini--Study metric and the Berry curvature.
    If the map of Krylov subspaces extends holomorphically to the entire compact $\Sigma$ without boundary, the curvature integral gives the Chern number $C_{n}$.
  }
  \label{fig:krylov-geometry-map}
\end{figure}

\subsection{Chern numbers and Arnoldi breakdown in a three-site model}
\label{subsec:three-site-chern-breakdown}

We use a three-site model with nonreciprocal, non-Hermitian couplings to examine the relation
between Arnoldi breakdown and changes in the Chern numbers.
In this model, Arnoldi breakdown appears as a reduction in the Krylov
dimension.
Let $0\leq\theta\leq\pi/2$ be an external parameter of the Hamiltonian.
We impose the cyclic-vector assumption used in the preceding sections only for
$0<\theta<\pi/2$ and relax it at the endpoints $\theta=0,\pi/2$ to treat
Arnoldi breakdown.
Let $\ket{e_{1}},\ket{e_{2}},\ket{e_{3}}$ be the standard basis of
$\mathcal{H}_{3}\simeq\C^{3}$.
Using components in this basis, we define the Hamiltonian and the initial
state as follows:
\begin{align}
  \HH\qty(\theta)
  &\coloneqq
  \begin{pmatrix}
    0 & 0 & 0\\
    \cos\qty(\theta) & 0 & 0\\
    0 & \sin\qty(\theta) & 0
  \end{pmatrix}
  =
  \cos\qty(\theta)\ket{e_{2}}\bra{e_{1}}
  +
  \sin\qty(\theta)\ket{e_{3}}\bra{e_{2}},
  \notag\\
  \ket{\psi\qty(0)}
  &\coloneqq
  \begin{pmatrix}
    1\\
    0\\
    0
  \end{pmatrix}
  =
  \ket{e_{1}}.
  \label{eq:three-site-hamiltonian-seed}
\end{align}
The static Krylov sequence generated from this initial state is as follows:
\begin{align}
  \HH\qty(\theta)\ket{\psi\qty(0)}
  &=
  \begin{pmatrix}
    0\\
    \cos\qty(\theta)\\
    0
  \end{pmatrix}
  =
  \cos\qty(\theta)\ket{e_{2}},
  \notag\\
  \HH\qty(\theta)^{2}\ket{\psi\qty(0)}
  &=
  \begin{pmatrix}
    0\\
    0\\
    \cos\qty(\theta)\sin\qty(\theta)
  \end{pmatrix}
  =
  \cos\qty(\theta)\sin\qty(\theta)\ket{e_{3}},
  \notag\\
  \HH\qty(\theta)^{3}
  &=0.
  \label{eq:three-site-static-krylov-chain}
\end{align}
For each $\theta$, define the unnormalized holomorphic state by
\begin{align}
  \ket{\psi\qty(z;\theta)}
  &\coloneqq
  \eno{-z\HH\qty(\theta)}
  \ket{\psi\qty(0)}.
  \notag
\end{align}
Because $\eno{-z\HH\qty(\theta)}$ is invertible and commutes with
$\HH\qty(\theta)$, the Krylov dimension is independent of $z$ and takes the
following values:
\begin{align}
  d_{\mathrm K}\qty(z;\theta)
  &\coloneqq
  \dim\qty[\Span\qty(
    \ket{\psi\qty(z;\theta)},
    \HH\qty(\theta)\ket{\psi\qty(z;\theta)},
    \HH\qty(\theta)^{2}\ket{\psi\qty(z;\theta)}
  )]
  =
  \begin{cases}
    2,
    &\theta=0,\\[1mm]
    3,
    &0<\theta<\dfrac{\pi}{2},\\[2mm]
    1,
    &\theta=\dfrac{\pi}{2}
  \end{cases}.
  \label{eq:three-site-krylov-dimension-holomorphic}
\end{align}

For each $\theta$, we apply the Arnoldi method with the normalized seed vector
\begin{align}
  \ket{u_{0}\qty(z,\zb;\theta)}
  &\coloneqq
  \frac{
    \ket{\psi\qty(z;\theta)}
  }{
    \sqrt{\tau_{1}\qty(z,\zb;\theta)}
  }.
  \notag
\end{align}
The squared Arnoldi coefficients $b_{n}\qty(z,\zb;\theta)^{2}$ are given
below:
\begingroup
\small
\begin{align}
  b_{1}\qty(z,\zb;\theta)^{2}
  &\!=\!
  \frac{
    \cos^{2}\qty(\theta)
    \qty[
      1
      +\sin^{2}\qty(\theta)\abs{z}^{2}
      +\frac{\cos^{2}\qty(\theta)\sin^{2}\qty(\theta)}{4}
      \abs{z}^{4}
    ]
  }{
    \qty[
      1
      +\cos^{2}\qty(\theta)\abs{z}^{2}
      +\frac{\cos^{2}\qty(\theta)\sin^{2}\qty(\theta)}{4}
      \abs{z}^{4}
    ]^{2}
  },
  \label{eq:three-site-b-one}
  \\
  b_{2}\qty(z,\zb;\theta)^{2}
  &\!=\!
  \frac{
    \sin^{2}\qty(\theta)
    \qty[
      1
      +\cos^{2}\qty(\theta)\abs{z}^{2}
      +\frac{\cos^{2}\qty(\theta)\sin^{2}\qty(\theta)}{4}
      \abs{z}^{4}
    ]
  }{
    \qty[
      1
      +\sin^{2}\qty(\theta)\abs{z}^{2}
      +\frac{\cos^{2}\qty(\theta)\sin^{2}\qty(\theta)}{4}
      \abs{z}^{4}
    ]^{2}
  }.
  \label{eq:three-site-b-two}
\end{align}
\endgroup
Equation~\eqref{eq:three-site-b-one} holds for
$0\leq\theta\leq\pi/2$, whereas Eq.~\eqref{eq:three-site-b-two} holds for
$0\leq\theta<\pi/2$.
At $\theta=0$, Eq.~\eqref{eq:three-site-b-two} gives
$b_{2}\qty(z,\zb;0)^{2}=0$, which marks the end of the two-dimensional
Arnoldi chain.
At $\theta=\pi/2$, however, the second Arnoldi basis vector cannot be
constructed, so $b_{2}\qty(z,\zb;\pi/2)^{2}$ is undefined rather than zero.
Equation~\eqref{eq:three-site-b-one} then gives
$b_{1}\qty(z,\zb;\pi/2)^{2}=0$.
For each $\theta$ and $1\leq n\leq d_{\mathrm K}$, the Pl\"ucker
representative of the Krylov map is a nonzero polynomial vector in $z$.
After removing any common polynomial factor and homogenizing its components,
the representative defines a holomorphic extension to
$\Sigma=\C\cup\qty{\infty}\simeq\CP^{1}$.
Consistently, $b_{n}\qty(z,\zb;\theta)^{2}$ decays sufficiently fast at
$z=\infty$.
Integrating Eq.~\eqref{eq:krylov-exterior-first-chern-number} using this
extension of the Pl\"ucker map gives Table~\ref{tab:three-site-chern-strata}.

\begin{table}[t]
  \centering
  \caption{
    Krylov dimensions and Chern numbers of the three-site model.
    At $\theta=\pi/2$, the second exterior-product state cannot be
    constructed, so $C_{2}$ is undefined rather than zero.
  }
  \label{tab:three-site-chern-strata}
  \begin{tabular}{c c c c}
      \hline
      $\theta$
      & $d_{\mathrm K}$
      & $C_{1}$
      & $C_{2}$
      \\
      \hline
      $0$
      & $2$
      & $-1$
      & $0$
      \\
      $0<\theta<\pi/2$
      & $3$
      & $-2$
      & $-2$
      \\
      $\pi/2$
      & $1$
      & $0$
      & $\text{undefined}$
      \\
      \hline
  \end{tabular}
\end{table}

Figure~\ref{fig:three-site-chern-breakdown} shows the directed three-site
chain and the $\abs{z}$ dependence of $b_{1}\qty(z,\zb)^{2}$ and
$b_{2}\qty(z,\zb)^{2}$ at $\theta=\pi/6$.
At $\theta=\pi/6$, $b_{1}^{2}$ decreases monotonically from $\abs{z}=0$,
whereas $b_{2}^{2}$ has a broad maximum at a finite $\abs{z}$ and then
decreases.
In the interior of the parameter interval, both links in this model remain
active, so all three sites are reachable from the seed vector.
When either link is cut at an endpoint, fewer sites are reachable and the
Krylov dimension decreases.
Because this reduction changes the object on which the Chern number is
defined, the Chern number jumps to a different integer value or becomes
undefined when the corresponding Krylov subspace cannot be constructed.

\begin{figure}[t]
  \centering
  \definecolor{arnoldiblue}{RGB}{0,114,178}
  \definecolor{todaorange}{RGB}{213,94,0}
  \resizebox{\linewidth}{!}{%
  \begin{tikzpicture}[
    font=\small,
    site/.style={
      circle,
      draw=black!70,
      minimum size=12mm,
      inner sep=0pt,
      line width=0.6pt,
      fill=white
    },
    seed/.style={
      site,
      fill=blue!8,
      line width=0.9pt,
      draw=arnoldiblue
    },
    hop/.style={
      -{Latex[length=2.2mm,width=1.45mm]},
      line width=0.9pt
    },
    panellabel/.style={
      font=\bfseries\small,
      anchor=west
    },
    paneltitle/.style={
      font=\bfseries\small,
      anchor=west
    }
  ]
    \node[panellabel] at (0,5.35) {(a)};
    \node[paneltitle] at (0.55,5.35)
      {Directed three-site Krylov chain};

    \node[seed] (s1) at (1.00,3.80) {$\ket{e_{1}}$};
    \node[site] (s2) at (3.25,3.80) {$\ket{e_{2}}$};
    \node[site] (s3) at (5.50,3.80) {$\ket{e_{3}}$};

    \draw[hop]
      (s1)
      --
      node[above=3pt] {$\cos\qty(\theta)$}
      (s2);
    \draw[hop]
      (s2)
      --
      node[above=3pt] {$\sin\qty(\theta)$}
      (s3);

    \node[align=center,font=\small,arnoldiblue]
      at (1.00,2.95)
      {seed vector};

    \node[align=center,font=\normalsize,inner sep=0pt]
      at (3.25,1.75)
      {$\displaystyle
       \HH\ket{e_{1}}
       =\cos\qty(\theta)\ket{e_{2}}$
       \\[7pt]
       $\displaystyle
       \HH^{2}\ket{e_{1}}
       =\cos\qty(\theta)\sin\qty(\theta)\ket{e_{3}}$};

    \node[panellabel] at (8.15,5.35) {(b)};
    \node[paneltitle] at (8.70,5.35)
      {Squared Arnoldi coefficients at $\theta=\pi/6$};

    \begin{axis}[
      at={(8.65cm,0.45cm)},
      anchor=south west,
      width=7.55cm,
      height=4.35cm,
      scale only axis,
      xmin=0,
      xmax=6,
      ymin=0,
      ymax=0.8,
      xtick={0,1,2,3,4,5,6},
      ytick={0,0.2,0.4,0.6,0.8},
      xlabel={$\abs{z}$},
      ylabel={$b_{1,2}\qty(z,\zb)^{2}$},
      tick label style={font=\footnotesize},
      label style={font=\small},
      axis line style={black!70},
      tick style={black!70},
      grid=major,
      grid style={black!12,line width=0.35pt},
      legend style={
        draw=none,
        fill=none,
        font=\footnotesize,
        at={(0.97,0.97)},
        anchor=north east,
        cells={anchor=west}
      },
      clip=true
    ]
      \addplot[
        arnoldiblue,
        line width=1.20pt,
        domain=0:6,
        samples=241
      ]
        {48*(3*x^4+16*x^2+64)/
          (3*x^4+48*x^2+64)^2};
      \addlegendentry{$n=1$}

      \addplot[
        todaorange,
        line width=1.20pt,
        dashed,
        domain=0:6,
        samples=241
      ]
        {16*(3*x^4+48*x^2+64)/
          (3*x^4+16*x^2+64)^2};
      \addlegendentry{$n=2$}
    \end{axis}
  \end{tikzpicture}%
  }
  \caption{
    Three-site non-Hermitian nonreciprocal model.
    (a)~Directed Krylov chain generated by the seed vector
    $\ket{\psi\qty(0)}=\ket{e_{1}}$, with couplings
    $\cos\qty(\theta)$ and $\sin\qty(\theta)$.
    (b)~Dependence of $b_{1}\qty(z,\zb)^{2}$ (blue solid line) and
    $b_{2}\qty(z,\zb)^{2}$ (orange dashed line) on $\abs{z}$ at
    $\theta=\pi/6$, where $d_{\mathrm K}=3$
    [Eqs.~\eqref{eq:three-site-b-one} and~\eqref{eq:three-site-b-two}].
    Both tend to zero as $\abs{z}\to\infty$.
  }
  \label{fig:three-site-chern-breakdown}
\end{figure}

\section{Shortcut to adiabaticity in the Arnoldi--Krylov chain}
\label{STA-in-Krylov-chain}

In this section, we fix $\HH$ and the initial seed vector and consider a smooth
real path along which all $D$ Arnoldi basis vectors remain defined.

We represent the Arnoldi matrix as a Krylov chain and construct a Hermitian
auxiliary term that generates its isospectral deformation.
We then derive the evolution operator for the resulting total generator.
When the spectrum is nondegenerate and complete sets of right and left
eigenvectors exist, we show that this evolution cancels transitions between
distinct instantaneous eigenspaces.
Appendix~\ref{app:sta-basics} reviews general counterdiabatic driving.
For Hermitian systems, counterdiabatic terms have been constructed in an
operator-space Krylov basis using the Lanczos method
~\cite{TakahashiDelCampo2024,Bhattacharjee2023AGP}.
Extensions to non-Hermitian systems using the bi-Lanczos and Arnoldi methods
have also been proposed~\cite{ShresthaBhattacharjeeDelCampo2026}.
These approaches expand the adiabatic gauge potential in an operator-space
Krylov basis.
Here, instead, we use the connection of the Arnoldi basis obtained from state
vectors to construct an auxiliary term that generates an isospectral
deformation of the Arnoldi matrix.

\subsection{Arnoldi--Krylov chain and counterdiabatic generator}
\label{subsec:sta-arnoldi-path-generator}

Let $\ket{u_{0}},\ldots,\ket{u_{D-1}}$ denote the Arnoldi basis.
The Arnoldi matrix in Eq.~\eqref{eq:toda-arnoldi-hessenberg-form} can be
written in the following form:
\begin{align}
  \Aarn
  &=
  \begin{pmatrix}
    a_{0} & h_{0,1} & \cdots & h_{0,D-1} \\
    b_{1} & a_{1} & \ddots & \vdots \\
    \vdots & \ddots & \ddots & h_{D-2,D-1} \\
    0 & \cdots & b_{D-1} & a_{D-1}
  \end{pmatrix}.
  \label{eq:sta-krylov-chain}
\end{align}
Here, $h_{m,n}$ with $m<n$ are the strictly upper-triangular entries of
$\Aarn$.
Only the nearest-neighbor coupling $b_{n+1}$ appears in the direction from
$\ket{u_{n}}$ to $\ket{u_{n+1}}$, whereas long-range couplings $h_{m,n}$ may
appear in the reverse direction $\ket{u_{n}}\to\ket{u_{m}}$ $\qty(m<n)$.
Thus, a general Arnoldi matrix has different coupling ranges in the two
directions.

We parameterize a real path in the complex deformation-parameter space as
follows:
\begin{align}
  \gamma:
  s
  &\longmapsto
  \qty(z\qty(s),\zb\qty(s)),
  &
  \zb\qty(s)=z\qty(s)^{*}.
  \label{eq:sta-protocol-path}
\end{align}
Here, $s$ denotes the protocol time, and $z$ is the holomorphic deformation
parameter in Eq.~\eqref{eq:EqMstate}.
An overdot denotes differentiation with respect to $s$.

Let $\Uarn$ be the unitary matrix whose columns are the Arnoldi basis vectors.
Appendix~\ref{app:main-theorem-qr-proof},
Eq.~\eqref{eq:qr-frame-connections}, gives the following frame connections:
\begin{align}
  \mathsf{F}_{z}
  &\coloneqq
  -\Uarn^{\dagger}\partial_{z}\Uarn
  =
  \begin{pmatrix}
    \dfrac{a_{0}}{2} & 0 & \cdots & 0 \\
    b_{1} & \dfrac{a_{1}}{2} & \ddots & \vdots \\
    \vdots & \ddots & \ddots & 0 \\
    0 & \cdots & b_{D-1} & \dfrac{a_{D-1}}{2}
  \end{pmatrix},
  \notag\\
  \mathsf{F}_{\zb}
  &\coloneqq
  -\Uarn^{\dagger}\partial_{\zb}\Uarn
  =
  -
  \begin{pmatrix}
    \dfrac{\bar{a}_{0}}{2} & b_{1} & \cdots & 0 \\
    0 & \dfrac{\bar{a}_{1}}{2} & \ddots & \vdots \\
    \vdots & \ddots & \ddots & b_{D-1} \\
    0 & \cdots & 0 & \dfrac{\bar{a}_{D-1}}{2}
  \end{pmatrix}.
  \label{eq:sta-frame-connection}
\end{align}

We define the auxiliary generator along the path $\gamma$ as follows:
\begin{align}
  \Gmat_{\gamma}
  &\coloneqq
  \ii
  \qty(
    \dot z\mathsf{F}_{z}
    +
    \dot{\zb}\mathsf{F}_{\zb}
  )
  =
  -\ii\Uarn^{\dagger}\dot{\Uarn}.
  \label{eq:sta-path-generator}
\end{align}
Substituting Eq.~\eqref{eq:sta-frame-connection} gives
\begin{align}
  d_{n}
  &\coloneqq
  \frac{\ii}{2}
  \qty(
    \dot z a_{n}
    -
    \dot{\zb}\bar{a}_{n}
  ),
  \notag\\
  \Gmat_{\gamma}
  &=
  \begin{pmatrix}
    d_{0} & -\ii\dot{\zb}b_{1} & \cdots & 0 \\
    \ii\dot z b_{1} & d_{1} & \ddots & \vdots \\
    \vdots & \ddots & \ddots & -\ii\dot{\zb}b_{D-1} \\
    0 & \cdots & \ii\dot z b_{D-1} & d_{D-1}
  \end{pmatrix}.
  \label{eq:sta-path-generator-elements}
\end{align}
Because $\dot{\zb}=\dot z^{*}$ along a real path, $\Gmat_{\gamma}$ is a
Hermitian tridiagonal matrix in the Arnoldi basis.

Since $\HH$ is independent of $s$, $\Aarn=\Uarn^{\dagger}\HH\Uarn$ and
Eq.~\eqref{eq:sta-path-generator} give the Lax equation:
\begin{align}
  \ii\dot{\Aarn}
  &=
  \comm{\Gmat_{\gamma}}{\Aarn}.
  \label{eq:sta-path-lax}
\end{align}
Fix the initial point at $s=0$ on the real path and define
\begin{align}
  \mathsf{W}\qty(s)
  &\coloneqq
  \Uarn^{\dagger}\qty(s)\Uarn\qty(0).
  \label{eq:sta-frame-intertwiner}
\end{align}
This definition gives
\begin{align}
  \ii\partial_{s}\mathsf{W}\qty(s)
  &=
  \Gmat_{\gamma}\qty(s)\mathsf{W}\qty(s),
  \notag\\
  \Aarn\qty(s)
  &=
  \mathsf{W}\qty(s)
  \Aarn\qty(0)
  \mathsf{W}\qty(s)^{\dagger}.
  \label{eq:sta-frame-transport}
\end{align}
Thus, the spectrum and Jordan canonical form of $\Aarn\qty(s)$ are preserved
along the path.

\subsection{Exact transport and non-Hermitian counterdiabatic driving}
\label{subsec:sta-exact-transport}

Consider the evolution generated by adding $\Gmat_{\gamma}\qty(s)$ to
$\Aarn\qty(s)$:
\begin{align}
  \ii\dv{}{s}\ket{\Psi\qty(s)}
  &=
  \qty[
    \Aarn\qty(s)
    +
    \Gmat_{\gamma}\qty(s)
  ]
  \ket{\Psi\qty(s)}.
  \label{eq:sta-krylov-total-evolution}
\end{align}
The corresponding evolution operator takes the form:
\begin{align}
  \mathsf{S}\qty(s)
  &=
  \mathsf{W}\qty(s)
  \eno{-\ii\Aarn\qty(0)s}
  =
  \Uarn^{\dagger}\qty(s)
  \eno{-\ii\HH s}
  \Uarn\qty(0).
  \label{eq:sta-krylov-propagator}
\end{align}
Therefore, an initial state $\ket{\Psi\qty(0)}$ evolves as follows:
\begin{align}
  \ket{\Psi\qty(s)}
  &=
  \mathsf{S}\qty(s)\ket{\Psi\qty(0)}
  =
  \Uarn^{\dagger}\qty(s)
  \eno{-\ii\HH s}
  \Uarn\qty(0)\ket{\Psi\qty(0)}.
  \label{eq:sta-krylov-state-evolution}
\end{align}
Direct differentiation using Eq.~\eqref{eq:sta-frame-transport} verifies that
this expression satisfies the evolution equation
~\eqref{eq:sta-krylov-total-evolution}.
The second equality in Eq.~\eqref{eq:sta-krylov-propagator} is an identity
that expresses evolution under the fixed $\HH$ in the time-dependent Arnoldi
basis.

Equation~\eqref{eq:sta-krylov-propagator} does not assume that $\Aarn$ is
diagonalizable and therefore also applies when it is nondiagonalizable.

Next, suppose that the eigenvalues of $\Aarn\qty(s)$ are nondegenerate on the
interval of interest and that complete sets of right and left eigenvectors
exist.
We normalize the right and left eigenvectors as follows:
\begin{align}
  \Aarn\qty(s)\ket{r_{n}\qty(s)}
  &=
  E_{n}\ket{r_{n}\qty(s)},
  &
  \bra{l_{n}\qty(s)}\Aarn\qty(s)
  &=
  E_{n}\bra{l_{n}\qty(s)},
  \notag\\
  \braket{l_{m}\qty(s)}{r_{n}\qty(s)}
  &=
  \delta_{m,n}.
  \label{eq:sta-krylov-biorthogonal-eigenvectors}
\end{align}
Equation~\eqref{eq:sta-frame-transport} implies that the eigenvalues $E_{n}$
are independent of $s$.
Choose the eigenvectors as follows:
\begin{align}
  \ket{r_{n}\qty(s)}
  &=
  \mathsf{W}\qty(s)
  \ket{r_{n}\qty(0)},
  &
  \bra{l_{n}\qty(s)}
  &=
  \bra{l_{n}\qty(0)}
  \mathsf{W}\qty(s)^{\dagger}.
  \label{eq:sta-krylov-transported-eigenvectors}
\end{align}
With this choice,
$\ii\ket{\dot r_{n}}=\Gmat_{\gamma}\ket{r_{n}}$.
Substituting $\ket{\Psi}=\sum_{n}c_{n}\ket{r_{n}}$ into
Eq.~\eqref{eq:sta-krylov-total-evolution} gives
\begin{align}
  \ii\dot c_{n}
  &=
  E_{n}c_{n}.
  \label{eq:sta-krylov-coefficient-evolution}
\end{align}
Equation~\eqref{eq:sta-krylov-coefficient-evolution} shows that no transitions
occur between distinct instantaneous eigenspaces.
In this sense, $\Gmat_{\gamma}$ acts as a counterdiabatic term that cancels
nonadiabatic transitions.

Let $\Pmat_{n}\coloneqq\ket{r_{n}}\bra{l_{n}}$.
For a non-Hermitian system, the standard auxiliary term with zero diagonal
entries in the biorthogonal eigenbasis can be written as
follows~\cite{IbanezEtAl2011,IbanezEtAl2012Erratum}:
\begin{align}
  \HH_{\mathrm{cd,NH}}^{\qty(0)}
  &\coloneqq
  \ii
  \sum_{n}
  \qty[
    \ket{\dot r_{n}}\bra{l_{n}}
    -
    \braket{l_{n}}{\dot r_{n}}\Pmat_{n}
  ]
  =
  \Gmat_{\gamma}
  \notag\\
  &\quad
  -
  \sum_{n}\Pmat_{n}\Gmat_{\gamma}\Pmat_{n}.
  \label{eq:sta-krylov-standard-nonhermitian-cd}
\end{align}
The off-diagonal entries of $\HH_{\mathrm{cd,NH}}^{\qty(0)}$ and
$\Gmat_{\gamma}$ coincide in the biorthogonal eigenbasis.
Because $\HH_{\mathrm{cd,NH}}^{\qty(0)}$ contains oblique projections, it is
generally non-Hermitian and need not be tridiagonal in the Arnoldi basis.
By contrast, $\Gmat_{\gamma}$ includes diagonal entries in the biorthogonal
eigenbasis but is Hermitian and tridiagonal in the Arnoldi basis, as shown in
Eq.~\eqref{eq:sta-path-generator-elements}.

\subsection{Two-dimensional Toda lattice and Krylov geometry}
\label{subsec:sta-toda-geometry}

Because $\mathsf{F}_{z}$ and $\mathsf{F}_{\zb}$ are the Arnoldi-frame
connections obtained from $-\Uarn^{\dagger}\dd\Uarn$, they satisfy the
flatness condition:
\begin{align}
  \partial_{z}\mathsf{F}_{\zb}
  -
  \partial_{\zb}\mathsf{F}_{z}
  -
  \comm{\mathsf{F}_{z}}{\mathsf{F}_{\zb}}
  &=
  0.
  \label{eq:sta-frame-flatness}
\end{align}
Using the Toda--Arnoldi correspondence, the coefficients in these connections
satisfy the following equations~\cite{UenoTakasaki1984,Takasaki2018Toda}:
\begin{align}
  \partial_{\zb}a_{n}
  &=
  b_{n}^{2}-b_{n+1}^{2},
  \notag\\
  \partial_{z}b_{n}^{2}
  &=
  \qty(a_{n-1}-a_{n})b_{n}^{2}.
  \label{eq:sta-toda-flows}
\end{align}
The first equation holds for $0\leq n\leq D-1$, and the second for
$1\leq n\leq D-1$, with boundary conditions $b_{0}=b_{D}=0$.
Equation~\eqref{eq:sta-frame-flatness} expresses the compatibility between
changes in the Arnoldi basis along the $z$ and $\zb$ directions.

In counterdiabatic driving, we do not treat $z$ and $\zb$ as two independent
times.
We instead evaluate the connections along the real path in
Eq.~\eqref{eq:sta-protocol-path}.
Multiplying their component along the path
$\dot z\mathsf{F}_{z}+\dot{\zb}\mathsf{F}_{\zb}$ by $\ii$ gives
$\Gmat_{\gamma}$ in Eq.~\eqref{eq:sta-path-generator}.
The Flaschka variables of the two-dimensional Toda lattice therefore
determine the diagonal and nearest-neighbor entries of $\Gmat_{\gamma}$.
In contrast, the general upper-triangular entries $h_{m,n}$ in
Eq.~\eqref{eq:sta-krylov-chain} need not be determined by $a_{n}$ and $b_{n}$
alone.

Equation~\eqref{eq:krylov-fubini-study-metric} gives the Fubini--Study metric
of the $n$-dimensional Krylov subspace as
$g_{z\zb}^{\qty(n)}=b_{n}^{2}$.
Let $\ell_{n,\mathrm{FS}}$ denote the Fubini--Study arc length along the real
path.
For $1\leq n\leq D-1$, we then have
\begin{align}
  \qty(
    \dv{\ell_{n,\mathrm{FS}}}{s}
  )^{2}
  &=
  g_{z\zb}^{\qty(n)}
  \dot z\dot{\zb}
  =
  b_{n}^{2}\dot z\dot{\zb},
  \notag\\
  \abs{
    \qty(\Gmat_{\gamma})_{n,n-1}
  }^{2}
  &=
  b_{n}^{2}\dot z\dot{\zb}
  =
  \qty(
    \dv{\ell_{n,\mathrm{FS}}}{s}
  )^{2}.
  \label{eq:sta-fs-speed}
\end{align}
Thus, the squared magnitude of the subdiagonal matrix element
$\qty(\Gmat_{\gamma})_{n,n-1}$ equals the squared speed measured by the
Fubini--Study metric of the corresponding Krylov subspace.

Finally, consider the case $\HH=\HH^{\dagger}$.
Then $\Aarn=\LLop$ and $a_{n}\in\R$.
Equations~\eqref{eq:sta-frame-connection},
\eqref{eq:sta-path-generator},
\eqref{eq:toda-lax-m}, and~\eqref{eq:toda-lax-l} give
\begin{align}
  \Gmat_{\gamma}
  &=
  \qty(\dot z+\dot{\zb})\ii\MM
  -
  \frac{\dot z-\dot{\zb}}{2\ii}\LLop.
  \label{eq:sta-hermitian-limit-generator}
\end{align}
The second term commutes with $\LLop$ and does not contribute to transitions
between instantaneous eigenstates.
The auxiliary term with zero diagonal entries in the instantaneous eigenbasis
therefore reduces to
\begin{align}
  \Hcd^{\qty(0)}
  &=
  \qty(\dot z+\dot{\zb})\ii\MM.
  \label{eq:sta-hermitian-zero-diagonal}
\end{align}
Thus, the general expression in Eq.~\eqref{eq:sta-path-generator-elements}
reduces in the Hermitian limit to the result of Okuyama and Takahashi
~\cite{OkuyamaTakahashi2016}.

\section{Conclusion}
\label{sec:conclusion}

We have established a finite-dimensional Toda--Arnoldi correspondence for a
fixed, generally non-Hermitian Hamiltonian and a holomorphically deformed
cyclic state.
In this correspondence, each $\tau$ function can be regarded as the square of
the volume spanned by the corresponding unnormalized holomorphic Krylov vectors.
We proved directly that the Toda Flaschka variables coincide with the diagonal and subdiagonal Arnoldi coefficients, and obtained the same result independently by QR factorization.
Differentiating the QR factorization also gives the Arnoldi Lax equation.
The correspondence reduces to the Toda--Lanczos relation in the Hermitian
limit~\cite{DymarskyGorsky2020,TakahashiNandyDelCampo2026}.
It is exact in the cyclic region, but it does not determine the remaining
entries of the upper Hessenberg matrix.

The same correspondence shows that $b_{n}^{2}$ determines both the Fubini--Study metric and the Berry curvature of the holomorphic Krylov subspace, while $\sum_{m=0}^{n-1}a_{m}$ determines the $z$ component of the Berry connection in the chosen phase convention.
When the Krylov-subspace map extends holomorphically to a compact Riemann
surface without boundary, the curvature integral defines a Chern number.
When Arnoldi breakdown changes the Krylov dimension, the associated Chern
numbers can change or become undefined.
Such breakdown does not occur along the holomorphic flow of a cyclic seed,
because that flow preserves cyclicity.

Along a smooth real path in the cyclic region, the Arnoldi-frame connection
provides a Hermitian tridiagonal generator for an exact similarity evolution
of the Arnoldi matrix.
This evolution preserves the spectrum even when the matrix is
nondiagonalizable.
In the nondegenerate, diagonalizable case, the same generator cancels
transitions between instantaneous eigenspaces and thus realizes
counterdiabatic driving.
This construction uses a moving state-space Arnoldi frame, whereas related
approaches expand an adiabatic gauge potential in an operator-space Krylov
basis~\cite{TakahashiDelCampo2024,ShresthaBhattacharjeeDelCampo2026}.

A natural extension of this correspondence is to superoperator dynamics, where Arnoldi and bi-Lanczos
methods have already been applied
~\cite{MingantiHuybrechts2022,BhattacharyaEtAl2022Arnoldi,BhattacharyaEtAl2023BiLanczos}.
In particular, we would like to extend the present correspondence to
generators of the Gorini--Kossakowski--Sudarshan--Lindblad form
~\cite{GoriniKossakowskiSudarshan1976,Lindblad1976}.
Other possible directions include broader Toda hierarchies
~\cite{KodamaYe1996,KodamaWilliams2015,UenoTakasaki1984,Takasaki2018Toda},
a Sato-theoretic formulation of the Arnoldi data
~\cite{AlexandrovZabrodin2013}, and discrete or ultradiscrete analogues of the
correspondence~\cite{Tsujimoto2002DiscreteToda,MarunoBiondini2004}.
Finally, the correspondence may provide a Toda-lattice framework for studying
state spreading in Krylov space and the associated Krylov complexity in
non-Hermitian dynamics
~\cite{LiuTangZhai2023,NandyEtAl2025Review,BhattacharyaEtAl2023BiLanczos,BhattacharyaEtAl2024Spread,NandyEtAl2025SVD}.

\section*{Acknowledgements}
We thank Y. Ito, K. Shimomura, K. Takahashi, K. Totsuka, and S. Tsujimoto for helpful discussions.
This work was inspired by a talk at the YITP International Molecule-type Workshop YITP-T-25-03, ``Hydrodynamics of low-dimensional interacting systems: Advances, challenges, and future directions,'' held at the Yukawa Institute for Theoretical Physics, Kyoto University.

\paragraph{Author contributions}
U.M. performed the research and wrote the manuscript.

\paragraph{Funding information}
U.M. was supported by JSPS KAKENHI Grant Number~JP25KJ1573.

\paragraph{Data and code availability}
No external datasets are associated with this theoretical work.
All material needed to reproduce the derivations and figures is contained in the manuscript and its source files.

\paragraph{Generative AI disclosure}
OpenAI Codex (GPT-5) was used to assist with English-language editing and LaTeX formatting.
The author critically reviewed and approved all changes and remains responsible for the content.


\appendix
\numberwithin{equation}{section}

\section{Review of the two-dimensional Toda lattice}
\label{app:toda-review}

This appendix summarizes the two-dimensional Toda lattice used in this paper
in terms of $\tau$ functions and Flaschka variables.
We first review the Toda lattice and then its two-dimensional extension.
In the two-dimensional Toda hierarchy, which includes the two-dimensional Toda lattice,
the standard zero-curvature condition is equivalent to a Lax equation.
We do not state either formulation explicitly here
~\cite{UenoTakasaki1984,Takasaki2018Toda}.

\subsection{Toda lattice}
\label{appsub:toda-lattice}

We consider a finite nonperiodic Toda lattice with $D$ sites.
Let $q_{n}\qty(t)$, $\qty(0\leq n\leq D-1)$, denote dynamical variables
that depend on time $t\in\R$~\cite{Toda1967,Flaschka1974Integrals,Moser1975}:
\begin{align}
  \dv[2]{}{t}q_{n}
  &=
  \eno{q_{n+1}-q_{n}}
  - \eno{q_{n}-q_{n-1}}.
  \label{eq:toda-newton}
\end{align}
At the lower boundary $n=0$, we set the second term on the right-hand side to zero;
at the upper boundary $n=D-1$, we set the first term to zero.
We suppress time arguments when convenient.

Many integrable systems can be recast in Hirota bilinear form by introducing
$\tau$ functions~\cite{Hirota2004}. These functions provide a starting point
for Sato theory and integrable hierarchies~\cite{UenoTakasaki1984,Takasaki2018Toda}.
The positive $\tau$ functions $\tau_{n}\qty(t)$ and the dynamical variables
$q_{n}\qty(t)$ are related for $0\leq n\leq D-1$ as follows
~\cite{KajiwaraMazzoccoOhta2007}:
\begin{align}
  q_{n}
  &=
  \ln\qty(\frac{\tau_{n+1}}{\tau_{n}}).
  \label{eq:toda-q-tau}
\end{align}
Here, $\tau_{0}\coloneqq1$, and we impose $\tau_{D+1}\coloneqq0$ as the
boundary condition for a finite chain.
The Hirota bilinear identity for the Toda lattice takes the form
~\cite{KajiwaraMazzoccoOhta2007}:
\begin{align}
  \tau_{n}\dv[2]{\tau_{n}}{t}
  - \qty(\dv{\tau_{n}}{t})^{2}
  &=
  \tau_{n+1}\tau_{n-1},
  \qquad
  1\leq n\leq D.
  \label{eq:toda-bilinear}
\end{align}
The logarithmic-derivative form of the same equation is
\begin{align}
  \dv[2]{}{t}\ln\qty(\tau_{n})
  &=
  \frac{\tau_{n+1}\tau_{n-1}}{\tau_{n}^{2}},
  \qquad
  1\leq n\leq D.
  \label{eq:toda-log-tau}
\end{align}
A Hankel-determinant solution of this bilinear identity is constructed from
$\tau_{1}\qty(t)$ as follows~\cite{KajiwaraMazzoccoOhta2007}:
\begin{align}
  \tau_{n}\qty(t)
  &=
  \det\qty[
    \dv[i+j]{}{t}\tau_{1}\qty(t)
  ]_{0 \leq i,j \leq n-1},
  \qquad
  1\leq n\leq D.
  \label{eq:toda-hankel}
\end{align}
For a finite chain, we use a solution for which the same determinant vanishes at $n=D+1$.
Equation~\eqref{eq:toda-hankel} satisfies the Hirota bilinear identity by the following
Desnanot--Jacobi identity for a general square matrix $\mathsf{M}$
~\cite{Krattenthaler1999} ($i \neq j$):
\begin{align}
  \det\qty(\mathsf{M})
  \det\qty(\mathsf{M}_{i,j}^{i,j})
  &=
  \det\qty(\mathsf{M}_{i}^{i})
  \det\qty(\mathsf{M}_{j}^{j})
  -
  \det\qty(\mathsf{M}_{i}^{j})
  \det\qty(\mathsf{M}_{j}^{i}).
  \label{eq:desnanot-jacobi}
\end{align}
Here, $\mathsf{M}_{i}^{j}$ denotes the submatrix obtained from $\mathsf{M}$ by
deleting row $i$ and column $j$, whereas $\mathsf{M}_{i,j}^{i,j}$ denotes the
submatrix obtained by deleting the rows indexed by $i,j$ and the columns indexed
by $i,j$.

The integrability of the Toda lattice is expressed by a Lax representation in
terms of the Flaschka variables. The variables $a_{n}\qty(t)\in\R$ for
$\qty(0\leq n\leq D-1)$ and $b_{n}\qty(t)>0$ for
$\qty(1\leq n\leq D-1)$ are related to the $\tau$ functions as follows
~\cite{Flaschka1974Integrals,KajiwaraMazzoccoOhta2007}:
\begin{align}
  a_{n}
  &=
  -\dv{}{t}\ln\qty(\frac{\tau_{n+1}}{\tau_{n}}),
  \label{eq:toda-flaschka-a}\\
  b_{n}^{2}
  &=
  \frac{\tau_{n+1}\tau_{n-1}}{\tau_{n}^{2}}.
  \label{eq:toda-flaschka-b}
\end{align}
Here, $b_{0}\coloneqq b_{D}\coloneqq0$.
The Flaschka variables obey~\cite{Flaschka1974Integrals,KajiwaraMazzoccoOhta2007}:
\begin{align}
  \dv{}{t}a_{n}
  &=
  b_{n}^{2}-b_{n+1}^{2},
  \label{eq:toda-flaschka-a-flow}\\
  \dv{}{t}b_{n}^{2}
  &=
  \qty(a_{n-1}-a_{n})b_{n}^{2}.
  \label{eq:toda-flaschka-b-flow}
\end{align}
The first equation holds for $0\leq n\leq D-1$, and the second for
$1\leq n\leq D-1$. These equations of motion admit a Lax representation in
terms of the matrices $\LLop\qty(t),\MM\qty(t)$
~\cite{Lax1968,Flaschka1974Integrals}:
\begin{align}
  \dv{}{t}\LLop
  &=
  \comm{\MM}{\LLop}.
  \label{eq:toda-lax}
\end{align}
Here, $\MM,\LLop$ are defined by
\begin{align}
  \MM
  &\coloneqq
  \frac{1}{2}
  \begin{pmatrix}
    0     & -b_{1} &          &              & 0 \\
    b_{1} & 0      & -b_{2}   &              &   \\
          & \ddots & \ddots   & \ddots       &   \\
          &        & b_{D-2}  & 0            & -b_{D-1} \\
    0     &        &          & b_{D-1}      & 0
  \end{pmatrix},
  \label{eq:toda-lax-m}\\
  \LLop
  &\coloneqq
  \begin{pmatrix}
    a_{0} & b_{1} &          &              & 0 \\
    b_{1} & a_{1} & b_{2}    &              &   \\
          & \ddots & \ddots   & \ddots       &   \\
          &        & b_{D-2}  & a_{D-2}      & b_{D-1} \\
    0     &        &          & b_{D-1}      & a_{D-1}
  \end{pmatrix}.
  \label{eq:toda-lax-l}
\end{align}
$\MM$ and $\LLop$ are $D\times D$ matrices.
Under the commutator flow generated by the real antisymmetric matrix $\MM$,
the eigenvalues of $\LLop$ are conserved and give independent conserved
quantities of the finite nonperiodic Toda lattice.
This finite-dimensional system is completely integrable
~\cite{Flaschka1974Integrals,Moser1975}.

\subsection{Two-dimensional Toda equation}
\label{appsub:phi-rep}

The two-dimensional Toda lattice is defined by the following partial differential
equation for dynamical variables $\phi_{n}\qty(z,\zb)\in\C$ with
$\qty(0\leq n\leq D-1)$ and $\qty(z,\zb)\in\C\times\C$
~\cite{HirotaItoKako1988,Takasaki2018Toda}:
\begin{align}
  \partial_{z}\partial_{\zb}\phi_{n}
  &=
  \eno{\phi_{n+1}-\phi_{n}}
  -
  \eno{\phi_{n}-\phi_{n-1}}.
  \label{eq:2d-toda-field}
\end{align}
At the lower boundary $n=0$, we set the second term on the right-hand side to zero;
at the upper boundary $n=D-1$, we set the first term to zero.
Compared with the Toda lattice, the second derivative with respect to time is
replaced by partial derivatives with respect to the two independent variables
$\qty(z,\zb)$.

For the Gram-determinant solution used in this paper, setting $\zb\coloneqq z^{*}$
gives $\tau_{n}\qty(z,\zb)\geq0$ by Eq.~\eqref{eq:taun-innerproduct}.
The condition $\tau_{n}\qty(z,\zb)>0$ is equivalent to
$\operatorname{rank}\mathsf{K}_{n}\qty(z)=n$ [see
Eq.~\eqref{eq:krylov-rank-tau-equivalence}].
We impose this complex-conjugation relation below and define logarithms and
Flaschka variables where the required $\tau$ functions are positive.

\subsection{\texorpdfstring{$\tau$}{tau} functions and moment/Gram determinants}
\label{appsub:tau-rep}

In a region where the required $\tau_{n}\qty(z,\zb)$ and
$\tau_{n+1}\qty(z,\zb)$ are positive, we define the real dynamical variable
$\phi_{n}\qty(z,\zb)\in\R$ as follows
~\cite{UenoTakasaki1984,Takasaki2018Toda}:
\begin{align}
  \phi_{n}
  &=
  \ln\qty(\frac{\tau_{n+1}}{\tau_{n}}).
  \label{eq:2d-toda-phi-tau}
\end{align}
Here, $\tau_{0}\coloneqq1$, and we impose $\tau_{D+1}\coloneqq0$ as the
boundary condition for a finite chain.
The Hirota bilinear identity for the two-dimensional Toda lattice takes the form
~\cite{UenoTakasaki1984,Takasaki2018Toda}
\begin{align}
  \tau_{n}\partial_{z}\partial_{\zb}\tau_{n}
  -
  \qty(\partial_{z}\tau_{n})
  \qty(\partial_{\zb}\tau_{n})
  &=
  \tau_{n+1}\tau_{n-1},
  \qquad
  1\leq n\leq D.
  \label{eq:2d-toda-bilinear}
\end{align}
Equivalently, its logarithmic-derivative form is
\begin{align}
  \partial_{z}\partial_{\zb}\ln\qty(\tau_{n})
  &=
  \frac{\tau_{n+1}\tau_{n-1}}{\tau_{n}^{2}},
  \qquad
  1\leq n\leq D.
  \label{eq:2d-toda-log-tau}
\end{align}

In the present construction, the moment determinant solution of this bilinear
identity is a Gram determinant constructed from $\tau_{1}\qty(z,\zb)$ as follows
~\cite{AdlerVanMoerbeke1997Moment}:
\begin{align}
  \tau_{n}\qty(z,\zb)
  &=
  \det\qty[
    \partial_{z}^{i}
    \partial_{\zb}^{j}
    \tau_{1}\qty(z,\zb)
  ]_{0 \leq i,j \leq n-1},
  \qquad
  1\leq n\leq D.
  \label{eq:2d-toda-hankel}
\end{align}
For a finite chain, we use a solution for which the same determinant vanishes at $n=D+1$.
Equation~\eqref{eq:2d-toda-hankel} satisfies the bilinear identity by the
Desnanot--Jacobi identity~\eqref{eq:desnanot-jacobi}.

\subsection{Flaschka variables}
\label{appsub:flaschka-rep}

The Flaschka variables $a_{n}\qty(z,\zb)\in\C$ for
$\qty(0\leq n\leq D-1)$ and $b_{n}\qty(z,\zb)>0$ for
$\qty(1\leq n\leq D-1)$ of the two-dimensional Toda lattice are defined in
terms of the $\tau$ functions as follows
~\cite{UenoTakasaki1984,Takasaki2018Toda}:
\begin{align}
  a_{n}
  &\coloneqq
  -\partial_{z}
  \ln\qty(\frac{\tau_{n+1}}{\tau_{n}}),
  \label{eq:2d-toda-flaschka-a}\\
  b_{n}^{2}
  &\coloneqq
  \frac{\tau_{n+1}\tau_{n-1}}{\tau_{n}^{2}}.
  \label{eq:2d-toda-flaschka-b}
\end{align}
Here, $b_{0}\coloneqq b_{D}\coloneqq0$.
These variables obey~\cite{UenoTakasaki1984,Takasaki2018Toda}:
\begin{align}
  \partial_{\zb}a_{n}
  &=
  b_{n}^{2}-b_{n+1}^{2},
  \label{eq:2d-toda-flaschka-a-flow}\\
  \partial_{z}b_{n}^{2}
  &=
  \qty(a_{n-1}-a_{n})b_{n}^{2}.
  \label{eq:2d-toda-flaschka-b-flow}
\end{align}
The first equation holds for $0\leq n\leq D-1$, and the second for
$1\leq n\leq D-1$. These equations of motion can also be written as a
standard zero-curvature condition.

\subsection{Toda reduction}
\label{appsub:toda-reduction}

To reduce the two-dimensional Toda lattice to the Toda lattice, define the real
time $t$ by
\begin{align}
  t
  &\coloneqq
  z+\zb.
  \label{eq:toda-reduction-time}
\end{align}
We consider the case in which $\phi_{n}$ and $\tau_{n}$ depend only on $t$ and
not on $z-\zb$~\cite{Takasaki2018Toda}.
Then $a_{n}$ and $b_{n}$ also depend only on $t$, and the partial derivatives
acting on these variables become
\begin{align}
  \partial_{z}
  &=
  \partial_{\zb}
  =
  \dv{}{t}.
  \label{eq:toda-reduction-derivatives}
\end{align}

Define the reduced dynamical variable $q_{n}$ by
\begin{align}
  q_{n}
  &=
  \phi_{n}.
  \label{eq:toda-reduction-coordinate}
\end{align}
The oscillator equation~\eqref{eq:2d-toda-field} in
Sec.~\ref{appsub:phi-rep} becomes
\begin{align}
  \dv[2]{q_{n}}{t}
  &=
  \eno{q_{n+1}-q_{n}}
  -
  \eno{q_{n}-q_{n-1}}.
  \label{eq:toda-reduction-newton}
\end{align}
This equation and its boundary conditions at both ends coincide with the Toda
equation~\eqref{eq:toda-newton} in Sec.~\ref{appsub:toda-lattice}.

\section{Arnoldi method}
\label{app:arnoldi}

Let $\HH \in \operatorname{End}\qty(\C^{N})$ be a square matrix that is not
necessarily Hermitian, and let $\ket{u_{0}} \in \C^{N}$ be a nonzero normalized
seed vector. The Arnoldi method uses Krylov subspaces to transform $\HH$ by a
unitary matrix $\Uarn$ into the upper Hessenberg form $\Aarn$
~\cite{Arnoldi1951,Saad1980}:
\begin{align}
  \Uarn^{\dagger}\HH\Uarn
  = \Aarn
  =
  \begin{pmatrix}
    \alpha_{0} & *          & *          & *          & \cdots \\
    \beta_{1}  & \alpha_{1} & *          & *          & \cdots \\
    0          & \beta_{2}  & \alpha_{2} & *          & \cdots \\
    0          & 0          & \beta_{3}  & \alpha_{3} & \ddots \\
    \vdots     & \vdots     & \vdots     & \ddots     & \ddots
  \end{pmatrix}.
  \label{eq:arnoldi-hessenberg-form}
\end{align}

If a matrix element is viewed as a transition amplitude from Krylov site $n$ to
site $m$, Eq.~\eqref{eq:arnoldi-hessenberg-form} permits only the nearest-neighbor
transition $n\to n+1$ in the direction away from the seed vector. In the opposite
direction, it permits long-range transitions $n\to m<n$.
For such upper Hessenberg Krylov-chain representations and Arnoldi constructions
in Floquet and dissipative open systems, see
Refs.~\cite{YatesMitra2021,MingantiHuybrechts2022,BhattacharyaEtAl2022Arnoldi,NizamiShrestha2023}.

We define the $n$th Krylov subspace generated by $\HH$ and $\ket{u_{0}}$ as
\begin{align}
  \K_{n}\qty(\HH,\ket{u_{0}})
  &\coloneqq
  \Span\qty(
    \ket{u_{0}},
    \HH\ket{u_{0}},
    \HH^{2}\ket{u_{0}},
    \cdots,
    \HH^{n-1}\ket{u_{0}}
  ).
  \label{eq:arnoldi-krylov-subspace}
\end{align}
In this appendix, we assume $\K_{N}\qty(\HH,\ket{u_{0}})=\C^{N}$; thus,
$\ket{u_{0}}$ is a cyclic vector of $\HH$. Under this assumption, we apply
Gram--Schmidt orthogonalization to the Krylov sequence and construct $\Uarn$ by
placing the resulting orthonormal basis vectors in its columns.
Related Krylov-subspace methods include the Lanczos and bi-Lanczos methods.
If $\HH$ is Hermitian, the Arnoldi reduction reduces to the Lanczos reduction,
and $\Aarn$ is a Hermitian tridiagonal matrix~\cite{Lanczos1950,Saad1980}.
The bi-Lanczos method tridiagonalizes a non-Hermitian matrix by a generally
nonunitary similarity transformation~\cite{Gutknecht1992,FreundGutknechtNachtigal1993}.
It is also used to describe dissipative open systems involving non-Hermitian
matrices~\cite{BhattacharyaEtAl2023BiLanczos}.
\subsection{Arnoldi iteration}
\label{appsub:arnoldi-pseudocode}

Consider $\HH \in \operatorname{End}\qty(\C^{N})$ and a nonzero seed vector
$\ket{\psi}$, and assume
\begin{align}
  \K_{N}\qty(\HH,\ket{\psi})
  &=
  \C^{N}.
  \label{eq:arnoldi-cyclic-vector-condition}
\end{align}
A vector $\ket{\psi}$ that satisfies this condition is called a cyclic vector
of $\HH$. The condition prevents Arnoldi breakdown for $n < N-1$ and yields a
complete unitary transformation~\cite{Arnoldi1951,Saad1980}.
The iteration based on Gram--Schmidt orthogonalization is as follows.

\noindent\textbf{Input:}
A square matrix $\HH \in \operatorname{End}\qty(\C^{N})$ and a nonzero seed
vector $\ket{\psi}$ satisfying $\K_{N}\qty(\HH,\ket{\psi}) = \C^{N}$.

\noindent\textbf{Initialization:}
$\ket{u_{0}}\gets\ket{\psi}/\norm{\ket{\psi}}$.

\noindent\textbf{Iteration:}
For $n = 0,1,\ldots,N-1$, perform the following steps.
\begin{enumerate}
  \item
  \label{item:arnoldi-projection-coefficients}
  For $m = 0,1,\ldots,n-1$, set
  \begin{align*}
    h_{m,n}
    &\gets
    \mel{u_{m}}{\HH}{u_{n}}.
  \end{align*}

  \item
  \label{item:arnoldi-diagonal-coefficient}
  Set the diagonal entry to
  \begin{align*}
    \alpha_{n}
    &\gets
    \mel{u_{n}}{\HH}{u_{n}}.
  \end{align*}

  \item
  \label{item:arnoldi-residual}
  Subtract the components along the known basis vectors:
  \begin{align*}
    \ket{r_{n+1}}
    &\gets
    \HH\ket{u_{n}}
    -\sum_{m = 0}^{n-1}h_{m,n}\ket{u_{m}}
    -\alpha_{n}\ket{u_{n}}.
  \end{align*}

  \item
  \label{item:arnoldi-normalization}
  If $n < N-1$, set
  \begin{align*}
    \beta_{n+1}
    &\gets
    \norm{\ket{r_{n+1}}},
    &
    \ket{u_{n+1}}
    &\gets
    \frac{\ket{r_{n+1}}}{\beta_{n+1}}.
  \end{align*}
\end{enumerate}

\noindent\textbf{Output:}
\begin{align*}
  \Uarn
  &=
  \qty(
    \ket{u_{0}},
    \ket{u_{1}},
    \ldots,
    \ket{u_{N-1}}
  ),
  \\
  \Aarn
  &=
  \Uarn^{\dagger}\HH\Uarn.
\end{align*}
Here, $\Uarn$ is unitary, and $\Aarn$ has upper Hessenberg form.

\subsection{Arnoldi relation and Hessenberg reduction}
\label{appsub:arnoldi-hessenberg-derivation}

Starting from a normalized seed vector $\ket{u_{0}}$, suppose that we have
obtained the orthonormal basis vectors $\ket{u_{0}},\ldots,\ket{u_{n}}$.
The next basis vector is obtained by orthogonalizing $\HH\ket{u_{n}}$ against
the known basis vectors. Define
\begin{align}
  h_{m,n}
  &\coloneqq
  \mel{u_{m}}{\HH}{u_{n}},
  \qquad
  0 \leq m \leq n-1,
  \label{eq:arnoldi-coefficient-h}\\
  \alpha_{n}
  &\coloneqq
  \mel{u_{n}}{\HH}{u_{n}}.
  \label{eq:arnoldi-coefficient-alpha}
\end{align}
We define the residual vector by
\begin{align}
  \ket{r_{n+1}}
  &\coloneqq
  \HH\ket{u_{n}}
  -\sum_{m = 0}^{n-1}h_{m,n}\ket{u_{m}}
  -\alpha_{n}\ket{u_{n}}.
  \label{eq:arnoldi-residual}
\end{align}
The definitions of the coefficients and the orthonormality of the basis imply
\begin{align}
  \braket{u_{m}}{r_{n+1}}
  &=
  0,
  \qquad
  0 \leq m \leq n.
  \label{eq:arnoldi-residual-orthogonality}
\end{align}
The assumption $\K_{N}\qty(\HH,\ket{u_{0}}) = \C^{N}$ ensures that
$\ket{r_{n+1}} \neq 0$ for $0 \leq n \leq N-2$. Therefore, define
\begin{align}
  \beta_{n+1}
  &\coloneqq
  \norm{\ket{r_{n+1}}},
  &
  \ket{u_{n+1}}
  &\coloneqq
  \frac{1}{\beta_{n+1}}\ket{r_{n+1}}.
  \label{eq:arnoldi-normalization}
\end{align}
Then $\ket{u_{n+1}}$ is orthogonal to the known basis vectors and has unit norm.
This definition gives $\beta_{n+1} > 0$ and simultaneously fixes the phase of
$\ket{u_{n+1}}$.

Rearranging the definition of the residual vector gives, for
$0 \leq n \leq N-2$,
\begin{align}
  \HH\ket{u_{n}}
  &=
  \sum_{m = 0}^{n-1}h_{m,n}\ket{u_{m}}
  +\alpha_{n}\ket{u_{n}}
  +\beta_{n+1}\ket{u_{n+1}}.
  \label{eq:arnoldi-recurrence}
\end{align}
For the last basis vector, $\ket{u_{0}},\ldots,\ket{u_{N-1}}$ span $\C^{N}$, so
\begin{align}
  \HH\ket{u_{N-1}}
  &=
  \sum_{m = 0}^{N-2}h_{m,N-1}\ket{u_{m}}
  +\alpha_{N-1}\ket{u_{N-1}}.
  \label{eq:arnoldi-last-column}
\end{align}
The resulting $N$ vectors form an orthonormal basis of $\C^{N}$. Hence,
\begin{align}
  \Uarn^{\dagger}\Uarn
  &=
  \Uarn\Uarn^{\dagger}
  =
  \Imat_{N},
  \label{eq:arnoldi-unitarity}
\end{align}
and $\Uarn$ is unitary. Writing these relations column by column gives
\begin{align}
  \HH\Uarn
  &=
  \Uarn\Aarn,
  &
  \Aarn
  &=
  \Uarn^{\dagger}\HH\Uarn.
  \label{eq:arnoldi-matrix-relation}
\end{align}
The only new basis vector that can appear in each column $\HH\ket{u_{n}}$ is
$\ket{u_{n+1}}$. Therefore, all entries of $\Aarn$ below the first subdiagonal vanish.
Thus, $\Aarn$ has upper Hessenberg form.

\subsection{Truncation and nonnormality}
\label{appsub:arnoldi-truncation}

Truncating the Krylov dimension leaves a boundary residual that leaks out of the
Krylov subspace. For a normal matrix, the distance from a Ritz value to the
spectrum of the original matrix is bounded above by the residual norm of the
corresponding Ritz vector~\cite{Saad1980}. For a nonnormal matrix, nonorthogonal
eigenvectors can produce a larger error for the same residual. If the nonnormal
matrix is diagonalizable, this bound is multiplied by the condition number of
its eigenvector matrix~\cite{BauerFike1960}. The error in time evolution is bounded
by a time integral of the boundary residual weighted by the norm of the evolution
operator~\cite{BotchevGrimmHochbruck2013,JaweckiAuzingerKoch2020}.
Strong nonnormality can cause transient amplification, requiring estimates based
on the numerical range or pseudospectrum~\cite{Trefethen1997,HochbruckLubich1997}.

\section{Counterdiabatic driving}
\label{app:sta-basics}

Shortcuts to adiabaticity (STA) are methods that reproduce in finite time the
adiabatic evolution generated by a Hamiltonian that would otherwise have to vary
sufficiently slowly. This section formulates counterdiabatic driving in
finite-dimensional systems. We then relate this construction to Lewis--Riesenfeld
invariants and Lax equations when the Hamiltonian has fixed eigenvalues, and
finally summarize what changes in non-Hermitian systems.
For a general formulation of STA, see Ref.~\cite{GueryOdelinEtAl2019}.
We denote time by $s$, which is sometimes called protocol time in STA.

\subsection{Hermitian counterdiabatic driving}
\label{appsub:sta-transitionless-target}
\label{appsub:sta-nonadiabatic-coupling}

Consider a Hermitian matrix $\Had\qty(s)$ that varies with time $s$.
We assume that its eigenvalues are nondegenerate over the time interval of
interest. The instantaneous eigenvalues $E_{n}\qty(s)$ and eigenstates
$\ket{n\qty(s)}$ satisfy
\begin{align}
  \Had\qty(s)
  \ket{n\qty(s)}
  &=
  E_{n}\qty(s)
  \ket{n\qty(s)},
  \notag\\
  \braket{m\qty(s)}{n\qty(s)}
  &=
  \delta_{m,n},
  \qquad
  \sum_{n}
  \ket{n\qty(s)}
  \bra{n\qty(s)}
  =
  \Imat_{D}.
  \label{eq:sta-instantaneous-eigenvalue}
\end{align}

First, consider the evolution generated by $\Had$ alone:
\begin{align}
  \ii\dv{}{s}\ket{\Psi\qty(s)}
  &=
  \Had\qty(s)\ket{\Psi\qty(s)}.
  \label{eq:sta-schrodinger}
\end{align}
Expand the wave function in the instantaneous eigenstates as
\begin{align}
  \ket{\Psi\qty(s)}
  &=
  \sum_{n}
  c_{n}\qty(s)\ket{n\qty(s)}.
  \label{eq:sta-instantaneous-expansion}
\end{align}
Substituting Eq.~\eqref{eq:sta-instantaneous-expansion} into
Eq.~\eqref{eq:sta-schrodinger} and multiplying from the left by
$\bra{m\qty(s)}$ gives
\begin{align}
  \ii\dot c_{m}\qty(s)
  &=
  E_{m}\qty(s)c_{m}\qty(s)
  -
  \ii
  \sum_{n}
  \braket{m\qty(s)}{\dot n\qty(s)}
  c_{n}\qty(s).
  \label{eq:sta-coefficient-evolution}
\end{align}
The components with $n\neq m$ in the second term on the right-hand side can
induce transitions between distinct instantaneous eigenstates.
Differentiating the eigenvalue equation gives, for $m\neq n$,
\begin{align}
  \braket{m\qty(s)}{\dot n\qty(s)}
  &=
  \frac{
    \mel{m\qty(s)}{\dot{\HH}_{\mathrm{ad}}\qty(s)}{n\qty(s)}
  }{
    E_{n}\qty(s)-E_{m}\qty(s)
  }.
  \label{eq:sta-nonadiabatic-coupling}
\end{align}
Thus, the right-hand side of Eq.~\eqref{eq:sta-nonadiabatic-coupling} is the
nonadiabatic coupling amplitude between the two instantaneous eigenstates.

Choose an auxiliary term $\Hcd\qty(s)$ that cancels this off-diagonal coupling
by requiring
\begin{align}
  \mel{m\qty(s)}{\Hcd\qty(s)}{n\qty(s)}
  &=
  \ii\braket{m\qty(s)}{\dot n\qty(s)},
  \qquad
  m\neq n.
  \label{eq:sta-counterdiabatic-offdiagonal}
\end{align}
Define the total Hamiltonian including this auxiliary term as
\begin{align}
  \HH_{\mathrm{tot}}\qty(s)
  &\coloneqq
  \Had\qty(s)
  +
  \Hcd\qty(s).
  \label{eq:sta-total-hamiltonian}
\end{align}
The full evolution is governed by
\begin{align}
  \ii\dv{}{s}\ket{\Psi\qty(s)}
  &=
  \HH_{\mathrm{tot}}\qty(s)\ket{\Psi\qty(s)}.
  \label{eq:sta-total-schrodinger}
\end{align}
A standard choice with zero diagonal entries in this eigenbasis is the following
~\cite{DemirplakRice2003,Berry2009Transitionless}:
\begin{align}
  \Hcd^{\qty(0)}\qty(s)
  &\coloneqq
  \ii
  \sum_{n}
  \qty[
    \ket{\dot n\qty(s)}
    \bra{n\qty(s)}
    -
    \braket{n\qty(s)}{\dot n\qty(s)}
    \ket{n\qty(s)}\bra{n\qty(s)}
  ].
  \label{eq:sta-counterdiabatic}
\end{align}
Equation~\eqref{eq:sta-counterdiabatic-offdiagonal} fixes only the off-diagonal
entries. Adding a Hermitian term that commutes with $\Had$ therefore leaves the
cancellation of transitions unchanged and modifies only the phase of each
eigenstate. For $\Hcd=\Hcd^{\qty(0)}$, Eq.~\eqref{eq:sta-total-schrodinger}
gives
\begin{align}
  \ii\dot c_{n}
  &=
  \qty(
    E_{n}
    -
    \ii\braket{n}{\dot n}
  )
  c_{n},
  \label{eq:sta-transitionless-coefficient}
\end{align}
so coefficients of distinct eigenstates do not couple. The exact solution for
the initial state $\ket{n\qty(0)}$ is therefore
\begin{align}
  \ket{\Psi_{n}^{\mathrm{ad}}\qty(s)}
  &=
  \eno{
    -\ii
    \int_{0}^{s}\dd{s'}\,
    E_{n}\qty(s')
    -
    \int_{0}^{s}\dd{s'}\,
    \braket{n\qty(s')}{\partial_{s'}n\qty(s')}
  }
  \ket{n\qty(s)}.
  \label{eq:sta-adiabatic-target}
\end{align}
The two integrals in the exponent give the dynamical and geometric phases,
respectively.

\subsection{Evolution with fixed eigenvalues, dynamical invariants, and Lax equations}
\label{appsub:sta-isospectral-lax}
\label{appsub:sta-counterdiabatic}

Define the orthogonal projector $\mathsf{P}_{n}$ onto each instantaneous
eigenstate and the generator $\mathsf{G}$ of the time dependence of the
eigenbasis by
\begin{align}
  \mathsf{P}_{n}\qty(s)
  &\coloneqq
  \ket{n\qty(s)}\bra{n\qty(s)}
  \label{eq:sta-spectral-projector}
  \\
  \mathsf{G}\qty(s)
  &\coloneqq
  \ii
  \sum_{n}
  \ket{\dot n\qty(s)}\bra{n\qty(s)}.
  \label{eq:sta-full-generator}
\end{align}
Differentiating the completeness relation gives
$\mathsf{G}^{\dagger}=\mathsf{G}$ and
\begin{align}
  \ii\dot{\mathsf{P}}_{n}
  &=
  \comm{\mathsf{G}}{\mathsf{P}_{n}}.
  \label{eq:sta-projector-transport}
\end{align}
Equation~\eqref{eq:sta-counterdiabatic} can also be written as
\begin{align}
  \Hcd^{\qty(0)}
  &=
  \mathsf{G}
  -
  \sum_{n}
  \mathsf{P}_{n}\mathsf{G}\mathsf{P}_{n}.
  \label{eq:sta-counterdiabatic-projector}
\end{align}
Thus, $\mathsf{G}$ and $\Hcd^{\qty(0)}$ differ only in their diagonal entries and
generate the same evolution of the eigenspaces. For $\Hcd=\mathsf{G}$, an
eigenstate initialized at $s=0$ evolves exactly as
\begin{align}
  \ket{\Psi_{n}^{\qty(G)}\qty(s)}
  &=
  \eno{
    -\ii
    \int_{0}^{s}\dd{s'}\,
    E_{n}\qty(s')
  }
  \ket{n\qty(s)}
  .
  \label{eq:sta-full-generator-state}
\end{align}

In general, the eigenvalues also depend on $s$.
Differentiating $\Had=\sum_{n}E_{n}\mathsf{P}_{n}$ and using
$\ii\dot{\mathsf{P}}_{n}=\comm{\mathsf{G}}{\mathsf{P}_{n}}$ gives
\begin{align}
  \ii\dot{\HH}_{\mathrm{ad}}
  &=
  \comm{\mathsf{G}}{\Had}
  +
  \ii
  \sum_{n}
  \dot E_{n}\mathsf{P}_{n}.
  \label{eq:sta-general-had-evolution}
\end{align}
The last term vanishes only when the eigenvalues are constant, giving
\begin{align}
  \ii\dot{\HH}_{\mathrm{ad}}
  &=
  \comm{\mathsf{G}}{\Had}
  =
  \comm{\Hcd^{\qty(0)}}{\Had}.
  \label{eq:sta-isospectral-lax}
\end{align}
Equation~\eqref{eq:sta-isospectral-lax} has the form of a Lax equation while
retaining the freedom to choose the diagonal entries of the counterdiabatic
term. This freedom remains because the diagonal term subtracted from $\mathsf{G}$
in Eq.~\eqref{eq:sta-counterdiabatic-projector} commutes with $\Had$ and does not
appear in the commutator.

A Lewis--Riesenfeld invariant $\mathsf{I}\qty(s)$ is defined as a Hermitian
matrix satisfying~\cite{LewisRiesenfeld1969}:
\begin{align}
  \ii\dot{\mathsf{I}}
  &=
  \comm{\HH_{\mathrm{tot}}}{\mathsf{I}}.
  \label{eq:sta-lr-invariant-definition}
\end{align}
For a matrix family with fixed eigenvalues, choose $\Hcd=\mathsf{G}$ and set
$\HH_{\mathrm{tot}}=\Had+\mathsf{G}$. Equation~\eqref{eq:sta-isospectral-lax}
then becomes
\begin{align}
  \ii\dot{\HH}_{\mathrm{ad}}
  &=
  \comm{\HH_{\mathrm{tot}}}{\Had}.
  \label{eq:sta-lr-invariant}
\end{align}
Thus, $\Had$ itself is a Lewis--Riesenfeld invariant of the total Hamiltonian in
this case.

Equation~\eqref{eq:sta-isospectral-lax} can also be written as
\begin{align}
  \dot{\HH}_{\mathrm{ad}}
  &=
  \comm{-\ii\mathsf{G}}{\Had},
  \label{eq:sta-standard-lax-equation}
\end{align}
so $\qty(\Had,-\ii\mathsf{G})$ forms a Lax pair.
Because $\mathsf{G}$ is Hermitian, $-\ii\mathsf{G}$ is anti-Hermitian.
Okuyama and Takahashi identified the Lewis--Riesenfeld invariant equation with
a Lax equation for a classical nonlinear integrable system and constructed a
counterdiabatic term from the corresponding Lax pair~\cite{OkuyamaTakahashi2016}.

\subsection{Non-Hermitian counterdiabatic driving}
\label{appsub:sta-nonhermitian-transport}

Thus far, we have considered Hermitian Hamiltonians.
For a non-Hermitian system, right and left eigenstates must be treated separately
~\cite{Brody2014,AshidaGongUeda2020}. We now allow $\Had\qty(s)$ to be
non-Hermitian. We assume that a complete set of right and left eigenstates exists
over the interval of interest and that all eigenvalues are nondegenerate.
We normalize the right and left eigenstates by
\begin{align}
  \Had\qty(s)
  \ket{r_{n}\qty(s)}
  &=
  E_{n}\qty(s)\ket{r_{n}\qty(s)},
  \notag\\
  \bra{l_{n}\qty(s)}
  \Had\qty(s)
  &=
  E_{n}\qty(s)\bra{l_{n}\qty(s)},
  \notag\\
  \braket{l_{m}\qty(s)}{r_{n}\qty(s)}
  &=
  \delta_{m,n}.
  \label{eq:sta-biorthogonal-eigenvectors}
\end{align}
The corresponding oblique projectors satisfy
\begin{align}
  \mathsf{P}_{n}\qty(s)
  &\coloneqq
  \ket{r_{n}\qty(s)}\bra{l_{n}\qty(s)},
  &
  \sum_{n}\mathsf{P}_{n}\qty(s)
  &=
  \Imat_{D}.
  \label{eq:sta-biorthogonal-projectors}
\end{align}
Unlike in a Hermitian system, $\mathsf{P}_{n}$ is generally not an orthogonal
projector, and $E_{n}$ can be complex. The state norm is generally not conserved,
and the biorthogonal normalization is not unique under complex rescaling.

Expanding a solution of $\ii\partial_{s}\ket{\Psi}=\Had\ket{\Psi}$ as
$\ket{\Psi}=\sum_{n}c_{n}\ket{r_{n}}$ gives
\begin{align}
  \ii\dot c_{m}
  &=
  E_{m}c_{m}
  -
  \ii
  \sum_{n}
  \braket{l_{m}}{\dot{r}_{n}}c_{n}.
  \label{eq:sta-nonhermitian-coefficient-evolution}
\end{align}
Thus, $\braket{l_{m}}{\dot{r}_{n}}$ with $m\neq n$ can induce transitions
between distinct instantaneous eigenstates.

We formally write the state corresponding to the adiabatic evolution under
$\Had$ and initialized at $s=0$ as
\begin{align}
  \ket{\Psi_{n}^{\mathrm{ad}}\qty(s)}
  &=
  \eno{
    -\ii
    \int_{0}^{s}\dd{s'}\,
    E_{n}\qty(s')
    -
    \int_{0}^{s}\dd{s'}\,
    \braket{l_{n}\qty(s')}{\partial_{s'}r_{n}\qty(s')}
  }
  \ket{r_{n}\qty(s)}.
  \label{eq:sta-nonhermitian-adiabatic-phase}
\end{align}
The standard auxiliary term that generates this state exactly in finite time is the following
~\cite{IbanezEtAl2011,IbanezEtAl2012Erratum}:
\begin{align}
  \Hcd^{\qty(0)}
  &\coloneqq
  \ii
  \sum_{n}
  \qty[
    \ket{\dot{r}_{n}}\bra{l_{n}}
    -
    \braket{l_{n}}{\dot{r}_{n}}
    \mathsf{P}_{n}
  ].
  \label{eq:sta-nonhermitian-counterdiabatic}
\end{align}
Indeed, for $m\neq n$,
\begin{align}
  \mel{l_{m}}{\Hcd^{\qty(0)}}{r_{n}}
  &=
  \ii\braket{l_{m}}{\dot{r}_{n}},
  \label{eq:sta-nonhermitian-offdiagonal}
\end{align}
and $\Had+\Hcd^{\qty(0)}$ has the state in
Eq.~\eqref{eq:sta-nonhermitian-adiabatic-phase} as an exact solution.
Because the right and left eigenstates are not adjoints of one another,
$\Hcd^{\qty(0)}$ is generally non-Hermitian. In the Hermitian limit, it reduces
to Eq.~\eqref{eq:sta-counterdiabatic}.

Equation~\eqref{eq:sta-nonhermitian-adiabatic-phase} corresponds to the adiabatic
approximation for evolution under $\Had$ alone, but it is an exact solution for
the total Hamiltonian including $\Hcd^{\qty(0)}$.

\section{QR proof of the Toda--Arnoldi correspondence}
\label{app:main-theorem-qr-proof}

We reprove the main theorem~\eqref{eq:toda-arnoldi-coefficients} from a QR
factorization and derive the Lax equations for the Arnoldi matrix.
See Refs.~\cite{Symes1982QR,DeiftLiTomei1989} for the classical relations
among the finite nonperiodic Toda lattice, the QR algorithm, and matrix
factorizations.
In the following, we assume that $\ket{\psi\qty(0)}$ is a cyclic vector of
$\HH$.
We suppress the $\qty(z,\zb)$ dependence except where needed.

\subsection{QR factorization and coefficient identification}
\label{appsub:qr-coefficients}

For $1\leq n\leq D$, define the matrix whose columns are the first $n$
Krylov vectors by
\begin{align}
  \mathsf{K}_{n}
  &\coloneqq
  \qty(
    \ket{\psi},
    \HH\ket{\psi},
    \ldots,
    \HH^{n-1}\ket{\psi}
  )
  \in\C^{D\times n}.
  \label{eq:qr-krylov-matrix}
\end{align}
The full matrix satisfies
$\mathsf{K}_{D}\qty(z)
=\eno{-z\HH}\mathsf{K}_{D}\qty(0)$.
The cyclic-vector assumption makes $\mathsf{K}_{D}\qty(0)$ invertible, and
$\eno{-z\HH}$ is also invertible.
Thus, $\mathsf{K}_{D}\qty(z)$ is invertible and has the unique QR
factorization
\begin{align}
  \mathsf{K}_{D}
  &=
  \Uarn\mathsf{R},
  &
  \Uarn
  &=
  \qty(
    \ket{u_{0}},\ldots,\ket{u_{D-1}}
  ),
  \notag\\
  \Uarn^{\dagger}\Uarn
  &=
  \Uarn\Uarn^{\dagger}
  =
  \Imat_{D},
  &
  r_{n}
  &\coloneqq
  \qty(\mathsf{R})_{nn}>0
  \quad\qty(0\leq n\leq D-1),
  \label{eq:qr-full-factorization}
\end{align}
where $\Imat_{D}$ is the $D$-dimensional identity matrix and $\mathsf{R}$ is
upper triangular.
With this convention, $\Uarn$ and $\mathsf{R}$ depend smoothly on the
parameters~\cite{DieciEirola1999}, and
$\mathsf{K}_{n}=\Uarn_{n}\mathsf{R}_{n}$.
Here, $\Uarn_{n}$ consists of the first $n$ columns of $\Uarn$, and
$\mathsf{R}_{n}$ is the upper-left $n\times n$ block of $\mathsf{R}$.
In particular, $\ket{u_{0}}$ coincides with the seed vector in
Eq.~\eqref{eq:toda-arnoldi-seed}.

By the uniqueness of the QR factorization with positive diagonal entries, the
columns of $\Uarn$ coincide with the Arnoldi basis used in the main text.
Therefore,
\begin{align}
  \Aarn
  &\coloneqq
  \Uarn^{\dagger}\HH\Uarn,
  \notag\\
  \qty(\Aarn)_{nn}
  &=\alpha_{n},
  &&0\leq n\leq D-1,
  \notag\\
  \qty(\Aarn)_{n+1,n}
  &=\beta_{n+1}
  =\frac{r_{n+1}}{r_{n}}>0,
  &&0\leq n\leq D-2.
  \label{eq:qr-diagonal-ratio}
\end{align}
For $1\leq n\leq D-1$, the inclusion
$\HH\operatorname{im}\qty(\mathsf{K}_{n})
\subseteq\operatorname{im}\qty(\mathsf{K}_{n+1})$ shows that $\Aarn$ has
upper Hessenberg form.
Applying $\HH$ to the $n$th column of the QR factorization and comparing the
$\ket{u_{n+1}}$ components gives $\beta_{n+1}=r_{n+1}/r_{n}$.

The moment matrix in Eq.~\eqref{eq:taun} is the transpose of
$\mathsf{K}_{n}^{\dagger}\mathsf{K}_{n}$.
Hence,
\begin{align}
  \tau_{n}
  &=
  \det\qty(\mathsf{K}_{n}^{\dagger}\mathsf{K}_{n})
  =
  \det\qty(\mathsf{R}_{n}^{\dagger}\mathsf{R}_{n})
  \notag\\
  &=
  \prod_{j=0}^{n-1}r_{j}^{2},
  &&1\leq n\leq D,
  \notag\\
  r_{n}^{2}
  &=
  \frac{\tau_{n+1}}{\tau_{n}},
  &&0\leq n\leq D-1,
  \label{eq:qr-tau-product}
\end{align}
where we used $\tau_{0}=1$ and the empty-product convention.
Equations~\eqref{eq:qr-diagonal-ratio} and \eqref{eq:qr-tau-product} give
$\beta_{n}^{2}=\tau_{n+1}\tau_{n-1}/\tau_{n}^{2}=b_{n}^{2}$ for
$1\leq n\leq D-1$.
Because $b_{n}>0$ and $\beta_{n}>0$, we obtain $b_{n}=\beta_{n}$.

Using the determinant derivative formula
$\partial_{z}\ln\qty(\det\qty(\mathsf{X}))
  =
  \operatorname{tr}\qty(
    \mathsf{X}^{-1}\partial_{z}\mathsf{X}
  )
$,
along with
$\partial_{z}\mathsf{K}_{n}=-\HH\mathsf{K}_{n}$,
$\partial_{z}\mathsf{K}_{n}^{\dagger}=0$,
$\mathsf{K}_{n}=\Uarn_{n}\mathsf{R}_{n}$, and the cyclicity of the trace, we
obtain
\begin{align}
  \partial_{z}\ln\qty(\tau_{n})
  &=
  \partial_{z}\ln\qty(
    \det\qty(\mathsf{K}_{n}^{\dagger}\mathsf{K}_{n})
  )
  =
  -\operatorname{tr}\qty(
    \qty(\mathsf{K}_{n}^{\dagger}\mathsf{K}_{n})^{-1}
    \mathsf{K}_{n}^{\dagger}\HH\mathsf{K}_{n}
  )
  \notag\\
  &=
  -\operatorname{tr}\qty(
    \qty(\mathsf{R}_{n}^{\dagger}\mathsf{R}_{n})^{-1}
    \mathsf{R}_{n}^{\dagger}
    \qty(\Uarn_{n}^{\dagger}\HH\Uarn_{n})
    \mathsf{R}_{n}
  )
  \notag\\
  &=
  -\operatorname{tr}\qty(
    \mathsf{R}_{n}^{-1}
    \qty(\Uarn_{n}^{\dagger}\HH\Uarn_{n})
    \mathsf{R}_{n}
  )
  \notag\\
  &=
  -\operatorname{tr}\qty(
    \Uarn_{n}^{\dagger}\HH\Uarn_{n}
  )
  =
  -\sum_{j=0}^{n-1}\alpha_{j},
  \qquad 1\leq n\leq D.
  \label{eq:qr-log-tau-alpha}
\end{align}
Taking the difference between consecutive values of $n$ and comparing the
result with Eq.~\eqref{eq:2d-toda-flaschka-a}, we obtain
\begin{align}
  a_{n}
  &=-\partial_{z}\ln\qty(\frac{\tau_{n+1}}{\tau_{n}})
  =\alpha_{n},
  &&0\leq n\leq D-1,
  \notag\\
  b_{n}
  &=\beta_{n},
  &&1\leq n\leq D-1.
  \label{eq:qr-coefficient-identification}
\end{align}
Equations~\eqref{eq:qr-tau-product} and
\eqref{eq:qr-log-tau-alpha} also give, for $0\leq n\leq D-1$,
\begin{align}
  \partial_{z}\ln\qty(r_{n})
  &=-\frac{\alpha_{n}}{2},
  &
  \partial_{\zb}\ln\qty(r_{n})
  &=-\frac{\bar{\alpha}_{n}}{2}.
  \label{eq:qr-r-derivatives}
\end{align}

\subsection{Differentiated QR factorization and Arnoldi Lax equation}
\label{appsub:qr-arnoldi-lax}

We define the frame connections by
\begin{align}
  \mathsf{F}_{z}
  &\coloneqq
  -\Uarn^{\dagger}\partial_{z}\Uarn,
  &
  \mathsf{F}_{\zb}
  &\coloneqq
  -\Uarn^{\dagger}\partial_{\zb}\Uarn.
  \label{eq:qr-frame-connections-definition}
\end{align}
Differentiating the full QR factorization
$\mathsf{K}_{D}=\Uarn\mathsf{R}$ and using
$\partial_{z}\mathsf{K}_{D}=-\HH\mathsf{K}_{D}$ and
$\partial_{\zb}\mathsf{K}_{D}=0$, we obtain
\begin{align}
  \partial_{z}\mathsf{K}_{D}
  &=
  \qty(\partial_{z}\Uarn)\mathsf{R}
  +
  \Uarn\partial_{z}\mathsf{R}
  =
  -\HH\Uarn\mathsf{R},
  \notag\\
  \partial_{\zb}\mathsf{K}_{D}
  &=
  \qty(\partial_{\zb}\Uarn)\mathsf{R}
  +
  \Uarn\partial_{\zb}\mathsf{R}
  =0.
  \label{eq:qr-expanded-differentiated-factorization}
\end{align}
Multiplying Eq.~\eqref{eq:qr-expanded-differentiated-factorization} by
$\Uarn^{\dagger}$ from the left and by $\mathsf{R}^{-1}$ from the right, and
then using Eq.~\eqref{eq:qr-frame-connections-definition} and
$\Aarn=\Uarn^{\dagger}\HH\Uarn$, gives
\begin{align}
  -\Aarn
  &=
  -\mathsf{F}_{z}
  +
  \qty(\partial_{z}\mathsf{R})\mathsf{R}^{-1},
  \notag\\
  0
  &=
  -\mathsf{F}_{\zb}
  +
  \qty(\partial_{\zb}\mathsf{R})\mathsf{R}^{-1}.
  \label{eq:qr-differentiated-factorization}
\end{align}
Both $\qty(\partial_{z}\mathsf{R})\mathsf{R}^{-1}$ and
$\qty(\partial_{\zb}\mathsf{R})\mathsf{R}^{-1}$ are upper triangular, while
$\mathsf{F}_{\zb}=-\mathsf{F}_{z}^{\dagger}$.
Comparing matrix elements using
Eqs.~\eqref{eq:qr-coefficient-identification} and
\eqref{eq:qr-r-derivatives} gives
\begin{align}
  \mathsf{F}_{z}
  &=
  \begin{pmatrix}
    \frac{1}{2}a_{0} & 0                   & 0                   & 0                   & \cdots \\
    b_{1}             & \frac{1}{2}a_{1} & 0                   & 0                   & \cdots \\
    0                 & b_{2}               & \frac{1}{2}a_{2} & 0                   & \cdots \\
    0                 & 0                   & b_{3}               & \frac{1}{2}a_{3} & \ddots \\
    \vdots            & \vdots              & \vdots              & \ddots              & \ddots
  \end{pmatrix},
  \notag\\
  \mathsf{F}_{\zb}
  &=
  -
  \begin{pmatrix}
    \frac{1}{2}\bar{a}_{0} & b_{1}                 & 0                       & 0                       & \cdots \\
    0                       & \frac{1}{2}\bar{a}_{1} & b_{2}                 & 0                       & \cdots \\
    0                       & 0                       & \frac{1}{2}\bar{a}_{2} & b_{3}                 & \cdots \\
    0                       & 0                       & 0                       & \frac{1}{2}\bar{a}_{3} & \ddots \\
    \vdots                & \vdots                & \vdots                  & \ddots                  & \ddots
  \end{pmatrix}.
  \label{eq:qr-frame-connections}
\end{align}
Because $\HH$ is independent of $z$ and $\zb$, differentiating
$\Aarn=\Uarn^{\dagger}\HH\Uarn$ gives
\begin{equation}
  \begin{aligned}
    \partial_{z}\Aarn
    &=
    \comm{\mathsf{F}_{z}}{\Aarn},
    \\
    \partial_{\zb}\Aarn
    &=
    \comm{\mathsf{F}_{\zb}}{\Aarn}.
  \end{aligned}
  \label{eq:qr-arnoldi-lax}
\end{equation}
Similarly, the relation
$\Aarn^{\dagger}=\Uarn^{\dagger}\HH^{\dagger}\Uarn$ and the independence of
$\HH^{\dagger}$ from $z$ and $\zb$ give
\begin{equation}
  \begin{aligned}
    \partial_{z}\Aarn^{\dagger}
    &=
    \comm{\mathsf{F}_{z}}{\Aarn^{\dagger}},
    \\
    \partial_{\zb}\Aarn^{\dagger}
    &=
    \comm{\mathsf{F}_{\zb}}{\Aarn^{\dagger}}.
  \end{aligned}
  \label{eq:arnoldi-adjoint-lax}
\end{equation}
Under the coefficient identification~\eqref{eq:qr-coefficient-identification},
the diagonal part of the second equation in
Eq.~\eqref{eq:qr-arnoldi-lax} coincides with
Eq.~\eqref{eq:2d-toda-flaschka-a-flow}.
The evolution equation for $b_{n}^{2}$ obtained from the subdiagonal part of
the first equation coincides with Eq.~\eqref{eq:2d-toda-flaschka-b-flow}.

\section{Grassmannian geometry of Krylov subspaces}
\label{app:grassmannian-geometry}

We define the images of holomorphic Krylov subspaces under the Pl\"ucker
embedding and the associated Fubini--Study metric.
We then show that the pullback metric equals $b_{n}^{2}$, the squared Arnoldi
coefficient.

\subsection{Projective space and the Pl\"ucker embedding}
\label{appsub:GrnN}

For $m\in\Z_{\geq0}$, the complex projective space is
\begin{align*}
  \CP^{m}
  &\coloneqq
  \qty(
    \C^{m+1}\setminus\{0\}
  )
  /\C^{\times}.
\end{align*}
We denote the point defined by a nonzero vector $\ket{\psi}$ by
$\qty[\psi]\coloneqq\Span\qty(\ket{\psi})$.
For a normalized representative $\braket{\psi}{\psi}=1$, the
Fubini--Study line element is
\begin{align*}
  \dd s_{\mathrm{FS}}^{2}
  &\coloneqq
  \bra{\dd\psi}
  \qty(
    \Imat_{m+1}
    -
    \ket{\psi}\bra{\psi}
  )
  \ket{\dd\psi},
\end{align*}
as given in Ref.~\cite{ProvostVallee1980}.
The projector $\Imat_{m+1}-\ket{\psi}\bra{\psi}$ removes the component
along the representative vector and retains only the variation in projective
space.

For $1\leq n\leq N$, let $\Gr\qty(n,\mathcal{H}_{N})$ be the
Grassmannian of $n$-dimensional subspaces of
$\mathcal{H}_{N}\simeq\C^{N}$.
Write a matrix whose columns form a basis of
$V_{n}\in\Gr\qty(n,\mathcal{H}_{N})$ as
\begin{align*}
  \Vmat
  &\coloneqq
  \qty(
    \ket{v_{1}},
    \ldots,
    \ket{v_{n}}
  ),
  &
  V_{n}
  &=
  \operatorname{im}\qty(\Vmat).
\end{align*}
A different basis is given by $\widetilde{\Vmat}=\Vmat\Gmat$ with
$\Gmat\in\mathrm{GL}\qty(n,\C)$.
The exterior-product vector
\begin{align*}
  \wket{\Psi_{n}}
  &\coloneqq
  \ket{v_{1}}
  \wedge\cdots\wedge
  \ket{v_{n}}
\end{align*}
transforms under this change of basis as
\begin{align*}
  \wket{\widetilde{\Psi}_{n}}
  &=
  \det\qty(\Gmat)
  \wket{\Psi_{n}}.
\end{align*}
Therefore, the map
$V_{n}\mapsto\Span\qty(\wket{\Psi_{n}})$ is independent of the choice of
basis.
It defines the Pl\"ucker embedding~\cite{GriffithsHarris1978}
\begin{align}
  \iota_{\mathsf{Pl}}:
  \Gr\qty(n,\mathcal{H}_{N})
  &\longrightarrow
  \mathbb{P}\qty(
    \bigwedge^{n}\mathcal{H}_{N}
  )
  \simeq
  \CP^{\binom{N}{n}-1},
  \notag\\
  V_{n}
  &\longmapsto
  \Span\qty(
    \wket{\Psi_{n}}
  ).
  \label{eq:plucker-embedding}
\end{align}

Let $\ket{e_{1}},\ldots,\ket{e_{N}}$ be an orthonormal basis of
$\mathcal{H}_{N}$.
We use the standard inner product for which the exterior-product basis is
also orthonormal.
Expanding the exterior state as
\begin{align*}
  \wket{\Psi_{n}}
  &=
  \sum_{1\leq i_{1}<\cdots<i_{n}\leq N}
  p_{i_{1}\cdots i_{n}}\,
  \ket{e_{i_{1}}}
  \wedge\cdots\wedge
  \ket{e_{i_{n}}},
\end{align*}
the coefficients
\begin{align*}
  p_{i_{1}\cdots i_{n}}
  &=
  \det\qty[
    \Vmat_{i_{\alpha},a}
  ]_{\substack{
    1\leq\alpha\leq n\\
    1\leq a\leq n
  }}
\end{align*}
are the Pl\"ucker coordinates.
The Cauchy--Binet formula gives
\begin{align}
  \wbraket{\Psi_{n}}{\Psi_{n}}
  &=
  \det\qty(
    \Vmat^{\dagger}\Vmat
  )
  =
  \sum_{1\leq i_{1}<\cdots<i_{n}\leq N}
  \abs{p_{i_{1}\cdots i_{n}}}^{2}.
  \label{eq:grassmann-wedge-norm}
\end{align}
The Desnanot--Jacobi identity used to derive the Hirota bilinear identity
from Eq.~\eqref{eq:taun} is an example of a Pl\"ucker relation among
minors~\cite{IshikawaWakayama2006}.

In this paper, we set $N=D$ and take $\Vmat$ to be the Krylov matrix
$\mathsf{K}_{n}\qty(z)$ in Eq.~\eqref{eq:qr-krylov-matrix}.
Define the region in which $\mathsf{K}_{n}\qty(z)$ has rank $n$ by
\begin{align*}
  \Sigma_{n}^{\circ}
  &\coloneqq
  \qty{
    z\in\Sigma
    \,\middle|\,
    \operatorname{rank}
    \qty(
      \mathsf{K}_{n}\qty(z)
    )
    =
    n
  }
  .
\end{align*}
In this region, define
\begin{align*}
  \Phi_{n}:
  \Sigma_{n}^{\circ}
  &\longrightarrow
  \Gr\qty(n,\mathcal{H}_{D}),
  &
  z
  &\longmapsto
  \operatorname{im}
  \qty(
    \mathsf{K}_{n}\qty(z)
  ).
\end{align*}
Each column of $\mathsf{K}_{n}\qty(z)$ depends holomorphically on $z$, so
$\Phi_{n}$ is holomorphic.
On the real section where $\zb$ is the complex conjugate of $z$,
Eqs.~\eqref{eq:taun} and \eqref{eq:grassmann-wedge-norm} give
\begin{align}
  \tau_{n}>0
  \quad\Longleftrightarrow\quad
  \operatorname{rank}
  \qty(
    \mathsf{K}_{n}
  )
  =
  n
  \quad\Longleftrightarrow\quad
  \wket{\Psi_{n}}
  \neq
  0.
  \label{eq:krylov-rank-tau-equivalence}
\end{align}
When $\tau_{n}=0$, all Pl\"ucker coordinates vanish, so $\Phi_{n}$ and the
normalized exterior state are undefined.
However, if removing a common holomorphic factor from all Pl\"ucker coordinates
yields a nonzero local representative, that representative extends $\Phi_{n}$.

\subsection{Fubini--Study metric and Arnoldi coefficients}
\label{appsub:Grmetric}

Define the metric induced on the Grassmannian by the Pl\"ucker embedding as
\begin{align}
  g_{\mathrm{Gr}}
  &\coloneqq
  \iota_{\mathsf{Pl}}^{*}
  g_{\mathrm{FS}}.
  \label{eq:grassmann-fubini-study-pullback}
\end{align}
For a basis matrix $\Vmat$ that depends holomorphically on a local complex
coordinate, the K\"ahler potential is
\begin{align*}
  K_{\mathrm{Gr}}
  &=
  \ln\qty[
    \det\qty(
      \Vmat^{\dagger}\Vmat
    )
  ].
\end{align*}
A holomorphic and invertible change of basis
$\Vmat\mapsto\Vmat\Gmat$ adds only
$\ln\qty(\abs{\det\qty(\Gmat)}^{2})$ to $K_{\mathrm{Gr}}$, so the metric is
independent of the choice of basis.

Let the orthogonal projector onto $V_{n}$ be
\begin{align}
  \Pmat_{V_{n}}
  &\coloneqq
  \Vmat
  \qty(
    \Vmat^{\dagger}\Vmat
  )^{-1}
  \Vmat^{\dagger}.
  \label{eq:grassmann-orthogonal-projector}
\end{align}
In terms of this projector, the line element can be written as
~\cite{MitscherlingAvdoshkinMoore2025}
\begin{align}
  \dd s_{\mathrm{Gr}}^{2}
  &=
  \operatorname{tr}\qty[
    \qty(\Vmat^{\dagger}\Vmat)^{-1}
    \dd\Vmat^{\dagger}
    \qty(\Imat_{N}-\Pmat_{V_{n}})
    \dd\Vmat
  ]
  =
  \frac{1}{2}
  \operatorname{tr}\qty[
    \qty(\dd\Pmat_{V_{n}})^{2}
  ].
  \label{eq:grassmann-fubini-study-metric}
\end{align}
An infinitesimal change of basis $\dd\Vmat=\Vmat\mathsf{A}$ that leaves the
subspace unchanged does not contribute to the line element because
$\qty(\Imat_{N}-\Pmat_{V_{n}})\Vmat=0$.
Thus, $g_{\mathrm{Gr}}$ measures the variation of the subspace $V_{n}$ itself.

Setting $\Vmat=\mathsf{K}_{n}$ gives
$K_{\mathrm{Gr}}=\ln\qty(\tau_{n})$.
Equation~\eqref{eq:2d-toda-log-tau} then gives
\begin{align*}
  \qty(
    \Phi_{n}^{*}g_{\mathrm{Gr}}
  )_{z\zb}
  &=
  \partial_{z}\partial_{\zb}
  \ln\qty(\tau_{n})
  =
  b_{n}^{2}.
\end{align*}
We next derive this equality directly from the Arnoldi basis.
We substitute the QR factorization
$\mathsf{K}_{n}=\Uarn_{n}\mathsf{R}_{n}$ into
Eq.~\eqref{eq:grassmann-fubini-study-metric}.
On $\Sigma_{n}^{\circ}$, $\mathsf{R}_{n}$ is invertible and
\begin{align*}
  \Pmat_{n}
  &\coloneqq
  \Uarn_{n}\Uarn_{n}^{\dagger}
\end{align*}
is the orthogonal projector onto $\K_{n}$.
Using also
\begin{align*}
  \partial_{z}\mathsf{K}_{n}
  &=
  -\HH\mathsf{K}_{n},
  &
  \partial_{\zb}\mathsf{K}_{n}^{\dagger}
  &=
  -\mathsf{K}_{n}^{\dagger}\HH^{\dagger},
\end{align*}
we obtain
\begin{align}
  \qty(\Phi_{n}^{*}g_{\mathrm{Gr}})_{z\zb}
  &=
  \operatorname{tr}\qty[
    \qty(\mathsf{K}_{n}^{\dagger}\mathsf{K}_{n})^{-1}
    \qty(\partial_{\zb}\mathsf{K}_{n}^{\dagger})
    \qty(\Imat_{D}-\Pmat_{n})
    \qty(\partial_{z}\mathsf{K}_{n})
  ]
  =
  \notag\\
  &\quad
  \norm{\qty(\Imat_{D}-\Pmat_{n})\HH\Uarn_{n}}_{\mathrm{F}}^{2}.
  \label{eq:krylov-grassmann-metric-projector}
\end{align}
Here, $\norm{\cdot}_{\mathrm{F}}$ denotes the Frobenius norm.

The Arnoldi relation and Eq.~\eqref{eq:toda-arnoldi-coefficients} give
\begin{align}
  \qty(
    \Imat_{D}
    -
    \Pmat_{n}
  )
  \HH\ket{u_{j}}
  &=
  0,
  &&0\leq j\leq n-2,
  \notag\\
  \qty(
    \Imat_{D}
    -
    \Pmat_{n}
  )
  \HH\ket{u_{n-1}}
  &=
  b_{n}\ket{u_{n}}.
  \label{eq:krylov-arnoldi-orthogonal-components}
\end{align}
Therefore, throughout the cyclic region and for $1\leq n\leq D-1$,
\begin{align}
  \qty(
    \Phi_{n}^{*}g_{\mathrm{Gr}}
  )_{z\zb}
  &=
  b_{n}^{2}.
  \label{eq:krylov-grassmann-metric-arnoldi}
\end{align}
The coefficient $b_{n}$ is the norm of the component of
$\HH\ket{u_{n-1}}$ orthogonal to $\K_{n}$.
Accordingly, $b_{n}^{2}$ measures the magnitude of the variation of
$\K_{n}$ along the $z$ direction in the Fubini--Study metric.
At the endpoint $n=D$, the Grassmannian
$\Gr\qty(D,\mathcal{H}_{D})$ consists of a single point, so $\Phi_{D}$ is a
constant map and $b_{D}^{2}=0$.
If Arnoldi breakdown occurs at a Krylov dimension $d_{\mathrm K}<D$, then
$\K_{d_{\mathrm K}}$ is invariant under the action of $\HH$.
Thus, $\Phi_{d_{\mathrm K}}$ is constant, and both the pullback metric and
$b_{d_{\mathrm K}}^{2}$ vanish.
For $n>d_{\mathrm K}$,
$\operatorname{rank}\qty(\mathsf{K}_{n})<n$, so $\Phi_{n}$ is undefined.

\section{Krylov electrostatics}
\label{app:krylov-electrostatics}

We reinterpret the two-dimensional Toda equation~\eqref{eq:2d-toda-field} as
in-plane electrostatics, with $\phi_{n}$ the potential on layer $n$ and
$a_{n}$ the in-plane complex electric field.
Regard each layer as a plane with a potential distribution, each link between
adjacent layers as a capacitor, and $b_{n}^{2}$ as its polarization.
The same equation then becomes a local form of Gauss's law in which the
polarization depends exponentially on the potential difference between layers.
Equation~\eqref{eq:krylov-exterior-stokes-theorem} in the main text is the
corresponding integral form obtained from Stokes' theorem. Applying it to a disk
gives a theorem for the circular mean of the cumulative potential.
\subsection{Potential and complex electric field}
\label{appsub:krylov-electrostatic-potential}

Following Eq.~\eqref{eq:2d-toda-phi-tau}, define the potential $\phi_{n}$ and
the complex electric field $E_{n}$ on layer $n$ by
\begin{align}
  \phi_{n}
  &=
  \ln\qty(\frac{\tau_{n+1}}{\tau_{n}}),
  \qquad
  0\leq n\leq D-1,
  \notag\\
  E_{n}
  &\coloneqq
  -\partial_{z}\phi_{n}
  =
  a_{n}.
  \label{eq:krylov-complex-electric-field}
\end{align}
Thus, the Toda Flaschka variable $a_{n}$ coincides with the in-plane complex
electric field on layer $n$.

\subsection{Link polarization and nonlinear Poisson equation}
\label{appsub:krylov-nonlinear-poisson}

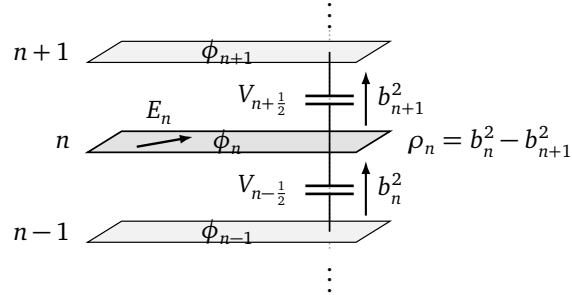
\begin{figure}[ht]
  \centering
  \begin{tikzpicture}[
    x=1cm,
    y=0.82cm,
    layer/.style={
      draw=black,
      fill=gray!10,
      line width=0.45pt
    },
    central layer/.style={
      draw=black,
      fill=gray!22,
      line width=0.65pt
    },
    link/.style={
      draw=black,
      line width=0.65pt
    },
    continuation/.style={
      draw=gray!75,
      densely dotted,
      line width=0.6pt
    },
    polarization/.style={
      -{Latex[length=1.8mm,width=1.2mm]},
      line width=0.8pt
    },
    field/.style={
      -{Latex[length=1.8mm,width=1.2mm]},
      line width=0.9pt
    },
    every node/.style={
      font=\small
    }
  ]
    \draw[continuation] (1.10,-0.72) -- (1.10,4.02);
    \node[fill=white,inner sep=0.5pt,text=black]
      at (1.10,-0.48) {$\vdots$};
    \node[fill=white,inner sep=0.5pt,text=black]
      at (1.10,3.78) {$\vdots$};

    \filldraw[layer]
      (-2.10,0)
      -- (1.45,0)
      -- (1.90,0.34)
      -- (-1.65,0.34)
      -- cycle;
    \filldraw[central layer]
      (-2.10,1.45)
      -- (1.45,1.45)
      -- (1.90,1.79)
      -- (-1.65,1.79)
      -- cycle;
    \filldraw[layer]
      (-2.10,2.90)
      -- (1.45,2.90)
      -- (1.90,3.24)
      -- (-1.65,3.24)
      -- cycle;

    \node[anchor=east] at (-2.20,0.17) {$n-1$};
    \node[anchor=east] at (-2.20,1.62) {$n$};
    \node[anchor=east] at (-2.20,3.07) {$n+1$};
    \node at (-0.25,0.17) {$\phi_{n-1}$};
    \node at (-0.25,1.62) {$\phi_{n}$};
    \node at (-0.25,3.07) {$\phi_{n+1}$};

    \draw[field]
      (-1.45,1.54)
      -- (-0.70,1.70);
    \node[anchor=south,fill=white,inner sep=0.5pt]
      at (-1.15,1.88) {$E_{n}$};

    \draw[link] (1.10,0.17) -- (1.10,0.78);
    \draw[line width=1pt] (0.78,0.78) -- (1.42,0.78);
    \draw[line width=1pt] (0.78,0.91) -- (1.42,0.91);
    \draw[link] (1.10,0.91) -- (1.10,1.62);
    \draw[polarization]
      (1.58,0.43)
      -- (1.58,1.32)
      node[midway,right] {$b_{n}^{2}$};
    \node[anchor=east] at (0.70,0.85)
      {$V_{n-\frac{1}{2}}$};

    \draw[link] (1.10,1.62) -- (1.10,2.23);
    \draw[line width=1pt] (0.78,2.23) -- (1.42,2.23);
    \draw[line width=1pt] (0.78,2.36) -- (1.42,2.36);
    \draw[link] (1.10,2.36) -- (1.10,3.07);
    \draw[polarization]
      (1.58,1.88)
      -- (1.58,2.77)
      node[midway,right] {$b_{n+1}^{2}$};
    \node[anchor=east] at (0.70,2.30)
      {$V_{n+\frac{1}{2}}$};

    \node[anchor=west] at (2.00,1.62)
      {$\rho_{n}=b_{n}^{2}-b_{n+1}^{2}$};

  \end{tikzpicture}
  \caption{
    Electrostatic analogy for the Toda lattice.
    Each layer is a copy of the parameter plane $\Sigma$ with potential
    $\phi_{m}$, and the arrow within layer $n$ denotes the in-plane complex
    electric field $E_{n}$.
    An arrow between layers denotes the link polarization
    $b_{m}^{2}=\exp\qty(V_{m-\frac{1}{2}})$ associated with the potential
    difference $V_{m-\frac{1}{2}}$.
    The polarization charge density on layer $n$ is
    $\rho_{n}=b_{n}^{2}-b_{n+1}^{2}$.
  }
  \label{fig:krylov-capacitor-network}
\end{figure}

Orient the link $n-\frac{1}{2}$ from layer $n-1$ to layer $n$, and define the
potential difference along this orientation by
\begin{align}
  V_{n-\frac{1}{2}}
  &\coloneqq
  \phi_{n}-\phi_{n-1},
  \notag\\
  b_{n}^{2}
  &=
  \exp\qty(V_{n-\frac{1}{2}}),
  \qquad
  1\leq n\leq D-1,
  \notag\\
  b_{0}^{2}
  &=
  b_{D}^{2}
  =
  0.
  \label{eq:krylov-link-constitutive-law}
\end{align}
For a finite chain, $b_{0}^{2}=b_{D}^{2}=0$ indicates that there are no
links beyond the boundaries.
The link between layers $n$ and $n+1$ is
$b_{n+1}^{2}=\exp\qty(V_{n+\frac{1}{2}})
=\exp\qty(\phi_{n+1}-\phi_{n})$.

Following the electrostatic analogy, we define the polarization charge density
on layer $n$ for $0\leq n\leq D-1$ by
\begin{align}
  \rho_{n}
  &\coloneqq
  -\partial_{z}\partial_{\zb}\phi_{n}
  =b_{n}^{2}-b_{n+1}^{2}.
  \label{eq:krylov-layer-charge-density}
\end{align}
At an interior layer $1\leq n\leq D-2$,
\begin{align}
  -\partial_{z}\partial_{\zb}\phi_{n}
  &=
  \exp\qty(V_{n-\frac{1}{2}})
  -
  \exp\qty(V_{n+\frac{1}{2}}).
  \label{eq:krylov-nonlinear-poisson-equation}
\end{align}
The quantity $\rho_{n}$ is the negative discrete divergence of the link
polarization $b_{n}^{2}$ and can have either sign.
The exponential constitutive relation in
Eq.~\eqref{eq:krylov-link-constitutive-law} corresponds to a nonlinear
dielectric response along the discrete Krylov direction.
Equation~\eqref{eq:krylov-nonlinear-poisson-equation} is therefore a Poisson
equation that combines the linear differential operator
$\partial_{z}\partial_{\zb}$ in the complex parameter plane with a nonlinear
constitutive relation in the discrete direction.
This layered structure is represented by the nonlinear capacitor chain shown in
Fig.~\ref{fig:krylov-capacitor-network}.

\subsection{Gauss's law and Stokes' theorem}
\label{appsub:krylov-gauss-law}

Equation~\eqref{eq:krylov-complex-electric-field} gives the local form of
Gauss's law,
\begin{align}
  \partial_{\zb}E_{n}
  &=
  \rho_{n}.
  \label{eq:krylov-complex-gauss-local}
\end{align}
In terms of the one-form $E_{n}\dd z$, this relation becomes
\begin{align}
  \dd\qty(E_{n}\dd z)
  &=
  -\rho_{n}\dd z\wedge\dd\zb.
  \label{eq:krylov-complex-gauss-one-form}
\end{align}
Let $\Omega\subset\Sigma$ be a compact region with a smooth boundary, and
give $\partial\Omega$ the positive orientation.
Applying Stokes' theorem to
Eq.~\eqref{eq:krylov-complex-gauss-one-form} gives
\begin{align}
  -\oint_{\partial\Omega}\dd z\,E_{n}
  &=
  \int_{\Omega}\dd z\wedge\dd\zb\,\rho_{n}.
  \label{eq:krylov-complex-gauss-law}
\end{align}
This is Gauss's law in complex one-form notation.
Taking the imaginary part of both sides relates the outward electric-field
flux across the boundary to the polarization charge in $\Omega$.

Summing the charge over all layers gives
\begin{align}
  \sum_{n=0}^{D-1}\rho_{n}
  &=
  \sum_{n=0}^{D-1}
  \qty(b_{n}^{2}-b_{n+1}^{2})
  =
  b_{0}^{2}-b_{D}^{2}
  =
  0,
  \label{eq:krylov-total-charge-neutrality}
\end{align}
so the total charge vanishes.
Each link polarization contributes with opposite signs to the two adjacent
layers, and the contributions cancel in the sum over layers.

The cumulative potential in
Eq.~\eqref{eq:krylov-fubini-study-kahler-potential} can be written as
\begin{align}
  \sum_{m=0}^{n-1}\phi_{m}
  &=
  \ln\qty(\tau_{n}),
  \notag\\
  -\partial_{z}\ln\qty(\tau_{n})
  &=
  \sum_{m=0}^{n-1}E_{m}
  =
  \sum_{m=0}^{n-1}a_{m},
  \notag\\
  \partial_{z}\partial_{\zb}\ln\qty(\tau_{n})
  &=
  b_{n}^{2},
  \notag\\
  \sum_{m=0}^{n-1}\rho_{m}
  &=
  -b_{n}^{2},
  \qquad
  1\leq n\leq D.
  \label{eq:krylov-cumulative-potential-field}
\end{align}
Since $E_{m}=a_{m}$, Eq.~\eqref{eq:krylov-exterior-gauss-law} is a complex
form of Gauss's law relating the cumulative electric field
$\sum_{m=0}^{n-1}E_{m}$ to the cumulative charge density $-b_{n}^{2}$ on the
first $n$ layers.
Its boundary term is
$\oint_{\partial\Omega}\dd z\,\sum_{m=0}^{n-1}E_{m}$, and this Gauss law
coincides with Stokes' formula~\eqref{eq:krylov-exterior-stokes-theorem}.

\subsection{Circular means of the cumulative potential}
\label{appsub:krylov-circular-mean}

Let $B_{r}$ be the disk of radius $r$ centered at $z_{0}$, and parametrize
its boundary by $z=z_{0}+r e^{i\theta}$.
For $1\leq n\leq D$, define the circular mean of the cumulative potential
$\sum_{m=0}^{n-1}\phi_{m}=\ln\qty(\tau_{n})$ in
Eq.~\eqref{eq:krylov-cumulative-potential-field} by
\begin{align}
  M_{n}\qty(r)
  &\coloneqq
  \frac{1}{2\pi}
  \int_{0}^{2\pi}\dd\theta\,
  \ln\qty(\tau_{n}\qty(z_{0}+r e^{i\theta},\zb_{0}+r e^{-i\theta})).
  \label{eq:krylov-circular-mean}
\end{align}
Applying Eq.~\eqref{eq:krylov-cumulative-potential-field} and Gauss's law for
the cumulative electric field to $B_{r}$ gives
\begin{align}
  \dv{M_{n}}{r}
  &=
  -\frac{1}{\pi r}
  \operatorname{Im}
  \oint_{\partial B_{r}}\dd z\,
  \sum_{m=0}^{n-1}E_{m}
  =
  \frac{\ii}{\pi r}
  \int_{B_{r}}\dd z\wedge\dd\zb\,b_{n}^{2}
  \geq
  0.
  \label{eq:krylov-circular-mean-gauss}
\end{align}
Thus, $M_{n}$ is nondecreasing in $r$, and the limit $r\to0$ gives
\begin{align}
  M_{n}\qty(r)
  &\geq
  \ln\qty(\tau_{n}\qty(z_{0},\zb_{0})),
  \qquad
  r>0.
  \label{eq:krylov-mean-value-property}
\end{align}
Equation~\eqref{eq:krylov-cumulative-potential-field} shows that
$\ln\qty(\tau_{n})$ is subharmonic, and
Eq.~\eqref{eq:krylov-mean-value-property} is its mean-value
inequality~\cite{Ransford1995}.


\nolinenumbers


\begin{thebibliography}{68}%
\makeatletter
\providecommand \@ifxundefined [1]{%
 \@ifx{#1\undefined}
}%
\providecommand \@ifnum [1]{%
 \ifnum #1\expandafter \@firstoftwo
 \else \expandafter \@secondoftwo
 \fi
}%
\providecommand \@ifx [1]{%
 \ifx #1\expandafter \@firstoftwo
 \else \expandafter \@secondoftwo
 \fi
}%
\providecommand \natexlab [1]{#1}%
\providecommand \enquote  [1]{``#1''}%
\providecommand \bibnamefont  [1]{#1}%
\providecommand \bibfnamefont [1]{#1}%
\providecommand \citenamefont [1]{#1}%
\providecommand \href@noop [0]{\@secondoftwo}%
\providecommand \href [0]{\begingroup \@sanitize@url \@href}%
\providecommand \@href[1]{\@@startlink{#1}\@@href}%
\providecommand \@@href[1]{\endgroup#1\@@endlink}%
\providecommand \@sanitize@url [0]{\catcode `\\12\catcode `\$12\catcode
  `\&12\catcode `\#12\catcode `\^12\catcode `\_12\catcode `\%12\relax}%
\providecommand \@@startlink[1]{}%
\providecommand \@@endlink[0]{}%
\providecommand \url  [0]{\begingroup\@sanitize@url \@url }%
\providecommand \@url [1]{\endgroup\@href {#1}{\urlprefix }}%
\providecommand \urlprefix  [0]{URL }%
\providecommand \Eprint [0]{\href }%
\providecommand \doibase [0]{https://doi.org/}%
\providecommand \selectlanguage [0]{\@gobble}%
\providecommand \bibinfo  [0]{\@secondoftwo}%
\providecommand \bibfield  [0]{\@secondoftwo}%
\providecommand \translation [1]{[#1]}%
\providecommand \BibitemOpen [0]{}%
\providecommand \bibitemStop [0]{}%
\providecommand \bibitemNoStop [0]{.\EOS\space}%
\providecommand \EOS [0]{\spacefactor3000\relax}%
\providecommand \BibitemShut  [1]{\csname bibitem#1\endcsname}%
\let\auto@bib@innerbib\@empty
\bibitem{HochbruckLubich1997}%
  \BibitemOpen
  \bibfield  {author} {\bibinfo {author} {\bibfnamefont {M.}~\bibnamefont
  {Hochbruck}}\ and\ \bibinfo {author} {\bibfnamefont {C.}~\bibnamefont
  {Lubich}},\ }\emph{\bibinfo {title} {On Krylov Subspace Approximations to the Matrix Exponential Operator}},\ \href {https://doi.org/10.1137/S0036142995280572} {\bibfield
  {journal} {\bibinfo  {journal} {SIAM Journal on Numerical Analysis}\ }\textbf
  {\bibinfo {volume} {34}},\ \bibinfo {pages} {1911} (\bibinfo {year}
  {1997})}\BibitemShut {NoStop}%
\bibitem{NandyEtAl2025Review}%
  \BibitemOpen
  \bibfield  {author} {\bibinfo {author} {\bibfnamefont {P.}~\bibnamefont
  {Nandy}}, \bibinfo {author} {\bibfnamefont {A.~S.}\ \bibnamefont
  {Matsoukas-Roubeas}}, \bibinfo {author} {\bibfnamefont {P.}~\bibnamefont
  {Mart{\'i}nez-Azcona}}, \bibinfo {author} {\bibfnamefont {A.}~\bibnamefont
  {Dymarsky}},\ and\ \bibinfo {author} {\bibfnamefont {A.}~\bibnamefont {del
  Campo}},\ }\emph{\bibinfo {title} {Quantum dynamics in Krylov space: Methods and applications}},\ \href {https://doi.org/10.1016/j.physrep.2025.05.001} {\bibfield
  {journal} {\bibinfo  {journal} {Physics Reports}\ }\textbf {\bibinfo {volume}
  {1125--1128}},\ \bibinfo {pages} {1} (\bibinfo {year} {2025})},\ \Eprint
  {https://arxiv.org/abs/2405.09628} {arXiv:2405.09628} \BibitemShut {NoStop}%
\bibitem{ParkerEtAl2019}%
  \BibitemOpen
  \bibfield  {author} {\bibinfo {author} {\bibfnamefont {D.~E.}\ \bibnamefont
  {Parker}}, \bibinfo {author} {\bibfnamefont {X.}~\bibnamefont {Cao}},
  \bibinfo {author} {\bibfnamefont {A.}~\bibnamefont {Avdoshkin}}, \bibinfo
  {author} {\bibfnamefont {T.}~\bibnamefont {Scaffidi}},\ and\ \bibinfo
  {author} {\bibfnamefont {E.}~\bibnamefont {Altman}},\ }\emph{\bibinfo {title} {A Universal Operator Growth Hypothesis}},\ \href
  {https://doi.org/10.1103/PhysRevX.9.041017} {\bibfield  {journal} {\bibinfo
  {journal} {Physical Review X}\ }\textbf {\bibinfo {volume} {9}},\ \bibinfo
  {pages} {041017} (\bibinfo {year} {2019})},\ \Eprint
  {https://arxiv.org/abs/1812.08657} {arXiv:1812.08657} \BibitemShut {NoStop}%
\bibitem{RabinoviciEtAl2021}%
  \BibitemOpen
  \bibfield  {author} {\bibinfo {author} {\bibfnamefont {E.}~\bibnamefont
  {Rabinovici}}, \bibinfo {author} {\bibfnamefont {A.}~\bibnamefont
  {S{\'a}nchez-Garrido}}, \bibinfo {author} {\bibfnamefont {R.}~\bibnamefont
  {Shir}},\ and\ \bibinfo {author} {\bibfnamefont {J.}~\bibnamefont {Sonner}},\
  }\emph{\bibinfo {title} {Operator complexity: a journey to the edge of Krylov space}},\ \href {https://doi.org/10.1007/JHEP06(2021)062} {\bibfield  {journal}
  {\bibinfo  {journal} {Journal of High Energy Physics}\ }\textbf {\bibinfo
  {volume} {2021}},\ \bibinfo {pages} {062} (\bibinfo {year}
  {2021})}\BibitemShut {NoStop}%
\bibitem{BhattacharjeeEtAl2022Saddle}%
  \BibitemOpen
  \bibfield  {author} {\bibinfo {author} {\bibfnamefont {B.}~\bibnamefont
  {Bhattacharjee}}, \bibinfo {author} {\bibfnamefont {X.}~\bibnamefont {Cao}},
  \bibinfo {author} {\bibfnamefont {P.}~\bibnamefont {Nandy}},\ and\ \bibinfo
  {author} {\bibfnamefont {T.}~\bibnamefont {Pathak}},\ }\emph{\bibinfo {title} {Krylov complexity in saddle-dominated scrambling}},\ \href
  {https://doi.org/10.1007/JHEP05(2022)174} {\bibfield  {journal} {\bibinfo
  {journal} {Journal of High Energy Physics}\ }\textbf {\bibinfo {volume}
  {2022}},\ \bibinfo {pages} {174} (\bibinfo {year} {2022})},\ \Eprint
  {https://arxiv.org/abs/2203.03534} {arXiv:2203.03534} \BibitemShut {NoStop}%
\bibitem{BhattacharjeeEtAl2023LargeQSYK}%
  \BibitemOpen
  \bibfield  {author} {\bibinfo {author} {\bibfnamefont {B.}~\bibnamefont
  {Bhattacharjee}}, \bibinfo {author} {\bibfnamefont {P.}~\bibnamefont
  {Nandy}},\ and\ \bibinfo {author} {\bibfnamefont {T.}~\bibnamefont
  {Pathak}},\ }\emph{\bibinfo {title} {Krylov complexity in large-{$q$} and double-scaled {SYK} model}},\ \href {https://doi.org/10.1007/JHEP08(2023)099} {\bibfield
  {journal} {\bibinfo  {journal} {Journal of High Energy Physics}\ }\textbf
  {\bibinfo {volume} {2023}},\ \bibinfo {pages} {099} (\bibinfo {year}
  {2023})},\ \Eprint {https://arxiv.org/abs/2210.02474}
  {arXiv:2210.02474} \BibitemShut {NoStop}%
\bibitem{TakahashiDelCampo2024}%
  \BibitemOpen
  \bibfield  {author} {\bibinfo {author} {\bibfnamefont {K.}~\bibnamefont
  {Takahashi}}\ and\ \bibinfo {author} {\bibfnamefont {A.}~\bibnamefont {del
  Campo}},\ }\emph{\bibinfo {title} {Shortcuts to Adiabaticity in Krylov Space}},\ \href {https://doi.org/10.1103/PhysRevX.14.011032} {\bibfield
  {journal} {\bibinfo  {journal} {Physical Review X}\ }\textbf {\bibinfo
  {volume} {14}},\ \bibinfo {pages} {011032} (\bibinfo {year} {2024})},\
  \Eprint {https://arxiv.org/abs/2302.05460} {arXiv:2302.05460} \BibitemShut
  {NoStop}%
\bibitem{TakahashiDelCampo2025}%
  \BibitemOpen
  \bibfield  {author} {\bibinfo {author} {\bibfnamefont {K.}~\bibnamefont
  {Takahashi}}\ and\ \bibinfo {author} {\bibfnamefont {A.}~\bibnamefont {del
  Campo}},\ }\emph{\bibinfo {title} {Krylov Subspace Methods for Quantum Dynamics with Time-Dependent Generators}},\ \href {https://doi.org/10.1103/PhysRevLett.134.030401} {\bibfield
  {journal} {\bibinfo  {journal} {Physical Review Letters}\ }\textbf {\bibinfo
  {volume} {134}},\ \bibinfo {pages} {030401} (\bibinfo {year} {2025})},\
  \Eprint {https://arxiv.org/abs/2408.08383} {arXiv:2408.08383} \BibitemShut
  {NoStop}%
\bibitem{Lanczos1950}%
  \BibitemOpen
  \bibfield  {author} {\bibinfo {author} {\bibfnamefont {C.}~\bibnamefont
  {Lanczos}},\ }\emph{\bibinfo {title} {An iteration method for the solution of the eigenvalue problem of linear differential and integral operators}},\ \href {https://doi.org/10.6028/jres.045.026} {\bibfield
  {journal} {\bibinfo  {journal} {Journal of Research of the National Bureau of
  Standards}\ }\textbf {\bibinfo {volume} {45}},\ \bibinfo {pages} {255}
  (\bibinfo {year} {1950})}\BibitemShut {NoStop}%
\bibitem{DymarskyGorsky2020}%
  \BibitemOpen
  \bibfield  {author} {\bibinfo {author} {\bibfnamefont {A.}~\bibnamefont
  {Dymarsky}}\ and\ \bibinfo {author} {\bibfnamefont {A.}~\bibnamefont
  {Gorsky}},\ }\emph{\bibinfo {title} {Quantum chaos as delocalization in Krylov space}},\ \href {https://doi.org/10.1103/PhysRevB.102.085137} {\bibfield
  {journal} {\bibinfo  {journal} {Physical Review B}\ }\textbf {\bibinfo
  {volume} {102}},\ \bibinfo {pages} {085137} (\bibinfo {year} {2020})},\
  \Eprint {https://arxiv.org/abs/1912.12227} {arXiv:1912.12227} \BibitemShut
  {NoStop}%
\bibitem{Arnoldi1951}%
  \BibitemOpen
  \bibfield  {author} {\bibinfo {author} {\bibfnamefont {W.~E.}\ \bibnamefont
  {Arnoldi}},\ }\emph{\bibinfo {title} {The principle of minimized iterations in the solution of the matrix eigenvalue problem}},\ \href {https://doi.org/10.1090/qam/42792} {\bibfield  {journal}
  {\bibinfo  {journal} {Quarterly of Applied Mathematics}\ }\textbf {\bibinfo
  {volume} {9}},\ \bibinfo {pages} {17} (\bibinfo {year} {1951})}\BibitemShut
  {NoStop}%
\bibitem{Saad1980}%
  \BibitemOpen
  \bibfield  {author} {\bibinfo {author} {\bibfnamefont {Y.}~\bibnamefont
  {Saad}},\ }\emph{\bibinfo {title} {Variations on Arnoldi's method for computing eigenelements of large unsymmetric matrices}},\ \href {https://doi.org/10.1016/0024-3795(80)90169-X} {\bibfield
  {journal} {\bibinfo  {journal} {Linear Algebra and its Applications}\
  }\textbf {\bibinfo {volume} {34}},\ \bibinfo {pages} {269} (\bibinfo {year}
  {1980})}\BibitemShut {NoStop}%
\bibitem{MingantiHuybrechts2022}%
  \BibitemOpen
  \bibfield  {author} {\bibinfo {author} {\bibfnamefont {F.}~\bibnamefont
  {Minganti}}\ and\ \bibinfo {author} {\bibfnamefont {D.}~\bibnamefont
  {Huybrechts}},\ }\emph{\bibinfo {title} {Arnoldi--Lindblad time evolution: Faster-than-the-clock algorithm for the spectrum of time-independent and Floquet open quantum systems}},\ \href {https://doi.org/10.22331/q-2022-02-10-649} {\bibfield
   {journal} {\bibinfo  {journal} {Quantum}\ }\textbf {\bibinfo {volume} {6}},\
  \bibinfo {pages} {649} (\bibinfo {year} {2022})}\BibitemShut {NoStop}%
\bibitem{BhattacharyaEtAl2022Arnoldi}%
  \BibitemOpen
  \bibfield  {author} {\bibinfo {author} {\bibfnamefont {A.}~\bibnamefont
  {Bhattacharya}}, \bibinfo {author} {\bibfnamefont {P.}~\bibnamefont {Nandy}},
  \bibinfo {author} {\bibfnamefont {P.~P.}\ \bibnamefont {Nath}},\ and\
  \bibinfo {author} {\bibfnamefont {H.}~\bibnamefont {Sahu}},\ }\emph{\bibinfo {title} {Operator growth and Krylov construction in dissipative open quantum systems}},\ \href
  {https://doi.org/10.1007/JHEP12(2022)081} {\bibfield  {journal} {\bibinfo
  {journal} {Journal of High Energy Physics}\ }\textbf {\bibinfo {volume}
  {2022}},\ \bibinfo {pages} {081} (\bibinfo {year} {2022})},\ \Eprint
  {https://arxiv.org/abs/2207.05347} {arXiv:2207.05347} \BibitemShut {NoStop}%
\bibitem{BhattacharyaEtAl2023BiLanczos}%
  \BibitemOpen
  \bibfield  {author} {\bibinfo {author} {\bibfnamefont {A.}~\bibnamefont
  {Bhattacharya}}, \bibinfo {author} {\bibfnamefont {P.}~\bibnamefont {Nandy}},
  \bibinfo {author} {\bibfnamefont {P.~P.}\ \bibnamefont {Nath}},\ and\
  \bibinfo {author} {\bibfnamefont {H.}~\bibnamefont {Sahu}},\ }\emph{\bibinfo {title} {On Krylov complexity in open systems: an approach via bi-Lanczos algorithm}},\ \href
  {https://doi.org/10.1007/JHEP12(2023)066} {\bibfield  {journal} {\bibinfo
  {journal} {Journal of High Energy Physics}\ }\textbf {\bibinfo {volume}
  {2023}},\ \bibinfo {pages} {066} (\bibinfo {year} {2023})},\ \Eprint
  {https://arxiv.org/abs/2303.04175} {arXiv:2303.04175} \BibitemShut {NoStop}%
\bibitem{BhattacharjeeEtAl2023DissipativeSYK}%
  \BibitemOpen
  \bibfield  {author} {\bibinfo {author} {\bibfnamefont {B.}~\bibnamefont
  {Bhattacharjee}}, \bibinfo {author} {\bibfnamefont {X.}~\bibnamefont {Cao}},
  \bibinfo {author} {\bibfnamefont {P.}~\bibnamefont {Nandy}},\ and\ \bibinfo
  {author} {\bibfnamefont {T.}~\bibnamefont {Pathak}},\ }\emph{\bibinfo {title} {Operator growth in open quantum systems: Lessons from the dissipative {SYK}}},\ \href
  {https://doi.org/10.1007/JHEP03(2023)054} {\bibfield  {journal} {\bibinfo
  {journal} {Journal of High Energy Physics}\ }\textbf {\bibinfo {volume}
  {2023}},\ \bibinfo {pages} {054} (\bibinfo {year} {2023})},\
  \Eprint {https://arxiv.org/abs/2212.06180} {arXiv:2212.06180} \BibitemShut
  {NoStop}%
\bibitem{BhattacharjeeEtAl2024LindbladianSYK}%
  \BibitemOpen
  \bibfield  {author} {\bibinfo {author} {\bibfnamefont {B.}~\bibnamefont
  {Bhattacharjee}}, \bibinfo {author} {\bibfnamefont {P.}~\bibnamefont
  {Nandy}},\ and\ \bibinfo {author} {\bibfnamefont {T.}~\bibnamefont
  {Pathak}},\ }\emph{\bibinfo {title} {Operator dynamics in Lindbladian {SYK}: A {Krylov} complexity perspective}},\ \href {https://doi.org/10.1007/JHEP01(2024)094} {\bibfield
  {journal} {\bibinfo  {journal} {Journal of High Energy Physics}\ }\textbf
  {\bibinfo {volume} {2024}},\ \bibinfo {pages} {094} (\bibinfo {year}
  {2024})},\ \Eprint {https://arxiv.org/abs/2311.00753} {arXiv:2311.00753}
  \BibitemShut {NoStop}%
\bibitem{UenoTakasaki1984}%
  \BibitemOpen
  \bibfield  {author} {\bibinfo {author} {\bibfnamefont {K.}~\bibnamefont
  {Ueno}}\ and\ \bibinfo {author} {\bibfnamefont {K.}~\bibnamefont
  {Takasaki}},\ }\emph{\bibinfo {title} {Toda Lattice Hierarchy}},\ in\ \href {https://doi.org/10.2969/aspm/00410001} {\emph
  {\bibinfo {booktitle} {Group Representations and Systems of Differential
  Equations}}},\ \bibinfo {series} {Advanced Studies in Pure Mathematics},
  Vol.~\bibinfo {volume} {4},\ \bibinfo {editor} {edited by\ \bibinfo {editor}
  {\bibfnamefont {K.}~\bibnamefont {Okamoto}}}\ (\bibinfo  {publisher}
  {North-Holland Publishing Co.},\ \bibinfo {year} {1984})\ pp.\ \bibinfo
  {pages} {1--95}\BibitemShut {NoStop}%
\bibitem{Takasaki2018Toda}%
  \BibitemOpen
  \bibfield  {author} {\bibinfo {author} {\bibfnamefont {K.}~\bibnamefont
  {Takasaki}},\ }\emph{\bibinfo {title} {Toda hierarchies and their applications}},\ \href {https://doi.org/10.1088/1751-8121/aabc14} {\bibfield
  {journal} {\bibinfo  {journal} {Journal of Physics A: Mathematical and
  Theoretical}\ }\textbf {\bibinfo {volume} {51}},\ \bibinfo {pages} {203001}
  (\bibinfo {year} {2018})},\ \Eprint {https://arxiv.org/abs/1801.09924}
  {arXiv:1801.09924 [math-ph]} \BibitemShut {NoStop}%
\bibitem{TakahashiNandyDelCampo2026}%
  \BibitemOpen
  \bibfield  {author} {\bibinfo {author} {\bibfnamefont {K.}~\bibnamefont
  {Takahashi}}, \bibinfo {author} {\bibfnamefont {P.}~\bibnamefont {Nandy}},\
  and\ \bibinfo {author} {\bibfnamefont {A.}~\bibnamefont {del Campo}},\ }\emph{\bibinfo {title} {Krylov complexity under Hamiltonian deformations and Toda flows}},\ \href
  {https://doi.org/10.1103/zt9g-scp5} {\bibfield  {journal} {\bibinfo
  {journal} {Physical Review B}\ }\textbf {\bibinfo {volume} {113}},\ \bibinfo
  {pages} {144312} (\bibinfo {year} {2026})},\ \Eprint
  {https://arxiv.org/abs/2510.19436} {arXiv:2510.19436} \BibitemShut {NoStop}%
\bibitem{Simon1983}%
  \BibitemOpen
  \bibfield  {author} {\bibinfo {author} {\bibfnamefont {B.}~\bibnamefont
  {Simon}},\ }\emph{\bibinfo {title} {Holonomy, the Quantum Adiabatic Theorem, and Berry's Phase}},\ \href {https://doi.org/10.1103/PhysRevLett.51.2167} {\bibfield
  {journal} {\bibinfo  {journal} {Physical Review Letters}\ }\textbf {\bibinfo
  {volume} {51}},\ \bibinfo {pages} {2167} (\bibinfo {year}
  {1983})}\BibitemShut {NoStop}%
\bibitem{Berry1984}%
  \BibitemOpen
  \bibfield  {author} {\bibinfo {author} {\bibfnamefont {M.~V.}\ \bibnamefont
  {Berry}},\ }\emph{\bibinfo {title} {Quantal phase factors accompanying adiabatic changes}},\ \href {https://doi.org/10.1098/rspa.1984.0023} {\bibfield
  {journal} {\bibinfo  {journal} {Proceedings of the Royal Society of London.
  A. Mathematical and Physical Sciences}\ }\textbf {\bibinfo {volume} {392}},\
  \bibinfo {pages} {45} (\bibinfo {year} {1984})}\BibitemShut {NoStop}%
\bibitem{ProvostVallee1980}%
  \BibitemOpen
  \bibfield  {author} {\bibinfo {author} {\bibfnamefont {J.-P.}\ \bibnamefont
  {Provost}}\ and\ \bibinfo {author} {\bibfnamefont {G.}~\bibnamefont
  {Vall{\'e}e}},\ }\emph{\bibinfo {title} {Riemannian structure on manifolds of quantum states}},\ \href {https://doi.org/10.1007/BF02193559} {\bibfield
  {journal} {\bibinfo  {journal} {Communications in Mathematical Physics}\
  }\textbf {\bibinfo {volume} {76}},\ \bibinfo {pages} {289} (\bibinfo {year}
  {1980})}\BibitemShut {NoStop}%
\bibitem{Bhattacharjee2023AGP}%
  \BibitemOpen
  \bibfield  {author} {\bibinfo {author} {\bibfnamefont {B.}~\bibnamefont
  {Bhattacharjee}},\ }\href@noop {} {\bibinfo {title} {A {Lanczos} approach to
  the adiabatic gauge potential}} (\bibinfo {year} {2023}),\ \Eprint
  {https://arxiv.org/abs/2302.07228} {arXiv:2302.07228 [quant-ph]} \BibitemShut
  {NoStop}%
\bibitem{ShresthaBhattacharjeeDelCampo2026}%
  \BibitemOpen
  \bibfield  {author} {\bibinfo {author} {\bibfnamefont {A.~W.}\ \bibnamefont
  {Shrestha}}, \bibinfo {author} {\bibfnamefont {B.}~\bibnamefont
  {Bhattacharjee}},\ and\ \bibinfo {author} {\bibfnamefont {A.}~\bibnamefont
  {del Campo}},\ }\href@noop {} {\bibinfo {title} {Shortcuts to adiabaticity
  for non-hermitian systems in {Krylov} space}} (\bibinfo {year} {2026}),\
  \bibinfo {note} {preprint; not peer reviewed as of 2026-07-25},\ \Eprint
  {https://arxiv.org/abs/2607.07802} {arXiv:2607.07802 [quant-ph]} \BibitemShut
  {NoStop}%
\bibitem{IbanezEtAl2011}%
  \BibitemOpen
  \bibfield  {author} {\bibinfo {author} {\bibfnamefont {S.}~\bibnamefont
  {Ib{\'a}{\~n}ez}}, \bibinfo {author} {\bibfnamefont {S.}~\bibnamefont
  {Mart{\'i}nez-Garaot}}, \bibinfo {author} {\bibfnamefont {X.}~\bibnamefont
  {Chen}}, \bibinfo {author} {\bibfnamefont {E.}~\bibnamefont {Torrontegui}},\
  and\ \bibinfo {author} {\bibfnamefont {J.~G.}\ \bibnamefont {Muga}},\ }\emph{\bibinfo {title} {Shortcuts to adiabaticity for non-Hermitian systems}},\ \href
  {https://doi.org/10.1103/PhysRevA.84.023415} {\bibfield  {journal} {\bibinfo
  {journal} {Physical Review A}\ }\textbf {\bibinfo {volume} {84}},\ \bibinfo
  {pages} {023415} (\bibinfo {year} {2011})},\ \Eprint
  {https://arxiv.org/abs/1106.2776} {arXiv:1106.2776} \BibitemShut {NoStop}%
\bibitem{IbanezEtAl2012Erratum}%
  \BibitemOpen
  \bibfield  {author} {\bibinfo {author} {\bibfnamefont {S.}~\bibnamefont
  {Ib{\'a}{\~n}ez}}, \bibinfo {author} {\bibfnamefont {S.}~\bibnamefont
  {Mart{\'i}nez-Garaot}}, \bibinfo {author} {\bibfnamefont {X.}~\bibnamefont
  {Chen}}, \bibinfo {author} {\bibfnamefont {E.}~\bibnamefont {Torrontegui}},\
  and\ \bibinfo {author} {\bibfnamefont {J.~G.}\ \bibnamefont {Muga}},\ }\emph{\bibinfo {title} {Erratum: Shortcuts to adiabaticity for non-Hermitian systems [Phys. Rev. A 84, 023415 (2011)]}},\ \href
  {https://doi.org/10.1103/PhysRevA.86.019901} {\bibfield  {journal} {\bibinfo
  {journal} {Physical Review A}\ }\textbf {\bibinfo {volume} {86}},\ \bibinfo
  {pages} {019901} (\bibinfo {year} {2012})}\BibitemShut {NoStop}%
\bibitem{OkuyamaTakahashi2016}%
  \BibitemOpen
  \bibfield  {author} {\bibinfo {author} {\bibfnamefont {M.}~\bibnamefont
  {Okuyama}}\ and\ \bibinfo {author} {\bibfnamefont {K.}~\bibnamefont
  {Takahashi}},\ }\emph{\bibinfo {title} {From Classical Nonlinear Integrable Systems to Quantum Shortcuts to Adiabaticity}},\ \href {https://doi.org/10.1103/PhysRevLett.117.070401}
  {\bibfield  {journal} {\bibinfo  {journal} {Physical Review Letters}\
  }\textbf {\bibinfo {volume} {117}},\ \bibinfo {pages} {070401} (\bibinfo
  {year} {2016})}\BibitemShut {NoStop}%
\bibitem{GoriniKossakowskiSudarshan1976}%
  \BibitemOpen
  \bibfield  {author} {\bibinfo {author} {\bibfnamefont {V.}~\bibnamefont
  {Gorini}}, \bibinfo {author} {\bibfnamefont {A.}~\bibnamefont
  {Kossakowski}},\ and\ \bibinfo {author} {\bibfnamefont {E.~C.~G.}\
  \bibnamefont {Sudarshan}},\ }\emph{\bibinfo {title} {Completely positive dynamical semigroups of $N$-level systems}},\ \href {https://doi.org/10.1063/1.522979}
  {\bibfield  {journal} {\bibinfo  {journal} {Journal of Mathematical Physics}\
  }\textbf {\bibinfo {volume} {17}},\ \bibinfo {pages} {821} (\bibinfo {year}
  {1976})}\BibitemShut {NoStop}%
\bibitem{Lindblad1976}%
  \BibitemOpen
  \bibfield  {author} {\bibinfo {author} {\bibfnamefont {G.}~\bibnamefont
  {Lindblad}},\ }\emph{\bibinfo {title} {On the generators of quantum dynamical semigroups}},\ \href {https://doi.org/10.1007/BF01608499} {\bibfield
  {journal} {\bibinfo  {journal} {Communications in Mathematical Physics}\
  }\textbf {\bibinfo {volume} {48}},\ \bibinfo {pages} {119} (\bibinfo {year}
  {1976})}\BibitemShut {NoStop}%
\bibitem{KodamaYe1996}%
  \BibitemOpen
  \bibfield  {author} {\bibinfo {author} {\bibfnamefont {Y.}~\bibnamefont
  {Kodama}}\ and\ \bibinfo {author} {\bibfnamefont {J.}~\bibnamefont {Ye}},\
  }\emph{\bibinfo {title} {Iso-spectral deformations of general matrix and their reductions on Lie algebras}},\ \href {https://doi.org/10.1007/BF02108824} {\bibfield  {journal} {\bibinfo
  {journal} {Communications in Mathematical Physics}\ }\textbf {\bibinfo
  {volume} {178}},\ \bibinfo {pages} {765} (\bibinfo {year} {1996})},\ \Eprint
  {https://arxiv.org/abs/solv-int/9506005} {arXiv:solv-int/9506005}
  \BibitemShut {NoStop}%
\bibitem{KodamaWilliams2015}%
  \BibitemOpen
  \bibfield  {author} {\bibinfo {author} {\bibfnamefont {Y.}~\bibnamefont
  {Kodama}}\ and\ \bibinfo {author} {\bibfnamefont {L.}~\bibnamefont
  {Williams}},\ }\emph{\bibinfo {title} {The Full Kostant--Toda Hierarchy on the Positive Flag Variety}},\ \href {https://doi.org/10.1007/s00220-014-2203-x} {\bibfield
  {journal} {\bibinfo  {journal} {Communications in Mathematical Physics}\
  }\textbf {\bibinfo {volume} {335}},\ \bibinfo {pages} {247} (\bibinfo {year}
  {2015})},\ \Eprint {https://arxiv.org/abs/1308.5011} {arXiv:1308.5011}
  \BibitemShut {NoStop}%
\bibitem{AlexandrovZabrodin2013}%
  \BibitemOpen
  \bibfield  {author} {\bibinfo {author} {\bibfnamefont {A.}~\bibnamefont
  {Alexandrov}}\ and\ \bibinfo {author} {\bibfnamefont {A.}~\bibnamefont
  {Zabrodin}},\ }\emph{\bibinfo {title} {Free fermions and tau-functions}},\ \href {https://doi.org/10.1016/j.geomphys.2013.01.007}
  {\bibfield  {journal} {\bibinfo  {journal} {Journal of Geometry and Physics}\
  }\textbf {\bibinfo {volume} {67}},\ \bibinfo {pages} {37} (\bibinfo {year}
  {2013})},\ \Eprint {https://arxiv.org/abs/1212.6049} {arXiv:1212.6049
  [math-ph]} \BibitemShut {NoStop}%
\bibitem{Tsujimoto2002DiscreteToda}%
  \BibitemOpen
  \bibfield  {author} {\bibinfo {author} {\bibfnamefont {S.}~\bibnamefont
  {Tsujimoto}},\ }\emph{\bibinfo {title} {On a Discrete Analogue of the Two-Dimensional Toda Lattice Hierarchy}},\ \href {https://doi.org/10.2977/prims/1145476418} {\bibfield
  {journal} {\bibinfo  {journal} {Publications of the Research Institute for
  Mathematical Sciences}\ }\textbf {\bibinfo {volume} {38}},\ \bibinfo {pages}
  {113} (\bibinfo {year} {2002})}\BibitemShut {NoStop}%
\bibitem{MarunoBiondini2004}%
  \BibitemOpen
  \bibfield  {author} {\bibinfo {author} {\bibfnamefont {K.-i.}\ \bibnamefont
  {Maruno}}\ and\ \bibinfo {author} {\bibfnamefont {G.}~\bibnamefont
  {Biondini}},\ }\emph{\bibinfo {title} {Resonance and web structure in discrete soliton systems: The two-dimensional Toda lattice and its fully discrete and ultra-discrete analogues}},\ \href {https://doi.org/10.1088/0305-4470/37/49/005} {\bibfield
   {journal} {\bibinfo  {journal} {Journal of Physics A: Mathematical and
  General}\ }\textbf {\bibinfo {volume} {37}},\ \bibinfo {pages} {11819}
  (\bibinfo {year} {2004})},\ \Eprint {https://arxiv.org/abs/nlin/0406059}
  {arXiv:nlin/0406059 [nlin.SI]} \BibitemShut {NoStop}%
\bibitem{LiuTangZhai2023}%
  \BibitemOpen
  \bibfield  {author} {\bibinfo {author} {\bibfnamefont {C.}~\bibnamefont
  {Liu}}, \bibinfo {author} {\bibfnamefont {H.}~\bibnamefont {Tang}},\ and\
  \bibinfo {author} {\bibfnamefont {H.}~\bibnamefont {Zhai}},\ }\emph{\bibinfo {title} {Krylov complexity in open quantum systems}},\ \href
  {https://doi.org/10.1103/PhysRevResearch.5.033085} {\bibfield  {journal}
  {\bibinfo  {journal} {Physical Review Research}\ }\textbf {\bibinfo {volume}
  {5}},\ \bibinfo {pages} {033085} (\bibinfo {year} {2023})},\ \Eprint
  {https://arxiv.org/abs/2207.13603} {arXiv:2207.13603} \BibitemShut {NoStop}%
\bibitem{BhattacharyaEtAl2024Spread}%
  \BibitemOpen
  \bibfield  {author} {\bibinfo {author} {\bibfnamefont {A.}~\bibnamefont
  {Bhattacharya}}, \bibinfo {author} {\bibfnamefont {R.~N.}\ \bibnamefont
  {Das}}, \bibinfo {author} {\bibfnamefont {B.}~\bibnamefont {Dey}},\ and\
  \bibinfo {author} {\bibfnamefont {J.}~\bibnamefont {Erdmenger}},\ }\emph{\bibinfo {title} {Spread complexity for measurement-induced non-unitary dynamics and Zeno effect}},\ \href
  {https://doi.org/10.1007/JHEP03(2024)179} {\bibfield  {journal} {\bibinfo
  {journal} {Journal of High Energy Physics}\ }\textbf {\bibinfo {volume}
  {2024}},\ \bibinfo {pages} {179} (\bibinfo {year} {2024})},\ \Eprint
  {https://arxiv.org/abs/2312.11635} {arXiv:2312.11635 [hep-th]} \BibitemShut
  {NoStop}%
\bibitem{NandyEtAl2025SVD}%
  \BibitemOpen
  \bibfield  {author} {\bibinfo {author} {\bibfnamefont {P.}~\bibnamefont
  {Nandy}}, \bibinfo {author} {\bibfnamefont {T.}~\bibnamefont {Pathak}},
  \bibinfo {author} {\bibfnamefont {Z.-Y.}\ \bibnamefont {Xian}},\ and\
  \bibinfo {author} {\bibfnamefont {J.}~\bibnamefont {Erdmenger}},\ }\emph{\bibinfo {title} {Krylov space approach to singular value decomposition in non-Hermitian systems}},\ \href
  {https://doi.org/10.1103/PhysRevB.111.064203} {\bibfield  {journal} {\bibinfo
   {journal} {Physical Review B}\ }\textbf {\bibinfo {volume} {111}},\ \bibinfo
  {pages} {064203} (\bibinfo {year} {2025})},\ \Eprint
  {https://arxiv.org/abs/2411.09309} {arXiv:2411.09309} \BibitemShut {NoStop}%
\bibitem{Toda1967}%
  \BibitemOpen
  \bibfield  {author} {\bibinfo {author} {\bibfnamefont {M.}~\bibnamefont
  {Toda}},\ }\emph{\bibinfo {title} {Vibration of a Chain with Nonlinear Interaction}},\ \href {https://doi.org/10.1143/JPSJ.22.431} {\bibfield  {journal}
  {\bibinfo  {journal} {Journal of the Physical Society of Japan}\ }\textbf
  {\bibinfo {volume} {22}},\ \bibinfo {pages} {431} (\bibinfo {year}
  {1967})}\BibitemShut {NoStop}%
\bibitem{Flaschka1974Integrals}%
  \BibitemOpen
  \bibfield  {author} {\bibinfo {author} {\bibfnamefont {H.}~\bibnamefont
  {Flaschka}},\ }\emph{\bibinfo {title} {The Toda lattice. II. Existence of integrals}},\ \href {https://doi.org/10.1103/PhysRevB.9.1924} {\bibfield
  {journal} {\bibinfo  {journal} {Physical Review B}\ }\textbf {\bibinfo
  {volume} {9}},\ \bibinfo {pages} {1924} (\bibinfo {year} {1974})}\BibitemShut
  {NoStop}%
\bibitem{Moser1975}%
  \BibitemOpen
  \bibfield  {author} {\bibinfo {author} {\bibfnamefont {J.}~\bibnamefont
  {Moser}},\ }\emph{\bibinfo {title} {Three integrable Hamiltonian systems connected with isospectral deformations}},\ \href {https://doi.org/10.1016/0001-8708(75)90151-6} {\bibfield
  {journal} {\bibinfo  {journal} {Advances in Mathematics}\ }\textbf {\bibinfo
  {volume} {16}},\ \bibinfo {pages} {197} (\bibinfo {year} {1975})}\BibitemShut
  {NoStop}%
\bibitem{Hirota2004}%
  \BibitemOpen
  \bibfield  {author} {\bibinfo {author} {\bibfnamefont {R.}~\bibnamefont
  {Hirota}},\ }\href {https://doi.org/10.1017/CBO9780511543043} {\emph
  {\bibinfo {title} {The Direct Method in Soliton Theory}}},\ \bibinfo {series}
  {Cambridge Tracts in Mathematics}, Vol.\ \bibinfo {volume} {155}\ (\bibinfo
  {publisher} {Cambridge University Press},\ \bibinfo {address} {Cambridge},\
  \bibinfo {year} {2004})\BibitemShut {NoStop}%
\bibitem{KajiwaraMazzoccoOhta2007}%
  \BibitemOpen
  \bibfield  {author} {\bibinfo {author} {\bibfnamefont {K.}~\bibnamefont
  {Kajiwara}}, \bibinfo {author} {\bibfnamefont {M.}~\bibnamefont {Mazzocco}},\
  and\ \bibinfo {author} {\bibfnamefont {Y.}~\bibnamefont {Ohta}},\ }\emph{\bibinfo {title} {A remark on the Hankel determinant formula for solutions of the Toda equation}},\ \href
  {https://doi.org/10.1088/1751-8113/40/42/S11} {\bibfield  {journal} {\bibinfo
   {journal} {Journal of Physics A: Mathematical and Theoretical}\ }\textbf
  {\bibinfo {volume} {40}},\ \bibinfo {pages} {12661} (\bibinfo {year}
  {2007})},\ \Eprint {https://arxiv.org/abs/nlin/0701029} {arXiv:nlin/0701029}
  \BibitemShut {NoStop}%
\bibitem{Krattenthaler1999}%
  \BibitemOpen
  \bibfield  {author} {\bibinfo {author} {\bibfnamefont {C.}~\bibnamefont
  {Krattenthaler}},\ }\emph{\bibinfo {title} {Advanced determinant calculus}},\ \href
  {https://www.mat.univie.ac.at/~slc/wpapers/s42kratt.html} {\bibfield
  {journal} {\bibinfo  {journal} {S{\'e}minaire Lotharingien de Combinatoire}\
  }\textbf {\bibinfo {volume} {42}},\ \bibinfo {pages} {B42q} (\bibinfo {year}
  {1999})},\ \Eprint {https://arxiv.org/abs/math/9902004} {arXiv:math/9902004}
  \BibitemShut {NoStop}%
\bibitem{Lax1968}%
  \BibitemOpen
  \bibfield  {author} {\bibinfo {author} {\bibfnamefont {P.~D.}\ \bibnamefont
  {Lax}},\ }\emph{\bibinfo {title} {Integrals of nonlinear equations of evolution and solitary waves}},\ \href {https://doi.org/10.1002/cpa.3160210503} {\bibfield
  {journal} {\bibinfo  {journal} {Communications on Pure and Applied
  Mathematics}\ }\textbf {\bibinfo {volume} {21}},\ \bibinfo {pages} {467}
  (\bibinfo {year} {1968})}\BibitemShut {NoStop}%
\bibitem{HirotaItoKako1988}%
  \BibitemOpen
  \bibfield  {author} {\bibinfo {author} {\bibfnamefont {R.}~\bibnamefont
  {Hirota}}, \bibinfo {author} {\bibfnamefont {M.}~\bibnamefont {Ito}},\ and\
  \bibinfo {author} {\bibfnamefont {F.}~\bibnamefont {Kako}},\ }\emph{\bibinfo {title} {Two-Dimensional Toda Lattice Equations}},\ \href
  {https://doi.org/10.1143/PTPS.94.42} {\bibfield  {journal} {\bibinfo
  {journal} {Progress of Theoretical Physics Supplement}\ }\textbf {\bibinfo
  {volume} {94}},\ \bibinfo {pages} {42} (\bibinfo {year} {1988})}\BibitemShut
  {NoStop}%
\bibitem{AdlerVanMoerbeke1997Moment}%
  \BibitemOpen
  \bibfield  {author} {\bibinfo {author} {\bibfnamefont {M.}~\bibnamefont
  {Adler}}\ and\ \bibinfo {author} {\bibfnamefont {P.}~\bibnamefont {van
  Moerbeke}},\ }\emph{\bibinfo {title} {Group factorization, moment matrices, and Toda lattices}},\ \href {https://doi.org/10.1155/S1073792897000378} {\bibfield
  {journal} {\bibinfo  {journal} {International Mathematics Research Notices}\
  }\textbf {\bibinfo {volume} {1997}},\ \bibinfo {pages} {555} (\bibinfo {year}
  {1997})}\BibitemShut {NoStop}%
\bibitem{YatesMitra2021}%
  \BibitemOpen
  \bibfield  {author} {\bibinfo {author} {\bibfnamefont {D.~J.}\ \bibnamefont
  {Yates}}\ and\ \bibinfo {author} {\bibfnamefont {A.}~\bibnamefont {Mitra}},\
  }\emph{\bibinfo {title} {Strong and almost strong modes of Floquet spin chains in Krylov subspaces}},\ \href {https://doi.org/10.1103/PhysRevB.104.195121} {\bibfield  {journal}
  {\bibinfo  {journal} {Physical Review B}\ }\textbf {\bibinfo {volume}
  {104}},\ \bibinfo {pages} {195121} (\bibinfo {year} {2021})},\ \Eprint
  {https://arxiv.org/abs/2105.13246} {arXiv:2105.13246} \BibitemShut {NoStop}%
\bibitem{NizamiShrestha2023}%
  \BibitemOpen
  \bibfield  {author} {\bibinfo {author} {\bibfnamefont {A.~A.}\ \bibnamefont
  {Nizami}}\ and\ \bibinfo {author} {\bibfnamefont {A.~W.}\ \bibnamefont
  {Shrestha}},\ }\emph{\bibinfo {title} {Krylov construction and complexity for driven quantum systems}},\ \href {https://doi.org/10.1103/PhysRevE.108.054222} {\bibfield
   {journal} {\bibinfo  {journal} {Physical Review E}\ }\textbf {\bibinfo
  {volume} {108}},\ \bibinfo {pages} {054222} (\bibinfo {year} {2023})},\
  \Eprint {https://arxiv.org/abs/2305.00256} {arXiv:2305.00256} \BibitemShut
  {NoStop}%
\bibitem{Gutknecht1992}%
  \BibitemOpen
  \bibfield  {author} {\bibinfo {author} {\bibfnamefont {M.~H.}\ \bibnamefont
  {Gutknecht}},\ }\emph{\bibinfo {title} {A Completed Theory of the Unsymmetric Lanczos Process and Related Algorithms, Part I}},\ \href {https://doi.org/10.1137/0613037} {\bibfield  {journal}
  {\bibinfo  {journal} {SIAM Journal on Matrix Analysis and Applications}\
  }\textbf {\bibinfo {volume} {13}},\ \bibinfo {pages} {594} (\bibinfo {year}
  {1992})}\BibitemShut {NoStop}%
\bibitem{FreundGutknechtNachtigal1993}%
  \BibitemOpen
  \bibfield  {author} {\bibinfo {author} {\bibfnamefont {R.~W.}\ \bibnamefont
  {Freund}}, \bibinfo {author} {\bibfnamefont {M.~H.}\ \bibnamefont
  {Gutknecht}},\ and\ \bibinfo {author} {\bibfnamefont {N.~M.}\ \bibnamefont
  {Nachtigal}},\ }\emph{\bibinfo {title} {An Implementation of the Look-Ahead Lanczos Algorithm for Non-Hermitian Matrices}},\ \href {https://doi.org/10.1137/0914009} {\bibfield  {journal}
  {\bibinfo  {journal} {SIAM Journal on Scientific Computing}\ }\textbf
  {\bibinfo {volume} {14}},\ \bibinfo {pages} {137} (\bibinfo {year}
  {1993})}\BibitemShut {NoStop}%
\bibitem{BauerFike1960}%
  \BibitemOpen
  \bibfield  {author} {\bibinfo {author} {\bibfnamefont {F.~L.}\ \bibnamefont
  {Bauer}}\ and\ \bibinfo {author} {\bibfnamefont {C.~T.}\ \bibnamefont
  {Fike}},\ }\emph{\bibinfo {title} {Norms and exclusion theorems}},\ \href {https://doi.org/10.1007/BF01386217} {\bibfield  {journal}
  {\bibinfo  {journal} {Numerische Mathematik}\ }\textbf {\bibinfo {volume}
  {2}},\ \bibinfo {pages} {137} (\bibinfo {year} {1960})}\BibitemShut {NoStop}%
\bibitem{BotchevGrimmHochbruck2013}%
  \BibitemOpen
  \bibfield  {author} {\bibinfo {author} {\bibfnamefont {M.~A.}\ \bibnamefont
  {Botchev}}, \bibinfo {author} {\bibfnamefont {V.}~\bibnamefont {Grimm}},\
  and\ \bibinfo {author} {\bibfnamefont {M.}~\bibnamefont {Hochbruck}},\ }\emph{\bibinfo {title} {Residual, Restarting, and Richardson Iteration for the Matrix Exponential}},\ \href
  {https://doi.org/10.1137/110820191} {\bibfield  {journal} {\bibinfo
  {journal} {SIAM Journal on Scientific Computing}\ }\textbf {\bibinfo {volume}
  {35}},\ \bibinfo {pages} {A1376} (\bibinfo {year} {2013})}\BibitemShut
  {NoStop}%
\bibitem{JaweckiAuzingerKoch2020}%
  \BibitemOpen
  \bibfield  {author} {\bibinfo {author} {\bibfnamefont {T.}~\bibnamefont
  {Jawecki}}, \bibinfo {author} {\bibfnamefont {W.}~\bibnamefont {Auzinger}},\
  and\ \bibinfo {author} {\bibfnamefont {O.}~\bibnamefont {Koch}},\ }\emph{\bibinfo {title} {Computable upper error bounds for Krylov approximations to matrix exponentials and associated $\varphi$-functions}},\ \href
  {https://doi.org/10.1007/s10543-019-00771-6} {\bibfield  {journal} {\bibinfo
  {journal} {BIT Numerical Mathematics}\ }\textbf {\bibinfo {volume} {60}},\
  \bibinfo {pages} {157} (\bibinfo {year} {2020})},\ \Eprint
  {https://arxiv.org/abs/1809.03369} {arXiv:1809.03369} \BibitemShut {NoStop}%
\bibitem{Trefethen1997}%
  \BibitemOpen
  \bibfield  {author} {\bibinfo {author} {\bibfnamefont {L.~N.}\ \bibnamefont
  {Trefethen}},\ }\emph{\bibinfo {title} {Pseudospectra of Linear Operators}},\ \href {https://doi.org/10.1137/S0036144595295284} {\bibfield
  {journal} {\bibinfo  {journal} {SIAM Review}\ }\textbf {\bibinfo {volume}
  {39}},\ \bibinfo {pages} {383} (\bibinfo {year} {1997})}\BibitemShut
  {NoStop}%
\bibitem{GueryOdelinEtAl2019}%
  \BibitemOpen
  \bibfield  {author} {\bibinfo {author} {\bibfnamefont {D.}~\bibnamefont
  {Gu{\'e}ry-Odelin}}, \bibinfo {author} {\bibfnamefont {A.}~\bibnamefont
  {Ruschhaupt}}, \bibinfo {author} {\bibfnamefont {A.}~\bibnamefont {Kiely}},
  \bibinfo {author} {\bibfnamefont {E.}~\bibnamefont {Torrontegui}}, \bibinfo
  {author} {\bibfnamefont {S.}~\bibnamefont {Mart{\'i}nez-Garaot}},\ and\
  \bibinfo {author} {\bibfnamefont {J.~G.}\ \bibnamefont {Muga}},\ }\emph{\bibinfo {title} {Shortcuts to adiabaticity: Concepts, methods, and applications}},\ \href
  {https://doi.org/10.1103/RevModPhys.91.045001} {\bibfield  {journal}
  {\bibinfo  {journal} {Reviews of Modern Physics}\ }\textbf {\bibinfo {volume}
  {91}},\ \bibinfo {pages} {045001} (\bibinfo {year} {2019})}\BibitemShut
  {NoStop}%
\bibitem{DemirplakRice2003}%
  \BibitemOpen
  \bibfield  {author} {\bibinfo {author} {\bibfnamefont {M.}~\bibnamefont
  {Demirplak}}\ and\ \bibinfo {author} {\bibfnamefont {S.~A.}\ \bibnamefont
  {Rice}},\ }\emph{\bibinfo {title} {Adiabatic Population Transfer with Control Fields}},\ \href {https://doi.org/10.1021/jp030708a} {\bibfield  {journal}
  {\bibinfo  {journal} {The Journal of Physical Chemistry A}\ }\textbf
  {\bibinfo {volume} {107}},\ \bibinfo {pages} {9937} (\bibinfo {year}
  {2003})}\BibitemShut {NoStop}%
\bibitem{Berry2009Transitionless}%
  \BibitemOpen
  \bibfield  {author} {\bibinfo {author} {\bibfnamefont {M.~V.}\ \bibnamefont
  {Berry}},\ }\emph{\bibinfo {title} {Transitionless quantum driving}},\ \href {https://doi.org/10.1088/1751-8113/42/36/365303} {\bibfield
   {journal} {\bibinfo  {journal} {Journal of Physics A: Mathematical and
  Theoretical}\ }\textbf {\bibinfo {volume} {42}},\ \bibinfo {pages} {365303}
  (\bibinfo {year} {2009})}\BibitemShut {NoStop}%
\bibitem{LewisRiesenfeld1969}%
  \BibitemOpen
  \bibfield  {author} {\bibinfo {author} {\bibfnamefont {H.~R.}\ \bibnamefont
  {Lewis}, \bibfnamefont {Jr.}}\ and\ \bibinfo {author} {\bibfnamefont {W.~B.}\
  \bibnamefont {Riesenfeld}},\ }\emph{\bibinfo {title} {An Exact Quantum Theory of the Time-Dependent Harmonic Oscillator and of a Charged Particle in a Time-Dependent Electromagnetic Field}},\ \href {https://doi.org/10.1063/1.1664991}
  {\bibfield  {journal} {\bibinfo  {journal} {Journal of Mathematical Physics}\
  }\textbf {\bibinfo {volume} {10}},\ \bibinfo {pages} {1458} (\bibinfo {year}
  {1969})}\BibitemShut {NoStop}%
\bibitem{Brody2014}%
  \BibitemOpen
  \bibfield  {author} {\bibinfo {author} {\bibfnamefont {D.~C.}\ \bibnamefont
  {Brody}},\ }\emph{\bibinfo {title} {Biorthogonal quantum mechanics}},\ \href {https://doi.org/10.1088/1751-8113/47/3/035305} {\bibfield
  {journal} {\bibinfo  {journal} {Journal of Physics A: Mathematical and
  Theoretical}\ }\textbf {\bibinfo {volume} {47}},\ \bibinfo {pages} {035305}
  (\bibinfo {year} {2014})}\BibitemShut {NoStop}%
\bibitem{AshidaGongUeda2020}%
  \BibitemOpen
  \bibfield  {author} {\bibinfo {author} {\bibfnamefont {Y.}~\bibnamefont
  {Ashida}}, \bibinfo {author} {\bibfnamefont {Z.}~\bibnamefont {Gong}},\ and\
  \bibinfo {author} {\bibfnamefont {M.}~\bibnamefont {Ueda}},\ }\emph{\bibinfo {title} {Non-Hermitian physics}},\ \href
  {https://doi.org/10.1080/00018732.2021.1876991} {\bibfield  {journal}
  {\bibinfo  {journal} {Advances in Physics}\ }\textbf {\bibinfo {volume}
  {69}},\ \bibinfo {pages} {249} (\bibinfo {year} {2020})}\BibitemShut
  {NoStop}%
\bibitem{Symes1982QR}%
  \BibitemOpen
  \bibfield  {author} {\bibinfo {author} {\bibfnamefont {W.~W.}\ \bibnamefont
  {Symes}},\ }\emph{\bibinfo {title} {The QR algorithm and scattering for the finite nonperiodic Toda lattice}},\ \href {https://doi.org/10.1016/0167-2789(82)90069-0} {\bibfield
  {journal} {\bibinfo  {journal} {Physica D: Nonlinear Phenomena}\ }\textbf
  {\bibinfo {volume} {4}},\ \bibinfo {pages} {275} (\bibinfo {year}
  {1982})}\BibitemShut {NoStop}%
\bibitem{DeiftLiTomei1989}%
  \BibitemOpen
  \bibfield  {author} {\bibinfo {author} {\bibfnamefont {P.}~\bibnamefont
  {Deift}}, \bibinfo {author} {\bibfnamefont {L.-C.}\ \bibnamefont {Li}},\ and\
  \bibinfo {author} {\bibfnamefont {C.}~\bibnamefont {Tomei}},\ }\emph{\bibinfo {title} {Matrix factorizations and integrable systems}},\ \href
  {https://doi.org/10.1002/cpa.3160420405} {\bibfield  {journal} {\bibinfo
  {journal} {Communications on Pure and Applied Mathematics}\ }\textbf
  {\bibinfo {volume} {42}},\ \bibinfo {pages} {443} (\bibinfo {year}
  {1989})}\BibitemShut {NoStop}%
\bibitem{DieciEirola1999}%
  \BibitemOpen
  \bibfield  {author} {\bibinfo {author} {\bibfnamefont {L.}~\bibnamefont
  {Dieci}}\ and\ \bibinfo {author} {\bibfnamefont {T.}~\bibnamefont {Eirola}},\
  }\emph{\bibinfo {title} {On Smooth Decompositions of Matrices}},\ \href {https://doi.org/10.1137/S0895479897330182} {\bibfield  {journal}
  {\bibinfo  {journal} {SIAM Journal on Matrix Analysis and Applications}\
  }\textbf {\bibinfo {volume} {20}},\ \bibinfo {pages} {800} (\bibinfo {year}
  {1999})}\BibitemShut {NoStop}%
\bibitem{GriffithsHarris1978}%
  \BibitemOpen
  \bibfield  {author} {\bibinfo {author} {\bibfnamefont {P.}~\bibnamefont
  {Griffiths}}\ and\ \bibinfo {author} {\bibfnamefont {J.}~\bibnamefont
  {Harris}},\ }\href {https://doi.org/10.1002/9781118032527} {\emph {\bibinfo
  {title} {Principles of Algebraic Geometry}}}\ (\bibinfo  {publisher} {John
  Wiley \& Sons},\ \bibinfo {address} {New York},\ \bibinfo {year}
  {1978})\BibitemShut {NoStop}%
\bibitem{IshikawaWakayama2006}%
  \BibitemOpen
  \bibfield  {author} {\bibinfo {author} {\bibfnamefont {M.}~\bibnamefont
  {Ishikawa}}\ and\ \bibinfo {author} {\bibfnamefont {M.}~\bibnamefont
  {Wakayama}},\ }\emph{\bibinfo {title} {Applications of minor summation formula III, Pl\"ucker relations, lattice paths and Pfaffian identities}},\ \href {https://doi.org/10.1016/j.jcta.2005.05.008} {\bibfield
  {journal} {\bibinfo  {journal} {Journal of Combinatorial Theory, Series A}\
  }\textbf {\bibinfo {volume} {113}},\ \bibinfo {pages} {113} (\bibinfo {year}
  {2006})},\ \Eprint {https://arxiv.org/abs/math/0312358} {arXiv:math/0312358
  [math.CO]} \BibitemShut {NoStop}%
\bibitem{MitscherlingAvdoshkinMoore2025}%
  \BibitemOpen
  \bibfield  {author} {\bibinfo {author} {\bibfnamefont {J.}~\bibnamefont
  {Mitscherling}}, \bibinfo {author} {\bibfnamefont {A.}~\bibnamefont
  {Avdoshkin}},\ and\ \bibinfo {author} {\bibfnamefont {J.~E.}\ \bibnamefont
  {Moore}},\ }\emph{\bibinfo {title} {Gauge-invariant projector calculus for quantum state geometry and applications to observables in crystals}},\ \href {https://doi.org/10.1103/qscv-qxqt} {\bibfield  {journal}
  {\bibinfo  {journal} {Physical Review B}\ }\textbf {\bibinfo {volume}
  {112}},\ \bibinfo {pages} {085104} (\bibinfo {year} {2025})},\ \Eprint
  {https://arxiv.org/abs/2412.03637} {arXiv:2412.03637} \BibitemShut {NoStop}%
\bibitem{Ransford1995}%
  \BibitemOpen
  \bibfield  {author} {\bibinfo {author} {\bibfnamefont {T.}~\bibnamefont
  {Ransford}},\ }\href {https://doi.org/10.1017/CBO9780511623776} {\emph
  {\bibinfo {title} {Potential Theory in the Complex Plane}}},\ \bibinfo
  {series} {London Mathematical Society Student Texts}, Vol.~\bibinfo {volume}
  {28}\ (\bibinfo  {publisher} {Cambridge University Press},\ \bibinfo
  {address} {Cambridge},\ \bibinfo {year} {1995})\BibitemShut {NoStop}%
\end{thebibliography}
\end{document}